\documentclass{amsart}
\usepackage{amsfonts}
\usepackage{amsmath,amscd,lscape}
\usepackage{amsthm}
\usepackage{amssymb}
\usepackage{latexsym}
\usepackage{mathrsfs}
\usepackage{tikz}

\usepackage[utf8]{inputenc}

\usepackage{graphicx}
\usepackage{subcaption}
\usepackage{float}

\newcommand{\nc}{\newcommand}
\nc{\ben}{\begin{eqnarray}}
\nc{\een}{\end{eqnarray}}

\newcommand{\beqa}{\begin{eqnarray}}
\newcommand{\eeqa}{\end{eqnarray}}
\nc{\Z}{{\bold Z}}

\newcommand{\fpt}[7]{{}_4\phi_3\left[\begin{matrix} #1 , #2, #3, #4 \\
#5, #6, #7 \end{matrix}\,; q^2,q^2\right]}

\usepackage[curve,matrix,arrow,color]{xy}

\usepackage{color}

\newtheorem{cor}{Corollary}[section]
\newtheorem{lem}{Lemma}[section]

\newtheorem{prop}{Proposition}[section]

\newtheorem{example}{Example}

\newtheorem{conj}{Conjecture}
\newtheorem{rem}{Remark}

\newcommand{\cal}{\mathcal}

\newcommand{\tA}{\textsf{A}}
\newcommand{\tB}{\textsf{B}}
\newcommand{\tC}{\tA^\diamond}
\newcommand{\tW}{\textsf{W}}

\newcommand{\tX}{\textsf{X}}
\newcommand{\tP}{\textsf{P}}
\newcommand{\tH}{\textsf{H}}

\newcommand{\bin}[2]{\left( \begin{matrix} #1  \\
#2 \end{matrix}\right)}

\newcommand{\cA}{\mathcal{A}}
\newcommand{\cB}{\mathcal{B}}
\newcommand{\cC}{\mathcal{C}}
\newcommand{\cD}{\mathcal{D}}

\newcommand{\bcV}{\bar {\cal V}}

\newcommand{\fa}{\mathfrak{a}}  
\newcommand{\fb}{\mathfrak{b}}  
\newcommand{\fc}{\mathfrak{c}}  
\newcommand{\fd}{\mathfrak{d}}

\newcommand{\non}{\nonumber}  
\newcommand{\q}{\quad}
\newcommand{\qq}{\qquad}

\numberwithin{equation}{section}
\begin{document}

\title[Diagonalization of the Heun-Askey-Wilson operator]{Diagonalization of the Heun-Askey-Wilson operator,\\  Leonard pairs and the algebraic Bethe ansatz}
\author{Pascal Baseilhac$^{*}$}
\address{$^*$ Institut Denis-Poisson CNRS/UMR 7013 - Universit\'e de Tours - Universit\'e d'Orl\'eans
Parc de Grammont, 37200 Tours, 
FRANCE}
\email{pascal.baseilhac@idpoisson.fr}

\author{Rodrigo A. Pimenta$^{*,**}$}
\address{$^{**}$ Instituto de F\'isica de S\~ao Carlos, Universidade de S\~ao Paulo, Caixa Postal 369, 13.560-590, S\~ao Carlos,
SP, BRAZIL} 
\address{CAPES Foundation, Ministry of Education of
Brazil, Brasilia - DF, Zip code 70.040-020, BRAZIL}
\email{pimenta@ifsc.usp.br}

\begin{abstract} An operator of Heun-Askey-Wilson type  is diagonalized within the framework of the algebraic Bethe ansatz using the theory of Leonard pairs. For different specializations and the generic case, the corresponding eigenstates are constructed in the form of Bethe states, whose Bethe roots satisfy Bethe ansatz equations of homogeneous or inhomogenous type.  For each set of Bethe equations, an alternative presentation is given in terms of `symmetrized' Bethe roots.   Also, two families of on-shell Bethe states are shown to generate two explicit bases on which a Leonard pair acts in a tridiagonal fashion. In a second part, the (in)homogeneous Baxter T-Q relations are derived. Certain realizations of the Heun-Askey-Wilson operator as second q-difference operators are introduced. Acting on the Q-polynomials, they produce the T-Q relations. For a special case, the Q-polynomial is identified with the Askey-Wilson polynomial, which allows one to obtain the solution of the associated Bethe ansatz equations. The analysis presented can be viewed as a toy model for studying integrable models generated from the Askey-Wilson algebra and its generalizations. As examples, the q-analog of the quantum Euler top and various types of three-sites Heisenberg spin chains in a magnetic field with inhomogeneous couplings, three-body and  boundary interactions are solved. Numerical examples are given. The results also apply to the time-band limiting problem in signal processing.
\end{abstract}

\maketitle

\vskip -0.5cm

{\small MSc:\ 81R50;\ 81R10;\ 81U15.}


\vspace{3mm}


\section{Introduction}

Although not pointed out in standard textbooks of quantum mechanics, the quantum harmonic oscillator is among the basic examples of quantum integrable systems generated from the  so-called  Askey-Wilson algebra introduced in \cite{Z91} (see (\ref{aw1}), (\ref{aw2}) below). Introducing the Heisenberg algebra with generators $a,a^\dagger$ and defining relations $[a,a^\dagger]=1$,
for an appropriate change of variables the Hamiltonian, position and momentum operators read, respectively:
\beqa
\tH = a^\dagger a +\frac{1}{2} \ ,\quad     \tX = \frac{1}{\sqrt{2}}(a^\dagger+a)\ ,\quad \tP = \frac{i}{\sqrt{2}}(a^\dagger-a)\  .\label{HXP}
\eeqa 
The triplet  $(\tH,\tX,i\tP)$ generates a specialization of the Askey-Wilson algebra. Considering the first presentation of the Askey-Wilson algebra with three generators given in \cite{Z91},  one routinely  finds:
$\big[\tH,\tX\big]= -i\tP$, $\big[\tX,i\tP\big]=-1$ and $\big[i\tP,\tH\big]= \tX$. It follows that the pair $(\tH,\tX)$ satisfies the so-called Askey-Wilson relations \cite{Z91}:
\beqa
\big[\tX,\big[\tX,\tH\big]\big] &=& -1\ ,\quad \big[\tH,\big[\tH,\tX\big]\big]= \tX\ \label{aw01}
\eeqa
giving a second presentation  \cite{Z91}.
Alternatively, for the pair $(\tH,i\tP)$ one also gets Askey-Wilson relations, with different structure constants:
\beqa
\big[\tH,\big[\tH,i\tP\big]\big] &=& i\tP\ ,\quad \big[i\tP,\big[i\tP,\tH\big]\big]= 1\ .\label{awp01}
\eeqa

In the literature, the analysis of the spectral problem for the Hamiltonian $\tH$ is usually based on two different representations\footnote{We assume the reader is familiar with the bases $\{|x\rangle\}$ and $\{|n\rangle\}$ of standard textbooks in quantum mechanics. For convenience, below we introduce $\langle x | = e^{-x^2/2} \langle \theta_x|$    and $|n \rangle = (2^n n!)^{-1/2} | \theta^*_n \rangle$. Note that a rigorous mathematical definition of the basis $\{|x\rangle\}$  requires the framework of rigged Hilbert spaces, see  \cite{Madr1,Madr2} and references therein.} of the Heisenberg algebra. There exists a basis $\{|\theta^*_n\rangle, \ n\in \mathbb{Z}_+\}$ such that the generators $a,a^\dagger$ and the so-called  number operator $N=a^\dagger a$ act respectively as
$a|\theta^*_n\rangle = \sqrt{2} n|\theta^*_{n-1}\rangle$\ , $a^\dagger|\theta^*_n\rangle = \frac{1}{\sqrt{2}}|\theta^*_{n+1}\rangle$ and  $N|\theta^*_n\rangle = n |\theta^*_n\rangle$ where $|\theta^*_0\rangle$ denotes the lowest weight (or `reference') state such that $a|\theta^*_0\rangle=0$. In terms of  the operators $(\tH,\tX)$, one has: 
\beqa
\left( \tX - [\tX,\tH]\right)^n   |  \theta^*_0\rangle=0 \quad \mbox{and}\quad  |\theta^*_n\rangle  = \left( \tX + [\tX,\tH]\right)^n   |  \theta^*_0\rangle \ .\label{nvect}
\eeqa

 On the other hand, there exists a  basis $\{| \theta_x\rangle , \ x \in {\mathbb R}\}$ to which one associates the linear functional denoted  $\langle \theta_x|$ such that the generators $a,a^\dagger$ and $N$ act as
\beqa
 && \langle \theta_x| a  = \frac{1}{\sqrt{2}}\frac{d}{dx} \langle \theta_x| \ ,\qquad   \langle \theta_x| a^\dagger  = \frac{1}{\sqrt{2}}\left(2x-\frac{d}{dx}\right) \langle \theta_x|\ ,\qquad  \langle \theta_x| N= \frac{1}{2}\left(2x-\frac{d}{dx}\right)\frac{d}{dx}\langle \theta_x|\ . \label{actthx}
\eeqa
Then, the transition coefficients   $\langle \theta_x| \theta^*_n \rangle $ solve a bispectral problem which reads as a second-order differential equation and a three term recurrence relation given by:
\beqa
&& \langle \theta_x| \tH| \theta^*_n\rangle = (n+\frac{1}{2}) \langle  \theta_x|  \theta^*_n\rangle\quad \qquad\quad \quad \Leftrightarrow \quad \left(2x-\frac{d}{dx}\right)\frac{d}{dx}H_n(x) = 2n H_n(x) \  \quad \mbox{with} \quad H_n(x) =\langle  \theta_x|  \theta^*_n\rangle \ ,\label{specH2}\\
&& \langle \theta_x| \tX| \theta^*_n\rangle =   \langle  \theta_x|  \theta^*_{n+1}\rangle  +  2n \langle  \theta_x|  \theta^*_{n-1}\rangle \   \quad \Leftrightarrow \quad      2xH_n(x) =  H_{n+1}(x) + 2n H_{n-1}(x) \ .
\eeqa
This system coincides with the well-known defining relations of the  orthogonal Hermite polynomials $H_n(x)$ \cite[eqs. (1.13.3), (1.13.4)]{KS} located at the bottom of the Askey-scheme:
\beqa
 H_n(x) = \sum_{k=0}^{[n/2]}(-1)^k \frac{n!}{k!(n-2k!)}(2x)^{n-2k}\ \quad \mbox{e.g.}\quad  H_0(x)=1\ ,\  H_1(x)=2x\ ,\ H_2(x)=4x^2-2,...
\eeqa

Alternatively, it is also known that the spectral problem (\ref{specH2}) can be formulated within the analytical Bethe ansatz framework, see e.g. \cite[Section 2.1]{Gro}. Namely, 
consider the elementary Baxter T-Q relation:
\beqa
2xQ_n'(x) - Q_n''(x)=\Lambda_n Q_n(x) \qquad \mbox{where} \qquad Q_n(x)=\prod_{i=1}^n(x-x_i)\ .\label{Qherm}
\eeqa
Given $n$ fixed, (\ref{Qherm}) implies that the zeroes $\{x_i|i=1,...,n\}$ satisfy the set of Bethe ansatz equations
\beqa
\frac{Q_n''(x_i)}{Q_n'(x_i)}=2x_i \quad \Leftrightarrow \quad \sum_{j=1,j\neq i}^n \frac{1}{x_i-x_j} = x_i \quad \mbox{for all} \quad i=1,...,n,\ \label{BAQx}
\eeqa
and the spectrum is given by $\Lambda_n=2n$.  Thus, by comparison of (\ref{Qherm}) with  (\ref{specH2}) it follows that the  Q-polynomial in (\ref{Qherm}) can be interpreted as the the transition coefficients $\langle  \theta_x|  \theta^*_n\rangle$:
\beqa
Q_n(x) = 2^{-n} \langle  \theta_x|  \theta^*_n\rangle = 2^{-n} H_n(x) \ .\label{Qx}
\eeqa
Furthermore, using  (\ref{HXP}) and (\ref{actthx}) together with  (\ref{nvect}), the Q-polynomial  can be written in terms of the pair $(\tH,\tX)$  as:
\beqa
Q_n(x) =    \langle  \theta_x| \left( \tX + [\tX,\tH]\right)^n   |  \theta^*_0\rangle  \qquad  \Leftrightarrow  \qquad Q_n(x) =   2^{-n}\left(2x-\frac{d}{dx}\right)^n {\bf 1}\ . \label{actQx}
\eeqa

More generally, given the triplet $(\tH,\tX,i\tP)$  satisfying the Askey-Wilson relations (\ref{aw01}) and (\ref{awp01}) it is natural to consider  the combination (recall that $i\tP=[\tX,\tH]$)
\beqa  
{\textsf I}_{ho}= \kappa\tX + \kappa^*\tH  + \kappa_+[\tX,\tH]\ \label{Iho}
\eeqa
and the corresponding spectral problem. From an algebraic perspective,  the pair $(\tX,{\textsf I}_{ho})$  (or alternatively $(\tH,{\textsf I}_{ho})$) generates the simplest specialization of the family of Heun algebras of Lie type recently introduced in \cite[Definition 2.1]{nico}. For generic parameters $\kappa,\kappa^*,\kappa_+$,  it is well-known that the spectral problem\footnote{However, adding a term such as $\{\tX,\tH\}$ to (\ref{Iho}) where $\{.,.\}$ is the anticommutator leads to a more complicated spectral problem.} for ${\textsf I}_{ho}$ can still be solved in terms of Hermite polynomials (see e.g. \cite{Co}).\vspace{1mm}

The possibility of generalizing the above picture to the Askey-Wilson algebra with generic structure constants \cite{Z91} and to the Heun-Askey-Wilson algebra recently introduced in \cite{BTVZ} is a natural problem to be considered, for various reasons. Some motivations are briefly described in Section \ref{persp}. Let $\rho,\omega,\eta,\eta^*$ be generic scalars. The  Askey-Wilson algebra is generated by $\tA,\tA^*$ subject to the relations \cite{Z91}
\beqa
&&\big[\tA,\big[\tA,\tA^*\big]_q\big]_{q^{-1}}=  \rho \,\tA^*+\omega \,\tA+\eta\mathcal{I} \ , \label{aw1} \\
&&\big[\tA^*,\big[\tA^*,\tA\big]_q\big]_{q^{-1}}= \rho \,\tA+\omega \,\tA^*+\eta^*\mathcal{I} \ .\label{aw2}
\eeqa
In the literature, these relations are called the Askey-Wilson relations. It is known that the Askey-Wilson algebra gives an algebraic framework for all the orthogonal polynomials of the Askey-scheme \cite{Z91}. Also, it is known that the zeroes of the Askey-Wilson polynomials satisfy a class of Bethe ansatz equations \cite{WZ,DE}. Define the Heun-Askey-Wilson element \cite{BTVZ,nico}:
\beqa
{\textsf I}(\kappa,\kappa^*,\kappa_+,\kappa_-)=  \kappa\,\tA+\kappa^* \tA^*+\kappa_+ \chi^{-1}\left[\tA,\tA^*\right]_q+\kappa_-\chi\left[\tA^*,\tA\right]_q\ \label{I}
\eeqa
where $\kappa,\kappa^*,\kappa_\pm$ and\footnote{The parameter $\chi$ is introduced for further convenience.} $\chi\neq 0$ are scalars. To our knowledge, the diagonalization of a  Heun-Askey-Wilson operator associated with (\ref{I}) has not been considered yet in full generality\footnote{For a certain image of (\ref{I}) in $U_q(sl_2)$, see however \cite{WZ}.}. For finite dimensional representations, as we will show
it turns out that this problem can be solved within the boundary quantum inverse scattering method \cite{Skly88}. In fact, the Heun-Askey-Wilson operator commutes with a given transfer matrix \cite{Z95,Bas2} associated with non-diagonal reflection matrices \cite{DeG,GZ}, see (\ref{trans}) below.
Such kind of reflection matrices contain arbitrary constant parameters (boundary couplings) which in general break the $U(1)$ symmetry of the transfer matrix, and it prevents the application of the standard Bethe ansatz techniques to study its diagonalization, unless that restrictions on the boundary couplings are imposed, see\footnote{For a more complete list of references and an account of historic details on this subject, we refer the reader to the Introduction of \cite{MABAq1}.} {\it e.g.} \cite{cao03,TQ1,YZ07,Doikou06,FNR,BCR12,PL13}. An important progress for the unrestricted cases was achieved by the introduction of the off-diagonal Bethe ansatz \cite{cao13}, a method that proposes an {\it inhomogeneous} Baxter T-Q equation as solution of the spectral problem for integrable models without $U(1)$ symmetry \cite{ODBA}. Beyond the computation of the spectrum, a modification of the algebraic Bethe ansatz was developed in \cite{BC2013,MABAq1,Cra14,MABAq2,MABAq3} providing the construction of the associated off-shell Bethe states\footnote{The on-shell Bethe states can also be retrieved from the off-diagonal Bethe ansatz \cite{onshellbv}.}, which in particular allows the computation of scalar products between on-shell/off-shell Bethe states, see  \cite{scalar1,scalar2,proof1,whydet} and references therein. The main feature of the modified algebraic Bethe ansatz are the off-shell relations satisfied by the `creation operators', see e.g. Lemma \ref{prop:diagonaloffshell} and Conjecture \ref{prop:genericoffshell}. In the algebraic Bethe ansatz perspective, these off-shell relations are the origin of the inhomogeneous term in the Baxter T-Q equation.\vspace{1mm}

In the present paper, for any irreducible finite dimensional representation of the Askey-Wilson algebra and $q\neq 1$, the construction of eigenvectors of (\ref{I}) within the (modified) algebraic Bethe ansatz framework and the analog of the construction  (\ref{Qherm}), (\ref{BAQx}), (\ref{Qx}) is considered. In this case, the theory of Leonard pairs developed in \cite{T03,T04} offers the proper mathematical setting. The main results are the following:
\vspace{1mm}

$(i)$ The Heun-Askey-Wilson operator associated with (\ref{I}) is diagonalized on any irreducible finite dimensional representation $(\bar\pi,{\bcV})$ of dimension $2s+1$ ($s$ is an integer or half-integer), using the theory of Leonard pairs combined with the algebraic Bethe ansatz technique. Three cases of (\ref{I})  indexed by `{\it a}' are considered in details: the {\it special} cases (i.e. $a=sp$) $\kappa\neq 0,\kappa^*=\kappa^\pm=0$ or $\kappa^*\neq 0,\kappa=\kappa_\pm=0$; the {\it diagonal} case (i.e. $a=d$)  $\kappa,\kappa^*\neq 0,\kappa_\pm=0$; the generic case (i.e. $a=g$) $\kappa,\kappa^*,\kappa_\pm\neq 0$. In all cases, the solution of the spectral problem\footnote{Note that for the generic case, the eigenvalue expression can be obtained from the results in \cite{cao15}.} takes the form:
\beqa
 \bar\pi\left({\textsf I}(\kappa,\kappa^*,\kappa_+,\kappa_-) \right)|\Psi_{{a},\epsilon}^{M}(\bar u)\rangle= \Lambda_{a,\epsilon}^{M}|\Psi_{{a},\epsilon}^{M}(\bar u)\rangle \ ,\label{specgeneral}
\eeqa
where the eigenstates $|\Psi_{{a},\epsilon}^{M}(\bar u)\rangle$ are given in the form of on-shell Bethe states (with $M=0,1,...,2s$ for $a=sp$ and $M=2s$ for $a=d,g$) obtained from  the so-called `reference state' with index $\epsilon\in\{\pm\}$ through successive actions of a `creation' operators $\{\mathscr{B}^{\epsilon}(u_i,m_i)\}$, thus generalizing (\ref{actQx}). According to the choice of parameters and reference state, the Bethe roots $\bar u=\{u_1,...,u_M\}$ satisfy {\it homogeneous} or {\it inhomogeneous} Bethe ansatz equations. See Propositions \ref{p31}, \ref{p33}, \ref{p34}, \ref{p35}. For the case $a=d$, the inhomogeneous term is related with the characteristic polynomial of an operator denoted $\tA^\diamond$ that forms a  Leonard triple with $\tA,\tA^*$. Importantly, this provides a  proof of the off-shell relations in Lemma \ref{prop:diagonaloffshell} only based on the representation theory of the Askey-Wilson algebra.\vspace{1mm}

$(ii)$ An infinite dimensional representation $(\pi, {\cal V})$ of the Askey-Wilson algebra is introduced. In this representation, the  operators $\pi(\tA),\pi(\tA^*)$ act as  second-order q-difference operators whereas the  Heun-Askey-Wilson operator $\pi(\textsf{I}(\kappa,\kappa^*,\kappa_+,\kappa_-))$ acts in general as a fourth-order q-difference operator. For the special case $\kappa=\kappa_\pm=0$, the spectral problem for the Heun-Askey-Wilson operator q-difference operator (reduced to the one for $\pi(\tA^*)$)  produces an homogeneous Baxter T-Q relation generalizing (\ref{Qherm}). In this case, the Baxter Q-polynomial of degree $M$  follows from the `transition' coefficient (see Lemma \ref{lem5}): 
\beqa
Q_M(Z) = {\cal N}_M({\bar u})^{-1} \langle z|\Psi_{{sp},+}^{M}(\bar u)\rangle = \prod_{i=1}^M\left(Z- \frac{qu_i^2+q^{-1}u_i^{-2}}{q+q^{-1}}\right)  \quad  \mbox{where} \quad  Z=\frac{z+z^{-1}}{q+q^{-1}}\ ,
\eeqa
thus generalizing (\ref{Qx}).
Furthermore, it solves a bispectral problem for the Askey-Wilson polynomial, as expected. For the diagonal case $\kappa,\kappa^*\neq 0,\kappa_\pm=0$,  the action of the q-difference operator $\pi(\textsf{I}(\kappa,\kappa^*,0,0))$ on the Baxter Q-polynomial  produces an inhomogeneous Baxter T-Q relation generalizing (\ref{Qherm}).  In this case, the Baxter Q-polynomial of degree $2s$ is also expressed in terms of the `transition' coefficient  (see Lemma \ref{lem5}):
\beqa
Q_{2s}(Z) = {\cal N}_{2s}({\bar u})^{-1} \langle z|\Psi_{{d},+}^{2s}(\bar u)\rangle.
\eeqa
For the generic case $\kappa,\kappa^*,\kappa_\pm\neq 0$  and a different realization of the Heun-Askey-Wilson operator, a similar picture holds as briefly mentioned.\vspace{1mm}

$(iii)$  For each system of homogeneous or inhomogeneous  Bethe ansatz equations derived in Section \ref{s3}, an alternative presentation in terms of the `symmetrized' Bethe roots (\ref{sBr}) is obtained. In particular, this presentation gives a simpler characterization of the Bethe roots than the one based on the standard Bethe equations.  See Proposition \ref{propP}.
\vspace{1mm}

$(iv)$  An algebraic Bethe ansatz solution for the q-analog of the quantum Euler top \cite{WZ,Tu16} and various examples of three-sites Heisenberg spin chains with three-body terms, inhomogeneous couplings and  boundary interactions. See Section \ref{App} where numerical examples are given. \vspace{1mm}

The paper is organized as follows. In Section 2, we review the derivation of the Heun-Askey-Wilson element (\ref{I}) based on the connection between the reflection equation algebra and the Askey-Wilson algebra \cite{WZ,Bas2}. Using the transfer matrix formalism and the theory of Leonard pairs \cite{T03,T04}, the diagonalization of the Heun-Askey-Wilson operator associated with  (\ref{I}) is considered within the framework of the algebraic Bethe ansatz in Section 3. Namely, for various choices of parameters $\{\kappa,\kappa^*,\kappa_+,\kappa_-\}$, the eigenstates of the Heun-Askey-Wilson operator  are given in the form of Bethe states and associated Bethe equations, and in terms of Leonard pairs' data. In each case, an alternative presentation of the Bethe ansatz equations is given in terms of a system of polynomial equations in  the `symmetrized' variables (\ref{sBr}). Also, for any off-shell or on-shell Bethe state an expansion formula is given using the Poincar\'e-Birkhoff-Witt basis of the Askey-Wilson algebra. It follows that certain families of Bethe states generate explicit bases for Leonard pairs.
In Section 4, four different types of Baxter T-Q relations are derived, either homogeneous or inhomogeneous. The corresponding system of Bethe equations produces the Bethe equations derived in Section 3 and the solutions are expressed as symmetric polynomials, identified as Q-polynomials.  Independently, a second-order q-difference operator realization of the Heun-Askey-Wilson element (\ref{I}) is given for the special and diagonal cases.  It is shown that its action on the Q-polynomial  produces the T-Q relations. For the generic case, a different realization has to be considered, as briefly mentioned. For the special case $\kappa=\kappa_\pm=0$ the Q-polynomial is expressed in terms of the orthogonal Askey-Wilson polynomials. An interpretation of the Q-polynomial as transition coefficients arises naturally.
In Section \ref{App}, we apply the previous results to the diagonalization of various examples of integrable models: a q-analog of the quantum Euler top and Heisenberg spin chains with three-sites, that can be viewed as deformations of the three-sites $U_q(sl_2)$-invariant XXZ spin chain. In Section \ref{persp}, some perspectives are briefly described.
In Appendices \ref{apA}, \ref{Sec:coefcommut}, formulas are collected. In Appendix \ref{apD}, the proof of Lemma \ref{prop:diagonaloffshell} is given. Appendix \ref{prP} is devoted to the proof of Proposition \ref{propBAU}. In Appendix \ref{appqdiff}, different realizations of the Askey-Wilson algebra in terms of first or second-order q-difference operators are given. In Appendix \ref{apF}, realizations of the dynamical operators in terms of the q-difference operators are displayed. 
\vspace{2mm}

{\bf Notations:} The parameter $q$ is assumed not to be a root of unity and is different than $1$. We write $[X,Y]_q = qXY - q^{-1}YX$, the commutator $[X,Y]=[X,Y]_{q=1}$. The identity element is denoted $\mathcal{I}$. We will use the standard $q$-shifted factorials (also called q-Pochhammer functions) \cite{KS}, $q$-number and binomial, respectively:
\beqa
(a;q)_n=\prod_{k=0}^{n-1}(1-aq^{k}), \label{poch} \quad  [n]_q=\frac{q^n-q^{-n}}{q-q^{-1}}\ , \ [0]_q=1\ ,\quad  \bin m  n=\frac{m!}{n!(m-n)!} \ .
\eeqa
We also use the notation
\beqa
b(x)=x-x^{-1}\,\,.\label{b}
\eeqa
For the elementary symmetric polynomials in the variables $\{x_i|i=1,...,n\}$, we use the notation:
\beqa
 \textsf{e}_{k}(x_1,x_2,...,x_n) = \sum_{1\leq j_1 < j_2 < \cdots < j_k \leq n } \! \! \! \! \! \! \! x_{j_1}x_{j_2}\cdots x_{j_k}\label{esymdef}
\eeqa
\vspace{1mm}

\section{The Heun-Askey-Wilson element from the reflection equation}
In this section, we show how the Heun-Askey-Wilson element (\ref{I}) introduced in \cite{BTVZ} follows from the transfer matrix associated with the solution of the reflection equation constructed in \cite{Z95,Bas2}. Most of the material given in this section is taken from the works \cite{Skly88,Z95,Bas2,BK2}, so we skip the details.
\vspace{1mm} 

Let the operator-valued function $R:{\mathbb{C}}^*\mapsto \mathrm{End}({\mathbb{C}}^2\otimes {\mathbb{C}}^2)$ be the intertwining operator (quantum $R-$matrix) between the tensor product of two fundamental representations associated with the quantum algebra $U_q(\widehat{sl_2})$. The element $R(u)$ depends on the deformation parameter $q$ and is defined by
\begin{align}
R(u) =\left(
\begin{array}{cccc} 
 uq -  u^{-1}q^{-1}    & 0 & 0 & 0 \\
0  &  u -  u^{-1} & q-q^{-1} & 0 \\
0  &  q-q^{-1} & u -  u^{-1} &  0 \\
0 & 0 & 0 & uq -  u^{-1}q^{-1}
\end{array} \right) \ ,\label{R}
\end{align}
where $u$ is the so-called spectral parameter. Then $R(u)$ satisfies the quantum Yang-Baxter equation in the space ${\mathbb{C}}^2\otimes {\mathbb{C}}^2\otimes {\mathbb{C}}^2$. Using the standard notation $R_{ij}(u)\in \mathrm{End}({({\mathbb{C}}^2)}_i\otimes {({\mathbb{C}}^2)}_j)$, it reads 
\begin{align}
R_{12}(u/v)R_{13}(u)R_{23}(v)=R_{23}(v)R_{13}(u)R_{12}(u/v)\ \qquad \forall u,v.\label{YB}
\end{align}

Consider the reflection equation (also called the boundary quantum Yang-Baxter equation) introduced in the context of the boundary quantum inverse scattering theory (see \cite{cher84},\cite{Skly88} for details). For the $U_q(\widehat{sl_2})$ $R-$matrix  (\ref{R}), the reflection equation reads:
\begin{align} R(u/v)\ (K(u)\otimes I\!\!I)\ R(uv)\ (I\!\!I \otimes K(v))\
= \ (I\!\!I \otimes K(v))\ R(uv)\ (K(u)\otimes I\!\!I)\ R(u/v)\ ,
\label{RE} \end{align}
where $K(u)$ is a $2\times 2$ matrix. We are now interested in a certain class of solutions of the reflection equation (\ref{RE}) for which the entries of the matrix $K(u)$ are assumed to be Laurent polynomials in the spectral parameter whose coefficients are elements in certain homomorphic  images of the q-Onsager algebra (see \cite{BK2}). Here, we consider the simplest solution of this type. We refer the reader to \cite{Z95,Bas2} for details.
\begin{prop}\label{Kd} Assume  $\{\chi,\rho,\omega,\eta,\eta^*\}$ are generic scalars in ${\mathbb{C}}^*$. Let $K_{11}(u)\equiv \cA(u)$, $K_{12}(u)\equiv \cB(u)$, $K_{21}(u)\equiv \cC(u)$, $K_{22}(u)\equiv \cD(u)$ be the elements of the $2\times 2$ square matrix  $K(u)$ such that:
\beqa
&&\cA(u)=(u^2-u^{-2})\left(qu\,\tA-q^{-1}u^{-1}\tA^*\right)-(q+q^{-1})\rho^{-1}\left(\eta u+\eta^*u^{-1}\right)\,,\label{monoA}\\
&&\cD(u)=(u^2-u^{-2})\left(qu\,\tA^*-q^{-1}u^{-1}\tA\right)-(q+q^{-1})\rho^{-1}\left(\eta^* u+\eta u^{-1}\right)\,,\label{monoD}\\
&&\cB(u)=\chi (u^2-u^{-2})\left(\rho^{-1}\left(\left[\tA^*,\tA\right]_q+\frac{\omega}{q-q^{-1}}\right)+\frac{qu^2+q^{-1}u^{-2}}{q^2-q^{-2}}\right)\,,\label{monoB}\\
&&\cC(u)=\rho \chi^{-1}\, (u^2-u^{-2})\left(\rho^{-1}\left(\left[\tA,\tA^*\right]_q+\frac{\omega}{q-q^{-1}}\right)+\frac{qu^2+q^{-1}u^{-2}}{q^2-q^{-2}}\right)\label{monoc}\,.
\eeqa
Then $K(u)$ satisfies the reflection equation (\ref{RE}) provided $\tA,\tA^*$ generate the Askey-Wilson algebra with defining relations (\ref{aw1}),(\ref{aw2}).
\end{prop}

Note that the scalar parameter $\chi$ does not enter in the structure constants of the Askey-Wilson relations but it is introduced for further convenience. \vspace{1mm}

Two explicit examples are now given. Let $\{q^{\pm s_3},S_{\pm}\}$ denote the generators of the quantum algebra $U_q(sl_2)$ with defining relations:
\ben
\left[s_3,S_{\pm}\right]=\pm S_{\pm}\,,\quad \left[S_+,S_-\right]=\frac{q^{2s_3}-q^{-2s_3}}{q-q^{-1}}\,.
\een
The Casimir element
\beqa
C = (q-q^{-1})^2S_- S_+ + q^{2s_3+1} + q^{-2s_3-1}\ 
\eeqa
is central. Introduce the coproduct $\Delta: U_q(sl_2) \rightarrow U_q(sl_2) \otimes U_q(sl_2)$ such that:
\beqa
\Delta(S_+) = S_+ \otimes I\!\!I + q^{2s_3}\otimes S_+ \, \quad \Delta(S_-) = S_- \otimes q^{-2s_3} +  I\!\!I \otimes S_- \ ,\quad 
\Delta(q^{s_3})= q^{s_3}\otimes q^{s_3}\ .\vspace{1mm}
\eeqa

\begin{example}\label{ex1}\cite{GZ2} 
Define the parameters
\beqa
&& k_+= -\frac{1}{2} q^{\nu } \left(q-q^{-1}\right)\ ,\quad k_-= \frac{1}{2}q^{\nu'}
   \left(q-q^{-1}\right) \ ,\quad \epsilon_+= \cosh (\mu ) q^{\frac{1}{2}
   \left(\nu+\nu' \right)}\ ,\quad \epsilon_-= \cosh \left(\mu ^\prime\right) q^{\frac{1}{2}
   \left(\nu + \nu' \right)}\ ,\label{parbis}
\eeqa
where $\mu,\mu',\nu,\nu',v$ are generic scalars. The following gives an example of elements $\tA,\tA^* $ satisfying the Askey-Wilson relations (\ref{aw1}), (\ref{aw2}):
\beqa
\tA &\rightarrow & k_+ v q^{1/2}S_+q^{s_3}  + k_- v^{-1} q^{-1/2}S_-q^{s_3} + \epsilon_+ q^{2s_3}\ ,\label{Aex1}\\
 \tA^* & \rightarrow & k_+ v^{-1} q^{-1/2}S_+q^{-s_3}  + k_- v q^{1/2}S_-q^{-s_3} + \epsilon_- q^{-2s_3}\ ,\label{Astarex1}
\eeqa
where the structure constants $\rho,\omega,\eta,\eta^*$ are given by (\ref{sc1})-(\ref{sc4}) with the substitution $q^{2s+1}+q^{-2s-1}\rightarrow C$.
\end{example}
Let us mention that this example can be derived using the `dressing' procedure \cite{Skly88} that generates solutions of  the reflection equation. For details, see \cite{Z95}. 

Another example follows from \cite{GZ93b} (see also \cite{Huang}). It will connect the Heun-Askey-Wilson element (\ref{I}) to Hamiltonians of  3-sites spin chains, see Section \ref{App}.
\begin{example}\label{ex2}\cite{GZ93b} (see also \cite{Huang}) The Askey-Wilson algebra is embedded into $U_q(sl_2)\otimes U_q(sl_2) \otimes U_q(sl_2)$. Define $c_1=C \otimes  I\!\!I \otimes  I\!\!I  $, $c_2=I\!\!I \otimes  C \otimes  I\!\!I $, $c_3=I\!\!I \otimes I\!\!I \otimes  C$ and $c_4=(id \times \Delta )\circ \Delta(C) $.
The following gives an example of elements $\tA,\tA^* $ satisfying the Askey-Wilson relations (\ref{aw1}), (\ref{aw2}):
\beqa
\tA &\rightarrow &\Delta(C)\otimes I\!\!I \ ,\label{Aex2}\\
 \tA^* & \rightarrow &  I\!\!I \otimes \Delta(C) \ ,\label{Astarex2}
\eeqa
where the structure constants  are given by:
\beqa
\rho&=&-(q^2-q^{-2})^2\  ,\label{scp1}\\
 \omega&=& -(q-q^{-1})^2 \left(c_1c_3 + c_2c_4\right)\ ,  \label{scp2}\\
  \eta^*&=&  \frac{(q^2-q^{-2})^2}{(q+q^{-1})}\left(  c_1c_2 + c_3c_4 \right) \ ,\label{scp3} \\ 
 \eta&=&  \frac{(q^2-q^{-2})^2}{(q+q^{-1})}\left(  c_2c_3+c_1c_4 \right)   \ . \label{scp4}
\eeqa
\end{example}

Given a solution of the reflection equation, it is known that a generating function for mutually commuting quantities is provided by the so-called transfer matrix  \cite{Skly88}. In the present case, consider the most general scalar solution of the  so-called ``dual'' reflection equation  given by \cite{DeG,GZ}:
\ben\label{KP}
K^{+}(u)=
\left(
\begin{array}{cc}
q u\kappa+q^{-1}u^{-1}\kappa^*  & \kappa_+ (q^2u^2 - q^{-2}u^{-2}) \\
\kappa_- \rho (q^2u^2 - q^{-2}u^{-2})  &  q u\kappa^*+q^{-1}u^{-1}\kappa 
\end{array}
\right)\ ,
\een
where $\kappa_\pm,\kappa,\kappa^*$ are generic scalars in ${\mathbb{C}}$. The transfer matrix reads $t(u)=\mbox{tr}\left(K^+(u)K(u)\right)$, where the trace is taken over the two-dimensional auxiliary space. Expanded in the spectral parameter $u$, for $K(u)$ with (\ref{monoA})-(\ref{monoc}) the transfer matrix produces the Heun-Askey-Wilson element  (\ref{I}):
\ben
&&t(u)=(q^2u^2-q^{-2}u^{-2})(u^2-u^{-2})\left(\kappa\,\tA+\kappa^* \tA^*+\kappa_+\chi^{-1}\left[\tA,\tA^*\right]_q+\kappa_-\chi\left[\tA^*,\tA\right]_q\right) + {\cal F}_0(u)\ , \label{trans}\een
where ${\cal F}_0(u)$ is a scalar function\footnote{The expression of ${\cal F}_0(u)$ is not needed in further analysis, so we omit its explicit expression.}. It follows that the spectral problem for the Heun-Askey-Wilson operator associated with (\ref{I}) and the transfer matrix (\ref{trans}) are identical. In the next section, we will use this connection to diagonalize the Heun-Askey-Wilson operator using the algebraic Bethe ansatz. \vspace{1mm}

Let us mention that the known central element (the so-called Casimir element) of the Askey-Wilson algebra can be extracted from the so-called Sklyanin's quantum determinant $\Gamma(u)$ given in \cite{Skly88}. For the reflection equation algebra, recall that the quantum determinant is defined as:
\beqa
\Gamma(u)=tr\big(P^{-}_{12}(K(u)\otimes I\!\!I)\ R(u^2q) (I\!\!I \otimes K(uq))\big),\ \label{gamma}
\eeqa
where $P^-_{12}=(1-P)/2$ with $P=R(1)/(q-q^{-1})$. As shown in \cite{Skly88}, \ 
$\big[\Gamma(u),(K(u))_{ij}\big]=0$  \ for any $i,j$.  Inserting  $K(u)$ with (\ref{monoA})-(\ref{monoc}) into (\ref{gamma}), one gets:
\beqa
\Gamma(u)= \frac{(u^2q^2-u^{-2}q^{-2})}{2(q-q^{-1})}\left(\Delta(u) - \frac{2\rho}{(q-q^{-1})}\right),\nonumber
\eeqa
with
\beqa
&&\Delta(u)=- \frac{2(q-q^{-1})}{\rho}  \left(u^2-u^{-2}\right) \left(q^2 u^2-q^{-2} u^{-2}\right)
\left(\Gamma + \gamma_0(u)\right)
 \  \nonumber
\eeqa
where the explicit form of the scalar function $\gamma_0(u)$ is omitted for simplicity and
\beqa
&& \Gamma = \left(q^2-q^{-2}\right) q^{-1}\tA \tA^*\left[\tA^*,\tA\right]_q
-q^{-2}\left[\tA^*,\tA\right]_q^2
+q^{-2}\rho \tA^2
+\rho q^2 {\tA^*}^2+\omega(\tA\tA^*+\tA^*\tA)
\label{Gam}\\&&
\quad\quad+ \ \eta^*(1+q^{-2})\tA
+q^2\eta(1+q^{-2})\tA^* \ .\nonumber
\eeqa
It follows that $\Gamma$ satisfies
\beqa
[\Gamma,\tA]=[\Gamma,\tA^*]=0\ \label{comcas}
\eeqa
 i.e. $\Gamma$ is central in the Askey-Wilson algebra. Note that this result provides an alternative derivation of the central element given in \cite[eq. (1.3)]{Z91}. Also, let us mention that $\Gamma$ is a restriction of the expression computed in \cite{BB}  for the q-Onsager algebra. Indeed, the Askey-Wilson algebra can be viewed as a certain quotient of the q-Onsager algebra by the relations (\ref{aw1}), (\ref{aw2}).\vspace{1mm}

To conclude this section, let us make few additional comments.  As mentioned in the introduction, for the quantum harmonic oscillator each couple $(\tH,\tX)$ and $(\tH,i\tP)$ satisfies the simplified Askey-Wilson relations  (\ref{aw01}), (\ref{awp01}). It is thus natural to ask for such triplet in the context of the Askey-Wilson algebra.  Let $\tA,\tA^*$ satisfy the Askey-Wilson relations (\ref{aw1}), (\ref{aw2}). Define the element:
\beqa
\tC =  \frac{i}{\sqrt\rho}[\tA^*,\tA]_q + \frac{i\omega}{(q-q^{-1})\sqrt \rho} \ .\label{Cop}
\eeqa
Then, it is easy to show that
the couple $(\tA,\tC)$ satisfies the Askey-Wilson relations:
\beqa
&&\big[\tA,\big[\tA,\tC\big]_q\big]_{q^{-1}}=  \rho \,\tC+\frac{i(q-q^{-1})}{\sqrt \rho}\eta \,\tA - \frac{i\sqrt \rho}{(q-q^{-1})}\omega \mathcal{I} \ , \label{awc1} \\
&&\big[\tC,\big[\tC,\tA\big]_q\big]_{q^{-1}}= \rho \,\tA+\frac{i(q-q^{-1})}{\sqrt \rho}\eta \,\tC+\eta^*\mathcal{I} \ .\label{awc2}
\eeqa
Similar relations are easily derived for the couple $(\tA^*,\tC)$. In the literature, the triplet $(\tA,\tA^*,\tC)$ is related with the concept of Leonard triple. We refer the reader to \cite{curt,Huang2} for details. Some properties of the element (\ref{Cop}) will be used in Appendix \ref{apD}.\vspace{1mm}

Let us also mention that the transformation from the pair $(\tH,\tX)$ to the pair (\tH,i\tP) relating the  Askey-Wilson relations (\ref{aw01}), (\ref{awp01}) generalizes to the q-deformed case as follows. Given $(\tA^*,\tA)$ satisfying (\ref{aw1}), (\ref{aw2}) with structure constants $\rho, \omega, \eta,\eta^*$, define a new pair of elements $(\bar\tA^*,\bar\tA)$ by:
\ben
\bar \tA^* = \tA^* \ ,\qquad \bar \tA =\frac{i}{\sqrt\rho}[\tA,\tA^*]_q+\frac{i\omega}{(q-q^{-1})\sqrt\rho}\,. \label{dirA}
\een
Then, by direct calculation one finds that $(\bar\tA^*,\bar\tA)$ satisfy
\beqa
&&\big[\bar\tA,\big[\bar\tA,\bar\tA^*\big]_q\big]_{q^{-1}}=  \rho \,\bar\tA^*+\bar\omega \,\bar\tA+\bar\eta\mathcal{I} \ , \label{aw1map} \\
&&\big[\bar\tA^*,\big[\bar\tA^*,\bar\tA\big]_q\big]_{q^{-1}}= \rho \,\bar\tA+\bar\omega \,\bar\tA^*+\bar\eta^*\mathcal{I} \ \label{aw2map}
\eeqa
with the structure constants
\ben\label{barsc}
\bar\omega = \frac{i(q-q^{-1})\eta^*}{\sqrt\rho}\,,\quad
\bar\eta^*=-\frac{i\omega\sqrt\rho}{q-q^{-1}}\,,\quad
\bar\eta = \eta\,.
\een

In Section \ref{s4} and Appendix \ref{appqdiff} we will give various examples of operators acting on an infinite dimensional vector space which provide realizations of the Askey-Wilson algebra (\ref{aw1}), (\ref{aw2}). In Appendix \ref{appqdiff}, the invertible transformation (\ref{dirA}) will be used.

\section{Diagonalization of the Heun-Askey-Wilson operator via the algebraic Bethe ansatz}\label{s3}
In this section, using (\ref{trans}) the Heun-Askey-Wilson operator associated with (\ref{I}) is diagonalized on any irreducible finite dimensional representation of the Askey-Wilson algebra (\ref{aw1})-(\ref{aw2}) for different choices of parameters $\kappa,\kappa^*,\kappa_\pm$ using the framework of the algebraic Bethe ansatz \cite{cao03,Doikou06}  and its modified version \cite{MABAq1,MABAq2,MABAq3}  combined with the theory of Leonard pairs \cite{T03,T04}. The eigenstates are obtained in the form of Bethe states, and the spectrum of the Heun-Askey-Wilson operator is given in terms of solutions (the so-called Bethe roots) of Bethe equations. For each choice of the parameters, it is shown that the corresponding Bethe ansatz equations admit an alternative presentation in terms of the `symmetrized' variables $U_i$ given by (\ref{sBr}). Also, up to an overall factor the Bethe states are expressed in the Poincar\'e-Birkhoff-Witt basis of the Askey-Wilson algebra with polynomial coefficients in $U_i$. Two eigenbases for Leonard pairs in terms of Bethe states are derived in this framework.

\subsection{Preliminaries}
\subsubsection{Leonard pairs}\label{ss31}
Let $(\bar\pi,{\bcV})$ denote a finite dimensional representation on which the elements of the Askey-Wilson algebra ${\textsf A},{\textsf A}^*$ act as $\bar\pi(\tA),\bar\pi(\tA^*)$. Assume $\bar\pi(\tA),\bar\pi(\tA^*)$ are diagonalizable on  ${\bcV}$, the spectra are multiplicity-free and ${\bcV}$  is irreducible. Then, by  \cite[Theorem 6.2]{T04} the operators  $\bar\pi(\tA),\bar\pi(\tA^*)$ form a Leonard pair, see \cite[Definition 1.1]{T04}. Define $\dim({\bcV})=2s+1$ with $s$ an integer or half-integer. Given  the eigenvalue sequence $\theta_0,\theta_1,...,\theta_{2s}$ associated with  $\bar\pi(\tA)$  (resp.  the eigenvalue sequence $\theta^*_0,\theta_1^*,...,\theta^*_{2s}$ associated with  $\bar\pi(\tA^*)$), one associates an eigenbasis with vectors  $|\theta_0\rangle,|\theta_1\rangle,...,|\theta_{2s}\rangle$ (resp. an eigenbasis with vectors $|\theta^*_0\rangle,|\theta^*_1\rangle,...,|\theta^*_{2s}\rangle$). For a Leonard pair, recall that:\vspace{1mm}

(i)  in the eigenbasis of  $\bar\pi(\tA)$, then  $\bar\pi(\tA^*)$ acts as a tridiagonal matrix;

(ii) in the eigenbasis of  $\bar\pi(\tA^*)$, then  $\bar\pi(\tA)$ acts as a tridiagonal matrix.\vspace{1mm}

According to the properties (i)-(ii) the action of the Leonard pair    $\bar\pi(\tA),\bar\pi(\tA^*)$ on the eigenvectors takes the form:
\beqa
\qquad \quad \bar\pi(\tA) |\theta_M\rangle &=& \theta_M|\theta_M\rangle \ ,\quad  \bar\pi(\tA^*) |\theta_M\rangle =
a_{M,M+1} |\theta_{M+1}\rangle +  a_{M,M} |\theta_{M}\rangle +  a_{M,M-1}|\theta_{M-1}\rangle \ ,\label{tridAstar}\\
\qquad \quad \bar\pi(\tA^*) |\theta^*_M\rangle &=& \theta^*_M|\theta^*_M\rangle \ ,\quad  \bar\pi(\tA) |\theta^*_M\rangle =
a^*_{M,M+1} |\theta^*_{M+1}\rangle +  a^*_{M,M} |\theta^*_{M}\rangle +  a^*_{M,M-1}|\theta^*_{M-1}\rangle\ , \label{tridAstar2}
\eeqa
where  the coefficients $a_{0,-1}=a_{2s,2s+1}=a^*_{0,-1}=a^*_{2s,2s+1}=0$, and the set of coefficients $\{a_{M,M\pm 1}\},\{a_{M,M}\}$ and $\{a^*_{M,M\pm 1}\},\{a^*_{M,M}\}$ are determined according to the representation chosen.
Furthermore,  the eigenvalue sequences are such that (see \cite[Theorem 4.4 (case I)]{Ter03}):
\beqa
 \theta_M= b q^{2M} + cq^{-2M}\ , \quad   \theta^*_M= b^* q^{2M} + c^*q^{-2M}\ ,\label{st}
\eeqa
where $b,c,b^*,c^*$ are generic scalars. Due to the existence of the Askey-Wilson relations (\ref{awc1}), (\ref{awc2}), on the vector space $\bar  {\cal V}$ it is clear that $(\tA,\tC)$ and $(\tA^*,\tC)$ both give examples of Leonard pairs. Let $|\theta^\diamond_0\rangle, ...,|\theta^\diamond_{2s}\rangle   $ denote the eigenvectors of $\bar \pi(\tC)$ with respective  eigenvalues $\theta^\diamond_0,...,\theta^\diamond_{2s}$. Adapting the results of \cite{Huang2}, one finds that the eigenvalues of $\tC$ take the form:
\beqa
\theta^\diamond_M = b^\diamond q^{2M} + c^\diamond q^{-2M}\ ,\label{sC}
\eeqa
where $b^\diamond,c^\diamond$ are generic scalars.
In this parametrization, note that the structure constant $\rho$ is given by \cite[Lemma 4.5]{Ter03}:
\beqa
\rho=-bc(q^2-q^{-2})^2=-b^*c^*(q^2-q^{-2})^2 =-b^\diamond c^\diamond(q^2-q^{-2})^2\ .\label{rho}
\eeqa

On the finite dimensional vector space $\bar {\cal V}$,  in addition to the Askey-Wilson relations (\ref{aw1}),  (\ref{aw2}) by the Cayley-Hamilton theorem one has:
\beqa
\prod_{M=0}^{2s}\left( \tA - \theta_M\right) = 0\ , \quad \prod_{M=0}^{2s}\left( \tA^* - \theta^*_M\right) = 0\ ,\quad \prod_{M=0}^{2s}\left( \tC - \theta^\diamond_M\right) = 0\ .\label{polyc}
\eeqa

\begin{example}\label{ex3}\cite{GZ2}
With respect to Example \ref{ex1}, let $(\bar\pi,V(s))$ denote the irreducible (spin-$s$) representation of $U_q(sl_2)$  of dimension $\dim(V(s))=2s+1$  on which the generators $\{q^{\pm s_3},S_\pm\}$ act as:
\ben
 \bar \pi (q^{\pm s_3}) = \sum_{k=1}^{2s+1}q^{\pm ( s+1-k)}E_{kk}\,,\  \bar \pi (S_-) = \sum_{k=1}^{2s}  \sqrt{[k]_q[2s+1-k]_q} E_{k+1k}\, , \  \bar \pi (S_+) = \sum_{k=1}^{2s}  \sqrt{[k]_q[2s+1-k]_q} E_{kk+1}\, ,\nonumber
\een
where the matrix $E_{ij}$ has a unit at the intersection of the $i-$th row and the $j-$th column, and all other entries are zero.
 For generic scalar parameters $k_\pm,\epsilon_\pm,v$ and $q$ in (\ref{Aex1}), (\ref{Astarex1}),  then $\bar\pi(\tA),\bar\pi(\tA^*)$ is a Leonard pair. There exists two bases on which $\bar \pi(\tA), \bar \pi(\tA^*)$ act as (\ref{tridAstar}), (\ref{tridAstar2}) and the spectra  take the form (\ref{st}) with (\ref{par}).  See e.g. \cite{BVZ16} for details.
\end{example}
\begin{example}\label{ex4} \cite{GZ93b,Huang} With respect to Example \ref{ex2}, consider the tensor product of three irreducible representations of $U_q(sl_2)$ denoted $V(j_1)\otimes V(j_2) \otimes V(j_3)$. Following \cite{Huang}, define $\Sigma$ the set consisting of all pairs $(\ell,k)$ of integers such that
\beqa
0 \leq \ell \leq \min \{2j_1+2j_2,2j_2+2j_3,2j_1+2j_3,j_1+j_2+j_3\} \ ,\qquad 0 \leq k  \leq 2j_1+2j_2+2j_3 - 2\ell \  .
\eeqa
Define  $V_k(\ell)$ as the vector subspace of $V(j_1)\otimes V(j_2) \otimes V(j_3)$ spanned by the simultaneous 
eigenvectors of $(id \times \Delta )\circ \Delta(q^{2s_3})$ and $c_4=(id \times \Delta )\circ \Delta(C)$ associated with eigenvalues $q^{2j_1+2j_2+2j_3 -2(k+\ell)}$ and $q^{2j_1+2j_2+2j_3 -2\ell + 1} + q^{-2j_1-2j_2-2j_3 +2\ell - 1}$ respectively. Then
\beqa
V(j_1)\otimes V(j_2) \otimes V(j_3) = \bigoplus_{(\ell,k)\in \Sigma} V_k(\ell) \ .
\eeqa
The irreducible finite dimensional representation  $(\bar\pi,V_k(\ell))$ is such that
\beqa
\dim(V_k(\ell))=\min(2j_1,\ell)+\min(2j_2,\ell)+\min(2j_3,\ell)-2\ell + 1\ .
\eeqa
For all $(\ell,k)\in \Sigma$,  $\bar\pi(\tA),\bar\pi(\tA^*)$ is a Leonard pair. There exists two bases on which $\bar \pi(\tA), \bar \pi(\tA^*)$ act as (\ref{tridAstar}), (\ref{tridAstar2}) and the spectra  take the form (\ref{st}) with (\ref{par}) with the identification
\beqa
b &=& q^{-2\min(2j_3,\ell)-2j_1-2j_2 -1 + 2\ell}\ ,\quad c= b^{-1}  \quad \quad \mbox{and}\quad  M = h - \ell + \min(2j_3,\ell) \ ,\label{parex} \\
 b^*&=& q^{-2\min(2j_1,\ell)-2j_2-2j_3 -1 + 2\ell}\ ,\quad 
 c^*={b^*}^{-1}  \quad \mbox{and} \quad   M = h - \ell + \min(2j_1, \ell) \ \nonumber
\eeqa
where $h$ is an integer such that 
\beqa
\ell - \min(2j_3,\ell) \leq h  \leq \min(2j_1,\ell) + \min(2j_2,\ell) -\ell. 
\eeqa
See  e.g. \cite[Theorem 5.4, Lemma 5.6]{Huang} for details.
\end{example}

\vspace{1mm}

\subsubsection{Gauge transformation and reference states} Our aim is to diagonalize the Heun-Askey-Wilson operator associated with (\ref{I}) on a finite dimensional vector space within the framework of the algebraic Bethe ansatz  and using the theory of Leonard pairs.  In this context, the first step is the construction of the so-called reference (or vacuum) state. Consider the $K-$ matrix with non-diagonal entries of the form (\ref{monoB}), (\ref{monoc}).  As  ${\textsf A},{\textsf A}^*$ act as a Leonard pair on ${\bcV}$, it is not possible to construct a state that is annihilated by the off-diagonal entries. Indeed, given the representation  ${\bcV}$ on which ${\textsf A},{\textsf A}^*$ act as 
$\bar\pi({\textsf A}), \bar\pi({\textsf A}^*)$, according to (i), (ii) there is no state $|\Omega\rangle$ such that $\bar\pi(\cC(u))|\Omega\rangle=0$ (and similarly for $\bar\pi(\cB(u))$) .   To circumvent this problem, an idea \cite{cao03}  is to apply a gauge transformation (parameterized by an integer $m$) to the $K-$matrices $K(u)$ with (\ref{monoA})-(\ref{monoc}) and  $K^+(u)$ given by (\ref{KP}):
\beqa
t(u)=\mbox{tr}\left(K^+(u)K(u)\right) = \mbox{tr}\left({\tilde K}^+(u|m) K(u|m)\right) \ , \label{tm}
\eeqa
such that the off-diagonal entries of ${\tilde K}(u|m)$ admit a reference state. 

The gauge transformation is built as follows. For more details, we refer the reader to   \cite{cao03} and \cite{MABAq2} for the notations used here. Let  $\epsilon=\pm 1$, $\alpha$, $\beta$ be generic complex parameters and $m$ be an integer. Introduce the covariant vectors
\ben\label{coVec}
&& |X^\epsilon(u,m)\rangle=\left(\begin{array}{c}
     \alpha q^{\epsilon m} u^{\epsilon} \\
       1
      \end{array}
\right),\quad
|Y^\epsilon(u,m)\rangle=\left(\begin{array}{c}
     \beta q^{-\epsilon m} u^{\epsilon} \\
       1
      \end{array}
\right)\een
and the contravariant vectors
\ben
\label{conVec}
&& \langle \tilde X^\epsilon(u,m)|=-\epsilon\frac{q^{-\epsilon} u^{-\epsilon}}{\gamma^\epsilon(1,m-1)}\left(\begin{array}{cc}
-1 \ & \alpha q^{\epsilon m} u^{\epsilon}
\end{array}
\right),\quad
 \langle \tilde Y^\epsilon(u,m)|=-\epsilon\frac{q^{-\epsilon} u^{-\epsilon}}{\gamma^\epsilon(1,m+1)}\left(\begin{array}{cc}
1 \ &
      - \beta q^{-\epsilon m} u^{\epsilon} 
      \end{array}
\right)
\een
where
\ben
\gamma^\epsilon(u,m)= \alpha ^{\frac{1-\epsilon }{2}} \beta ^{\frac{\epsilon +1}{2}}
   q^{-m} u -\alpha ^{\frac{\epsilon +1}{2}} \beta ^{\frac{1-\epsilon }{2}}
   q^m u^{-1} \ .\label{gam}
\een
Applying the gauge transformation to $K(u)$, one finds: 
\begin{align}
K(u|m) =\left(
\begin{array}{cc} 
 \mathscr{A}^{\epsilon}(u,m)   &   \mathscr{B}^{\epsilon}(u,m)\\
\mathscr{C}^{\epsilon}(u,m)  & \mathscr{D}^{\epsilon}(u,m)   \\
\end{array} \right) \ ,\label{Kmod}\nonumber
\end{align}
with
\ben\label{odyn}
&&\mathscr{A}^{\epsilon}(u,m)=\langle\tilde Y^{\epsilon}(u,m-2)|K(u)|X^{\epsilon}(u^{-1},m)\rangle, \label{Ae1}\\
&&\mathscr{B}^{\epsilon}(u,m)=\langle \tilde Y^{\epsilon}(u,m)|K(u)|Y^{\epsilon}(u^{-1},m)\rangle,\\
&&\mathscr{C}^{\epsilon}(u,m) =\langle\tilde X^{\epsilon}(u,m)|K(u)|X^{\epsilon}(u^{-1},m)\rangle, \label{ce1}\\
&&\mathscr{D}^{\epsilon}(u,m)=\frac{\gamma^\epsilon(1,m+1)}{\gamma^\epsilon(1,m)}\langle\tilde X^{\epsilon}(u,m+2)|K(u)|Y^{\epsilon}(u^{-1},m)\rangle-
\frac{(q-q^{-1})\gamma^\epsilon(u^{-2},m+1)}{(qu^2 - q^{-1}u^{-2})\gamma^\epsilon(1,m)} \mathscr{A}^{\epsilon}(u,m).\label{De1}
\een
In the literature, the elements above are called the `dynamical operators'. Their explicit expressions in terms of the elements $\textsf{A},\textsf{A}^*$ are reported in Appendix \ref{apA}. In terms of the dynamical operators, the transfer matrix (\ref{tm}) reads: 
\ben
&&t(u)= (q^2 u^2-q^{-2} u^{-2})\big(
a(u,m)\mathscr{A}^{\epsilon}(u,m)
+
d(u,m)\mathscr{D}^{\epsilon}(u,m)
+
b(u,m)
\mathscr{B}^{\epsilon}(u,m)
+
c(u,m)
\mathscr{C}^{\epsilon}(u,m) \big) \label{tmgauge}
\een
where
\ben
&&a(u,m)=
\frac{\alpha  u^{2 \epsilon }
   \left(\kappa  u+{\kappa^*} u^{-1}\right)+u^{\epsilon}(u^2-u^{-2})  q^{-(m+1) \epsilon } (\kappa_+-\alpha  \beta  \kappa_-
   \rho )-\beta  q^{-(2 m+2) \epsilon } \left({\kappa^*} u+\kappa  u^{-1}\right)}
{(\alpha-\beta q^{-\epsilon(2m+2)})(q u^2-q^{-1} u^{-2})}\,,\nonumber\\
&&d(u,m)=\frac{
-\beta  u^{2 \epsilon } q^{-(2 m+1) \epsilon } \left(q\kappa  
   u+q^{-1}{\kappa^*} u^{-1}\right)-u^{\epsilon } q^{-(m+1) \epsilon } \left(q^2 u^2-q^{-2} u^{-2}\right) (\kappa_+-\alpha 
   \beta  \kappa_- \rho )+\alpha  q^{-\epsilon } \left(q{\kappa^*} 
   u+q^{-1}\kappa   u^{-1}\right)}
{(\alpha-\beta q^{-\epsilon(2m+2)})(q^2 u^2-q^{-2} u^{-2})}\,,\nonumber\\
&&b(u,m)=\frac{
u^{\epsilon } \left(\alpha ^2 \kappa_- \rho  q^{(m+2) \epsilon
   }-\alpha  \epsilon  q^{\epsilon } \kappa ^{\frac{\epsilon
   +1}{2}} {\kappa^*}^{\frac{1-\epsilon }{2}}-\kappa_+ q^{-m \epsilon }\right)
}{\alpha-\beta q^{-2m\epsilon}}\,,\nonumber\\
&&c(u,m)=\frac{
u^{\epsilon } \left(-\beta ^2 \kappa_- \rho  q^{(2-3 m) \epsilon
   }+\beta  \epsilon  q^{(1-2 m)
   \epsilon }\kappa ^{\frac{\epsilon +1}{2}} {\kappa^*}^{\frac{1-\epsilon }{2}} +\kappa_+ q^{-m \epsilon }\right)
}{\alpha-\beta q^{-2m\epsilon}}\,.\nonumber
\een

According to a certain gauge transformation parametrized by the scalars $\alpha,\beta$  and $\epsilon=\pm1$, a reference state and its dual can be identified using the theory of Leonard pairs. Recall that the Leonard pair is formed by the operators  $\bar\pi({\textsf A}), \bar\pi({\textsf A}^*)$.
\begin{lem}\label{lem:g1} Let $m_0$ be an integer. If the parameters $\alpha,\beta$ are such that:
\beqa
\mbox{$(q^2-q^{-2})\chi^{-1}\alpha c^*q^{m_0}=1$ \qquad (resp. $(q^2-q^{-2})\chi^{-1}\beta c^*q^{-m_0}=1$)}\label{ab}
\eeqa
then  $|\Omega^+\rangle \equiv |\theta^*_0\rangle$   satisfies
\beqa
\bar\pi(\mathscr{C}^+(u,m_0))|\Omega^+\rangle =0\, \qquad \mbox{(resp.  $\bar\pi(\mathscr{B}^+(u,m_0))|\Omega^+\rangle  =0\,$)}.\label{cmO}
\eeqa
\end{lem}
\begin{proof} We show the first relation in (\ref{cmO}). By (\ref{ce1}), $\mathscr{C}^+(u,m_0)$ is given by (\ref{cm}). By (\ref{tridAstar2}), one has:
\beqa
&&\bar \pi\left(\frac{\alpha  q^{m_0} }{u \chi }[\textsf{A},\textsf{A}^*]_q-\frac{\chi  q^{-{m_0}} }{\alpha  \rho 
   u}[\textsf{A}^*,\textsf{A}]_q-\frac{ \left(q^2+1\right)}{q u}\textsf{A}+
   \left(\frac{1}{q u^3}+q u\right)\textsf{A}^* + g_0(u)\mathcal{I} \right)|\theta_0^*\rangle   \\
&& \qquad = \  u^{-1} const_1(u) |\theta^*_1\rangle  + u^{-1}const_0(u)|\theta^*_0\rangle \nonumber
\eeqa
where we denote
\beqa
g_0(u)&=& \frac{q^{-{m_0}-3}}{\alpha  \rho  \chi { (q^2-q^{-2})}}\Big( u^{-1} \rho q^3(qu^2+q^{-1}u^{-2})  \left(\alpha ^2 \rho  q^{2 {m_0}}-\chi ^2\right)
\nonumber\\&&\quad
\qquad \quad +u^{-1}\left(q^2+1\right)
   \left(  (\alpha ^2 \rho  q^{2 {m_0}} -  \chi ^2) q^2\omega  -   \alpha  \eta^* \chi  q^{m_0}   (q^4-1) \right)\Big)\ \nonumber
\eeqa
and
\beqa
\quad const_1(u)&=&  \left(\frac{\alpha  q^{m_0} }{\chi }(q\theta_0^*-q^{-1}\theta_1^*) -\frac{\chi  q^{-{m_0}} }{\alpha  \rho 
   }(q\theta_1^*-q^{-1}\theta_0^*) -\left(q+q^{-1}\right)\right)a^*_{0,1}\ ,   \label{c1} \\
\quad const_0(u)&=&  \left( \Big(\frac{\alpha  q^{m_0} }{ \chi }-\frac{\chi  q^{-{m_0}}}{\alpha  \rho 
   }\Big)(q-q^{-1})\theta^*_0 - (q+q^{-1})\right)  a^*_{0,0} + (qu^2+ q^{-1}u^{-2}) \theta_0^*+ ug_0(u) \label{c0}  .
\eeqa
Requiring  $\bar\pi(\mathscr{C}^+(u,{m_0}))|\Omega^+\rangle =0$,  it implies the conditions $const_1(u)=0$ and  $const_0(u)=0$ to be satisfied. First, let us  consider  $const_1(u)=0$. As $a^*_{0,1}\neq 0$ in (\ref{tridAstar2}), using (\ref{st}) for $M=0,1$,  the  equation  $const_1(u)=0$ gives a  quadratic equation in $\alpha$, which admits two solutions:
\beqa
\alpha=\frac{\chi q^{-{m_0}}}{(q^2-q^{-2})c^*}\quad \mbox{or} \quad  \alpha'=\frac{\chi q^{-{m_0}+2}}{(q^2-q^{-2})c^*}.
\eeqa
We now turn to $const_0(u)=0$.  It is found that only the first solution $\alpha$ is such that the coefficients of $u^2$ and $u^{-2}$ in  $const_0(u)$ are vanishing. For this choice of $\alpha$,  it follows:
\beqa
const_0(u)=0 \quad \Leftrightarrow \quad \left( (q-q^{-1})^2{\theta^*_0}^2 + \rho \right) a^*_{0,0} + \theta^*_0\omega + \eta^* =0\ .\label{c0} 
\eeqa
To show that the coefficient $a_{0,0}^*$ satisfies the second equation, we use the defining relations of the Askey-Wilson algebra. Indeed, by (\ref{aw2}) one has:
\beqa
\bar \pi\left(\big[\tA^*,\big[\tA^*,\tA\big]_q\big]_{q^{-1}}- \rho \,\tA-\omega \,\tA^*-\eta^*\mathcal{I}\right) |\theta_0^*\rangle = 0\ .
\eeqa
Using  (\ref{tridAstar}), (\ref{tridAstar2}), one finds that the Askey-Wilson relation (\ref{aw2}) implies (\ref{rho}) and (\ref{c0}). Thus, $const_0(u)=0$.
The second relation in (\ref{cmO}) is shown similarly.
\end{proof}

Similarly, the following lemma is shown. The proof follows the same steps as previously, so we omit it.
\begin{lem}\label{lem:g2} Let $m_0$ be an integer. If the parameters $\alpha,\beta$ are such that:
\beqa
\mbox{$(q^2-q^{-2})\chi^{-1}\alpha b q^{-{m_0} }=-1$ \qquad (resp. $(q^2-q^{-2})\chi^{-1}\beta b q^{{m_0} }=-1$)}\label{absol}
\eeqa
then  $|\Omega^-\rangle \equiv |\theta_0\rangle$  satisfies
\beqa
\bar\pi(\mathscr{C}^-(u,{m_0}))|\Omega^-\rangle =0\, \qquad \mbox{(resp.  $\bar\pi(\mathscr{B}^-(u,{m_0}))|\Omega^-\rangle =0\,$)}.
\eeqa
\end{lem}

For a proper choice of the gauge parameters $\alpha,\beta$, the action of the `diagonal' dynamical operators $\mathscr{A}^\pm(u,{m_0})$ and $\mathscr{D}^\pm(u,{m_0})$ on the reference states $|\Omega^\pm\rangle$ is easily computed.
\begin{lem}\label{lem:diagonalaction} Let $\alpha$ be fixed by (\ref{ab}) for $|\Omega^+\rangle$  or fixed by  (\ref{absol}) for  $|\Omega^-\rangle$. Then, the dynamical operators are such that:
\ben
\bar\pi(\mathscr{A}^\pm(u,{m_0}))|\Omega^\pm\rangle &=&\Lambda_1^\pm(u) |\Omega^\pm\rangle \qquad \mbox{and}\qquad 
\bar\pi(\mathscr{D}^\pm(u,m_0))|\Omega^\pm\rangle =\Lambda_2^\pm(u) |\Omega^\pm\rangle \ ,\label{actionADvac}
\een
where $\Lambda_1^\pm(u),\Lambda_2^\pm(u)$ are ratios of Laurent polynomials in the variable $u$.
\end{lem}
\begin{proof} We show $\bar\pi(\mathscr{A}^+(u,{m_0}))|\Omega^+\rangle =\Lambda_1^+(u) |\Omega^+\rangle$.   By (\ref{Ae1}), $\mathscr{A}^+(u,{m_0})$ is given by (\ref{Am}). The action of  $\mathscr{A}^+(u,{m_0})$ on the reference state $|\Omega^+\rangle \equiv |\theta_0^*\rangle $ reads
\beqa
\bar\pi(\mathscr{A}^+(u,{m_0}))|\Omega^+\rangle &=& \frac{u^2-u^{-2}}{u(\alpha q^{2m_0}-q^2\beta)}\bar \pi\left(\frac{\chi  q^{m_0} }{\rho }[\textsf{A}^*,\textsf{A}]_q-\frac{\alpha  \beta  q^{{m_0}+2} }{\chi
   }[\textsf{A},\textsf{A}^*]_q+  \left(\beta +\alpha  q^{2
   {m_0}}\right)q\textsf{A} \right. \nonumber\\
&& \qquad \qquad \qquad \qquad  \left.-\frac{ \left(\alpha  q^{2 {m_0}} +\beta  q^4 u^4\right)}{q u^2} \textsf{A}^*+ f_0(u)\mathcal{I}\right) |\theta_0^*\rangle  \ \nonumber
\eeqa
where we have denoted
\beqa
f_0(u)&=& -\frac{1}{\rho  q^3 u^2 \chi 
   \left(q^2-q^{-2}\right) \left(u^2-u^{-2}\right)}\left(     \rho q^{{m_0}+2}(\alpha\beta q^2\rho -\chi^2)(u^6q^2-u^{-2}) \right.  \nonumber\\
&& \qquad \qquad \qquad \qquad \qquad \left. + (q^2+1)q^{{m_0}}\left( \alpha \eta^* \chi q^{m_0}(q^4-1) -q^2\omega (\alpha\beta q^2\rho -\chi^2)   \right)                          \right.\nonumber\\
&& \qquad \qquad \qquad \qquad \qquad \left. - (q^2+1)q^{2}u^4\left( \beta \eta^* \chi (q^4-1) -q^{m_0}\omega (\alpha\beta q^2\rho -\chi^2)   \right) \right.\nonumber\\         
&& \qquad \qquad \qquad \qquad \qquad \left. + (q^2-1)u^2 \left(   ( \alpha\eta\chi q^{2{m_0}} -  \beta \eta \chi q^2) (1+q^2)^2 -\rho q^{{m_0}+2} (\alpha\beta q^2\rho -\chi^2)   \right) \right)\ .\nonumber\vspace{1mm}         
\eeqa
 For  the choice of gauge parameter $\alpha$ fixed by (\ref{ab}), on one hand using (\ref{tridAstar2}) one finds:
\beqa
\bar \pi\left(\frac{\chi  q^{m_0} }{\rho }[\textsf{A}^*,\textsf{A}]_q-\frac{\alpha  \beta  q^{{m_0}+2} }{\chi
   }[\textsf{A},\textsf{A}^*]_q\right) |\theta_0^*\rangle  &=&  \left(   \frac{\chi q^{m_0} + \beta b^* q^2(q^2-q^{-2})}{\rho} \right)(q-q^{-1})\theta^*_0 a^*_{0,0}    |\theta_0^*\rangle \nonumber\\
&& + \   \left(   \frac{\chi q^{m_0+1} b^* (q^2-q^{-2})}{\rho} -\beta q \right)a^*_{0,1}  |\theta_1^*\rangle  \ .\nonumber
\eeqa
On the other hand, using   (\ref{tridAstar2}) one finds:
\beqa
\bar \pi\left( \left(\beta +\alpha  q^{2m_0}\right)q\textsf{A}\right)|\theta_0^*\rangle   &=&\left(  \beta q - \chi q^{m_0+1} \frac{b^*(q^2-q^{-2})}{\rho}\right)a_{0,0}^* |\theta_0^*\rangle \nonumber\\    
&& - \ \left(   \frac{\chi q^{m_0+1} b^* (q^2-q^{-2})}{\rho} -\beta q  \right)a^*_{0,1}|\theta_1^*\rangle \ . \nonumber 
\eeqa
Combining both expressions, the off-diagonal contribution in $|\theta_1^*\rangle$ vanishes. It follows: 
\beqa
\bar\pi(\mathscr{A}^+(u,{m_0}))|\Omega^+\rangle &=& \frac{(u^2-u^{-2})(q^2-q^{-2})c^*}{u(q^{m_0}\chi-c^*\beta(q^4-1))} \left( \left(  \left( \frac{\chi q^{m_0} + \beta b^* q^2(q^2-q^{-2})}{\rho} \right)(q-q^{-1})\theta^*_0  \right. \right.\nonumber \\
&& \left.\left.  +  \left(  \beta q - \chi q^{m_0+1} \frac{b^*(q^2-q^{-2})}{\rho}\right) \right)a_{0,0}^* - \     \left(   \frac{\chi q^{m_0}}{c^*(q^2-q^{-2})} q^{-2}u^{-2} + \beta  q^2 u^2\right)q \theta^*_0+ f_0(u) \right)|\Omega^+\rangle \ . \nonumber
\eeqa
Inserting the expression of $a^*_{0,0}$  according to (\ref{c0}) and using (\ref{st}), $\Lambda_1^+(u)$ is obtained as a ratio of Laurent polynomials in $u$, depending on $b^*,c^*$ and the structure constant $\rho,\omega,\eta^*$. The other polynomial expressions  $\Lambda_1^-(u),\Lambda_2^\pm(u)$ are obtained following the same steps, so we omit the details.
\end{proof}

Without loss of generality, for a convenient choice of parametrization the eigenvalues $\Lambda_1^\pm(u),\Lambda_2^\pm(u)$ can be written in a factorized form. The proof of the following result is straightforward. Let $\mu,\mu',\nu,\nu',v$ be generic scalars. Introduce the parametrization:
\beqa
b&=& \frac{1}{2} q^{\frac{1}{2} (\nu +\nu')}e^{-\mu }q^{-2s}    \ , \quad c=\frac{1}{2} q^{\frac{1}{2} (\nu +\nu')}e^{\mu }q^{2s}    \ ,\label{par}\\
  b^*&=&     \frac{1}{2} q^{\frac{1}{2} (\nu +\nu')}e^{\mu' }q^{-2s}   \ , \quad c^*=   \frac{1}{2} q^{\frac{1}{2} (\nu +\nu')}e^{-\mu' }q^{2s} \  ,\nonumber\\
b^\diamond&=&   \frac{1}{2} q^{\frac{1}{2} (\nu +\nu')}v^2 q^{-2s}   \ , \quad c^\diamond=   \frac{1}{2} q^{\frac{1}{2} (\nu +\nu')}v^{-2} q^{2s} \ ,\nonumber
\eeqa
 and define the parameter $\chi$ as
\beqa
 \chi=-\frac{1}{2}q^\nu(q^2-q^{-2})\ .\nonumber
\eeqa
Adapting the results of \cite{Ter02a,Ter02b}  the structure constants are given by: 
\beqa
\rho&=&-\frac{1}{4}q^{\nu+\nu'}(q^2-q^{-2})^2\  ,\label{sc1}\\
 \omega&=&  \frac{1}{4}q^{\nu+\nu'}(q-q^{-1})^2 \left((v^2+v^{-2})(q^{2s+1}+q^{-2s-1}) - 4\cosh(\mu)\cosh(\mu')\right)\ ,  \label{sc2}\\
 \eta&=&  -\frac{1}{4}q^{\frac{3}{2}({\nu+\nu'})}\frac{(q^2-q^{-2})^2}{(q+q^{-1})}\left(  \cosh(\mu)(v^2+v^{-2}) - \cosh(\mu')(q^{2s+1}+q^{-2s-1})  \right)  ,\label{sc3} \\ 
 \eta^*&=&  -\frac{1}{4}q^{\frac{3}{2}({\nu+\nu'})}\frac{(q^2-q^{-2})^2}{(q+q^{-1})}\left(  \cosh(\mu')(v^2+v^{-2}) - \cosh(\mu)(q^{2s+1}+q^{-2s-1})  \right)   \ . \label{sc4}
\eeqa
Then, the eigenvalues of the dynamical operators take the factorized form:
\ben
&&\Lambda_1^\epsilon(u)=\frac{1}{2u^{\epsilon}}e^{\epsilon(\mu-\mu')/2-(\mu+\mu')/2}q^{(\nu +\nu')/2-(2s+1)}
\left( q^{2 s+1} uv^{-1}-u^{-1}v\right)\left( q^{2 s+1}uv-u^{-1}v^{-1}\right) \label{Lap}\\ && \qquad\qquad
\times
\left(u e^{\frac{1}{2} (1-\epsilon ) \left(\mu -\mu'\right)}  +u^{-1} e^{\frac{1}{2} (\epsilon +1) \left(\mu'+\mu \right)}\right)
\left(
u e^{\frac{1}{2} (1-\epsilon ) \left(\mu'+\mu
   \right)} + u^{-1} e^{\frac{1}{2} (\epsilon +1) \left(\mu'-\mu \right)}
\right)\,,\nonumber
\een
\ben
&& \Lambda_2^\epsilon(u)=
\frac{(u^2-u^{-2})}{2u^\epsilon (qu^2-q^{-1}u^{-2})}e^{\epsilon(\mu-\mu')/2-(\mu+\mu')/2}q^{(\nu +\nu')/2-(2s+ 1)}
\left( q^{2 s-1} u^{-1}v-u v^{-1}\right)
\left(q^{2 s-1} u^{-1} v^{-1}-u v\right) \nonumber
\\ && \qquad\qquad\qquad
\times
\left(q^2 u e^{\frac{1}{2} (\epsilon +1) \left(\mu'-\mu \right)}+u^{-1} e^{\frac{1}{2} (1-\epsilon ) \left(\mu'+\mu
   \right)}\right)
\left(q^2 u e^{\frac{1}{2} (\epsilon +1) \left(\mu'+\mu \right)}+u^{-1} e^{\frac{1}{2} (1-\epsilon ) \left(\mu -\mu'\right)}\right)\,.\nonumber
\een 

To summarize, given a Leonard pair there are at least\footnote{Note that if the constraints $\bar\pi(\mathscr{B}^\pm(u,m_0))|{\Omega}^\pm\rangle =0$ are satisfied  then the excited states would be created by successive actions of the dynamical operators $\bar\pi(\mathscr{C}^\pm (u,m))$ on $|{\Omega}^\pm\rangle$.} two possible choices of gauge transformation for which a reference state can be identified. In each case, the reference state is simply given by  the fundamental eigenvector of either $\bar\pi({\textsf A})$ or $\bar\pi({\textsf A}^*)$.\vspace{1mm}

\subsubsection{Eigenvalue of the Casimir $\Gamma$} On the vector space generated by successive actions of the dynamical operators $\{\mathscr{A}^{\epsilon}(u,m)$, $\mathscr{B}^{\epsilon}(u,m)$, $\mathscr{C}^{\epsilon}(u,m)$, $\mathscr{D}^{\epsilon}(u,m)\}$ on a reference state (the so-called Bethe states), the eigenvalue of the central element $\Gamma$ is computed as follows. The dynamical operators being combinations of $\textsf{A},\textsf{A}^*$, they commute with $\Gamma$. So, it is sufficient to compute  the action of $\Gamma$ on a reference state, for instance  $|\Omega^+ \rangle$. Recall $|\Omega^+ \rangle \equiv|\theta_0^*\rangle$. Applying  (\ref{Gam}) and using (\ref{tridAstar2}), one gets:
\beqa
\Gamma |\Omega^+\rangle = \Gamma_0 |\Omega^+ \rangle  +  a_{01}^{*} \Gamma_1|\theta_1^*\rangle +   a_{01}^*a_{12}^{*}\Gamma_2|\theta_2^*\rangle \ ,
\eeqa  
where $ \Gamma_0$ (  $\equiv \Gamma_0(a_{00}^*,a_{01}^*a_{10}^*)$), $\Gamma_1(\equiv \Gamma_1(a_{00}^*,a_{11}^*))$ and $\Gamma_2$ are functions of the spectral data $\theta_0^*,\theta^*_1, \theta^*_2$ and the structure constants $\rho,\omega,\eta,\eta^*$. Using (\ref{st}), (\ref{rho}), one finds $\Gamma_2=0$.  Applying (\ref{aw1}), (\ref{aw2}), to $|\theta_0^*\rangle$, one extracts the expressions of  $a_{00}^*,a_{11}^*$ as well as the product $a_{01}^*a_{10}^*$ in terms of the $\rho,\omega,\eta,\eta^*$. Inserted in $\Gamma_1(a_{00}^*,a_{11}^*)$, one finds $\Gamma_1(a_{00}^*,a_{11}^*)=0$. Similarly, inserted in the eigenvalue $\Gamma_0\equiv \Gamma_0(a_{00}^*,a_{01}^*a_{10}^*)$, using the parametrization (\ref{par}) and (\ref{sc1})-(\ref{sc4}) one gets finally:
\beqa
\Gamma |\Omega^+\rangle = \Gamma_0 |\Omega^+ \rangle  
\eeqa
with
\beqa
\Gamma_0 &=& \left(q-q^{-1}\right)^2 q^{2 (\nu +\nu'-2 s)}
\Big(
\frac{q^{4s}}{4}
\Big(
\left(q^2+q^{-2}\right)
\Big(
-(q^{-2s-1}+q^{2s+1})(v^2+v^{-2})\cosh(\mu)\cosh(\mu')
\nonumber\\
&&
\quad+\frac{1}{2}(\cosh(2\mu)+\cosh(2\mu'))
\Big)
-\cosh(2\mu)\cosh(2\mu')
\Big)
-\frac{1}{16}  \left(q^{4 s}-1\right) \left(q^2q^{4 s}-q^{-2}\right)(v^4+v^{-4})
\nonumber\\
&&
\quad+\frac{1}{16} \left(\left(1+q^4\right) q^{8 s}+\left(-q^{-4}+4 q^{-2}+2+4
   q^2 -q^4\right) q^{4 s}+q^{-4}+1	\right) 
\Big) \ .
\nonumber
\eeqa

\subsubsection{Commutation relations between the dynamical operators}
 In the following analysis, we will  need the commutation relations between the entries of $K(u|m)$. 
These relations are derived from the reflection equation (\ref{RE}). By straightforward calculations, in particular one has:
\ben
\qquad \mathscr{B}^{\epsilon}(u,m+2)\mathscr{B}^{\epsilon}(v,m) &=& \mathscr{B}^{\epsilon}(v,m+2)\mathscr{B}^{\epsilon}(u,m),\label{comBdBd} \\
\qquad \mathscr{A}^{\epsilon}(u,m+2)\mathscr{B}^{\epsilon}(v,m)&=&f(u,v)\mathscr{B}^{\epsilon}(v,m)\mathscr{A}^{\epsilon}(u,m) \label{comAdBd}\\
&& + g(u,v,m)\mathscr{B}^{\epsilon}(u,m)\mathscr{A}^{\epsilon}(v,m) + w(u,v,m)\mathscr{B}^{\epsilon}(u,m)\mathscr{D}^{\epsilon}(v,m),\nonumber  \\
 \mathscr{D}^{\epsilon}(u,m+2)\mathscr{B}^{\epsilon}(v,m)&=& h(u,v)\mathscr{B}^{\epsilon}(v,m) \mathscr{D}^{\epsilon}(u,m) 
,\label{comDdBd}\\
&&+k(u,v,m)\mathscr{B}^{\epsilon}(u,m)\mathscr{D}^{\epsilon}(v,m)+ n(u,v,m)\mathscr{B}^{\epsilon}(u,m)\mathscr{A}^{\epsilon}(v,m),\nonumber\\
\mathscr{C}^\epsilon(u,m+2)\mathscr{B}^\epsilon(v,m)&=& \mathscr{B}^\epsilon(v,m-2)\mathscr{C}^\epsilon(u,m)  \label{comcdBd} \\
&& +q(u,v,m) \mathscr{A}^\epsilon(v,m)\mathscr{D}^\epsilon(u,m)+r(u,v,m)\mathscr{A}^\epsilon(u,m)\mathscr{D}^\epsilon(v,m)
\nonumber\\
&& +s(u,v,m) \mathscr{A}^\epsilon(u,m)\mathscr{A}^\epsilon(v,m)+x(u,v,m)\mathscr{A}^\epsilon(v,m)\mathscr{A}^\epsilon(u,m)
\nonumber\\
&&+y(u,v,m) \mathscr{D}^\epsilon(u,m)\mathscr{A}^\epsilon(v,m)+z(u,v,m)\mathscr{D}^\epsilon(u,m)\mathscr{D}^\epsilon(v,m), \nonumber
\een
where the coefficients are collected in appendix \ref{Sec:coefcommut}. Using these commutation relations, the action of the entries on products of off-diagonal operators can be derived. Let  $\bar u$ denote the set of variables $\bar u = \{u_1,u_2,\dots,u_M\}$. Define the following strings of length $M$ of operators $ \mathscr{B}^{\epsilon}(u_i,m)$:
\ben
B^{\epsilon}(\bar u,m,M)&=&\mathscr{B}^{\epsilon}(u_1,m+2(M-1))\cdots \mathscr{B}^{\epsilon}(u_M,m)\, ,\label{SB}\\
B^{\epsilon}(\{u,\bar u_i\},m,M)&=&\mathscr{B}^{\epsilon}(u_1,m+2(M-1))\cdots \mathscr{B}^{\epsilon}(u,m+2(M-i)) \dots \mathscr{B}^{\epsilon}(u_M,m)\,.\label{SB2}
\een

Using the dynamical commutation relations (\ref{comBdBd}), (\ref{comAdBd}), (\ref{comDdBd}),
one shows that the action of the diagonal dynamical operators
$\{\mathscr{A}^{\epsilon}(u,m),\mathscr{D}^{\epsilon}(u,m)\}$ on the string (\ref{SB}) is given by
\ben\label{AonSB}
\qquad \mathscr{A}^{\epsilon}(u,m+2M)B^{\epsilon}(\bar u,m,M)&=&
\prod_{i=1}^Mf(u,u_i) B^{\epsilon}(\bar u,m,M)\mathscr{A}^{\epsilon}(u,m)\\
&+&
\sum_{i=1}^M g(u,u_i,m+2(M-1))\! \! \!\prod_{j=1,j\neq i}^M  \! f(u_i,u_j)   B^{\epsilon}(\{u,\bar u_i\},m,M)\mathscr{A}^{\epsilon}(u_i,m)
\nonumber
\\
&+&
\sum_{i=1}^M w(u,u_i,m+2(M-1)) \! \! \!\prod_{j=1,j\neq i}^M  \! h(u_i,u_j)   B^{\epsilon}(\{u,\bar u_i\},m,M)\mathscr{D}^{\epsilon}(u_i,m)
\nonumber
\een
and
 \ben\label{DonSB}
\qquad  \mathscr{D}^{\epsilon}(u,m+2M)B^{\epsilon}(\bar u,m,M)&=&
\prod_{i=1}^M h(u,u_i) B^{\epsilon}(\bar u,m,M)\mathscr{D}^{\epsilon}(u,m)\\\nonumber
&+&\sum_{i=1}^M k(u,u_i,m+2(M-1))\! \! \!\prod_{j=1,j\neq i}^M  \! h(u_i,u_j)   B^{\epsilon}(\{u,\bar u_i\},m,M)\mathscr{D}^{\epsilon}(u_i,m)\\
&+&\sum_{i=1}^M n(u,u_i,m+2(M-1)) \! \! \!\prod_{j=1,j\neq i}^M  \! f(u_i,u_j)   B^{\epsilon}(\{u,\bar u_i\},m,M)\mathscr{A}^{\epsilon}(u_i,m)\,.\nonumber
\een

In the next subsections, we will use extensively the above relations for various choices of gauge  parameters $\epsilon,\alpha,\beta$ in order to solve the spectral problem for the Heun-Askey-Wilson operator associated with (\ref{I}).\vspace{1mm}

\subsubsection{Action of $\bar \pi(\mathscr{B}^{\epsilon}(u,m+4s))$ on $\bar{\cal V}$} 
According to the choice of the gauge parameter $\alpha$, recall that the action of the dynamical operators on each reference state $|\Omega^\pm\rangle $ is given in Lemmas \ref{lem:g1}, \ref{lem:g2}, \ref{lem:diagonalaction}. The vector space $\bar{\cal V}$ being finite dimensional, we now study the action of the dynamical operator $\mathscr{B}^{+}(u,m_0+4s)$ (resp. $\mathscr{B}^{-}(u,m_0+4s)$) on the eigenvector $|\theta^*_{2s}\rangle $ (resp. $|\theta_{2s}\rangle $). Recall (\ref{tridAstar}), (\ref{tridAstar2}). 
\begin{lem} Assume $(q^2-q^{-2})\chi^{-1}\beta b^* q^{-m_0}=1$ (resp. $(q^2-q^{-2})\chi^{-1}\beta c q^{m_0}=-1$). Then
\beqa
\bar\pi(\mathscr{B}^{+}(u,m_0+4s))|\theta^*_{2s}\rangle  =0\ \quad (\mbox{resp.}\quad\  \bar\pi(\mathscr{B}^{-}(u,m_0+4s))|\theta_{2s}\rangle  =0)\ . \label{nulsp}
\eeqa
\end{lem}
\begin{proof}We show the first relation, the second one is derived similarly. 
Using (\ref{Bm}), acting on the eigenstate $|\theta^*_{j}\rangle$ of $\tA^*$ with eigenvalue $\theta_j^*$ one finds:
\beqa
\bar\pi(\mathscr{B}^{+}(u,m)) |\theta^*_{j}\rangle &=& \frac{\beta (u^2-u^{-2})}{(\alpha q^{2m+2}-\beta)u}\left(  B_{j j+1}(m) |\theta^*_{j+1}\rangle  +    B_{j j}(U,m)  |\theta^*_{j}\rangle  +  B_{jj-1}(m) |\theta^*_{j-1}\rangle\right) \label{Bplusm}
\eeqa
where $B_{0 -1}= B_{2s 2s+1}=0$,
\beqa
 B_{j j+1}(m) &=& \left(  \frac{\chi q^m}{\beta \rho}(q\theta_{j+1}^* - q^{-1}\theta_{j}^*) -   \frac{\beta q^{-m}}{\chi}(q\theta_{j}^* - q^{-1} \theta_{j+1}^*)  + (q+q^{-1})\right) a_{j j+1}^* \ ,\nonumber\\
 B_{j j}(U,m) &=&  \left( \left(\frac{\chi q^m}{\beta \rho} - \frac{\beta q^{-m}}{\chi}\right)(q-q^{-1})\theta_j^* + (q+q^{-1})\right)a_{j j}^*   - U (q+q^{-1}) \theta_j^*  \nonumber\\ 
&& \qquad \qquad \qquad  +  \frac{(\chi^2 q^{ m} - \beta^2\rho q^{- m})}{\beta\chi  (q-q^{-1})}\left(  U+ \frac{\omega}{\rho}   \right)  +\frac{\eta^*(q+q^{-1})}{\rho}\ ,\nonumber\\
 B_{j j-1} (m)&=&\left(  \frac{\chi q^m}{\beta \rho}(q\theta_{j-1}^* - q^{-1}\theta_{j}^*) -   \frac{\beta q^{-m}}{\chi}(q\theta_{j}^* - q^{-1}\theta_{j-1}^*)  + (q+q^{-1})\right) a_{j j-1}^* \ ,\nonumber
\eeqa
and we have denoted $U=(qu^2+q^{-1}u^{-2})/(q+q^{-1})$.
From the Askey-Wilson relations, one gets:
\beqa
a^*_{jj} = -  \frac{\omega \theta_j^* + \eta^*}{\rho + (q-q^{-1})^2(\theta_j^*)^2}\ .
\eeqa
For $j=2s$ it follows:
\beqa
\bar\pi(\mathscr{B}^{+}(u,m)) |\theta^*_{2s}\rangle &=& \frac{\beta (u^2-u^{-2})}{(\alpha q^{2m+2}-\beta)u}\left(     B_{2s 2s}(U,m) |\theta^*_{2s}\rangle +  B_{2s 2s-1}(m) |\theta^*_{2s-1}\rangle\right)\ . \nonumber
\eeqa
In particular, for $m=m_0+4s$ and $\beta$ such that $(q^2-q^{-2})\chi^{-1}\beta b^* q^{-m_0}=1$  one finds:
\beqa
B_{2s 2s}(U,m_0+4s) =0\ ,\qquad B_{2s 2s-1}(m_0+4s) =0 \ 
\eeqa
 which implies the first relation in (\ref{nulsp}).
\end{proof}
\subsection{Diagonalization of the `special' case $\kappa^*=\kappa_\pm=0$ or  $\kappa=\kappa_\pm=0$}
We start by considering two simple specializations of the Heun-Askey-Wilson operator (\ref{I}), namely:
\beqa
{\textsf I}(\kappa,0,0,0)=  \kappa\,\tA \qquad \mbox{or} \qquad {\textsf I}(0,\kappa^*,0,0)=  \kappa^*\,\tA^*\ .\label{Ispec}
\eeqa
In these cases, the spectral problem for the specialization of the Heun-Askey-Wilson operator acting on an irreducible finite dimensional vector space reduces to the spectral problem for either $\bar\pi(\textsf{A})$ or   $\bar\pi(\textsf{A}^*)$. Below, the corresponding eigenstates are written as Bethe states, and the eigenvalues are derived. As a byproduct, the two families of Bethe states provide two explicit bases for the Leonard pair $\bar\pi(\textsf{A})$, $\bar\pi(\textsf{A}^*)$, see subsection \ref{eigLP} .
 \vspace{1mm}

In the framework of the algebraic Bethe ansatz, we first need to write the two elements   $\textsf{A}$ and   $\textsf{A}^*$  in terms of the dynamical operators (\ref{Ae1})-(\ref{De1}) according to the reference state on which these operators act, either $|\Omega^-\rangle$ or   $|\Omega^+\rangle$.   From (\ref{monoA}),  (\ref{monoD}) and using (\ref{Ae1}), (\ref{De1}), according to the choice of gauge transformation one gets for instance:
\beqa
\label{defAepm}\textsf{A}&=&{\mathbb{A}}^-(u,m)  +\frac{\left(q\,u\,\bar\eta(u)+q^{-1}u^{-1}\bar\eta(u^{-1})\right)}{(u^2-u^{-2})(q^2u^2-q^{-2}u^{-2})}\ ,\\
\label{defAsepp}\textsf{A}^*&=& {\mathbb{A}}^+(u,m)  +\frac{\left(q\,u\,\bar\eta(u^{-1})+q^{-1}u^{-1}\bar\eta(u)\right)}{(u^2-u^{-2})(q^2u^2-q^{-2}u^{-2})} \, \quad \mbox{with}\quad 
\bar\eta(u)=(q+q^{-1})\rho^{-1}\left(\eta u+\eta^*u^{-1}\right)\, ,
\eeqa
where, for further convenience, we have introduced:
\beqa
&&\label{Aepm}{\mathbb{A}}^-(u,m)= \frac{u^{-1}}{(u^2-u^{-2})}\left(
\frac{1}{(qu^2-q^{-1}u^{-2})}\mathscr{A}^{-}(u,m) +
\frac{1}{(q^2u^2-q^{-2}u^{-2})}\mathscr{D}^{-}(u,m)\right)\,,\\
&&\label{Asepp}
{\mathbb{A}}^+(u,m)=  \frac{u}{(u^2-u^{-2})}\left(
\frac{1}{(qu^2-q^{-1}u^{-2})}\mathscr{A}^{+}(u,m) +
\frac{1}{(q^2u^2-q^{-2}u^{-2})}\mathscr{D}^{+}(u,m)\right)\,.
\eeqa

The spectral problem is now solved. Choose the gauge parameter $\alpha$ according to Lemma \ref{lem:g1} or  Lemma \ref{lem:g2}. For convenience, the parametrization (\ref{par})-(\ref{sc4})  is used. Recall the notations (\ref{SB}), (\ref{SB2}).
\begin{prop}\label{p31} Define
\ben\label{PsitildeA}
|\Psi_{sp,-}^M(\bar u,m_0)\rangle = \bar\pi(B^{-}(\bar u,m_0,M))|\Omega^{-}\rangle\,.\label{bas1}
\een
One has:
\beqa
 \bar\pi\left({\textsf I}(\kappa,0,0,0)\right) |\Psi_{sp,-}^M(\bar u,m_0)\rangle =  \frac{\kappa}{2} q^{\frac{1}{2} ({\nu +\nu'})}\left(e^{-\mu }q^{-2s+2 M}+e^{\mu } q^{2 s-2M}\right)  |\Psi_{sp,-}^M(\bar u,m_0)\rangle \  \label{specA}
\eeqa
where the set $\bar u$ satisfies the Bethe equations:
\beqa
\prod_{j=1,j\neq i}^M\left(\frac{b(u_i/(qu_j))b(u_iu_j)}{b(qu_i/u_j)b(q^2u_iu_j)}\right)=
\frac{\left(q e^{\mu'} u_i+q^{-1}e^{\mu} u_i^{-1}\right) \left(q e^{-\mu }  u_i+q^{-1} e^{\mu'}  u_i^{-1}\right)
b\left(q^{\frac{1}{2}-s} v u_i \right) b\left( q^{\frac{1}{2}-s} v^{-1}u_i\right)}
{\left(e^{\mu'} u_i+e^{-\mu } u_i^{-1}\right) \left(e^{\mu } u_i+e^{\mu'} u_i^{-1}\right)
b\left(q^{s+\frac{1}{2}} v u_i \right) b\left( q^{s+\frac{1}{2}} v^{-1}u_i\right)}\ \nonumber
\eeqa
for $i=1,\dots,M$.
\end{prop}
\begin{proof} Recall the notation (\ref{b}).
Using the (off-shell) relations (\ref{AonSB}), (\ref{DonSB}) and the actions (\ref{actionADvac}), one finds:
\ben
\bar\pi({\mathbb{A}}^-(u,m_0+2M))|\Psi_{sp,-}^M(\bar u,m_0)\rangle&=&{\Lambda}_{-}^M(u,\bar u)|\Psi_{sp,-}^M(\bar u,m_0)\rangle \label{Amm}\\  & - &\frac{u^{-1}b(q)}{b(u^2)} \sum_{i=1}^M \frac{E_{sp}^M(u_i,\bar u_i)}{b(uu_i^{-1})b(quu_i)} |\Psi_{sp,-}^M(\{u,\bar u_i\},m_0)\rangle \nonumber
\een
where we denote $ |\Psi_{sp,-}^M(\{u,\bar u_i\},m_0)\rangle = \bar\pi(B^{\epsilon}(\{u,\bar u_i\},m_0,M))|\Omega^{-}\rangle\,$,
\ben
\qquad \quad &&{\Lambda}_{-}^M(u,\bar u)=  \frac{u^{-1}}{(u^2-u^{-2})}\left(
\frac{1}{(qu^2-q^{-1}u^{-2})}
\prod_{j=1}^Mf(u,u_j)\Lambda_1^-(u)+
  \frac{1}{(q^2u^2-q^{-2}u^{-2})}\prod_{j=1}^Mh(u,u_j)\Lambda_2^-(u)\right)\label{LsM}
\een
with (\ref{Lap}) and
\ben
{E}_{sp}^M(u_i,\bar u_i)=-\frac{b(u_i^2)}{b(qu_i^2)}\prod_{j=1,j\neq i}^Mf(u_i,u_j)\Lambda_1^-(u_i)+\prod_{j=1,j\neq i}^Mh(u_i,u_j)\Lambda_2^-(u_i)\,\label{Esp}
\een
 for $i=1,\dots,M$.
Requiring ${E}_{sp}^M(u_i,\bar u_i)=0$, one gets the so-called Bethe equations given below (\ref{specA}) using (\ref{Lap}). To determine the eigenvalues in the r.h.s of (\ref{specA}), we proceed as follows. From (\ref{LsM}), observe that 
\beqa
\Lambda_{-}^M(u,\bar u)  +\frac{\left(q\,u\,\bar\eta(u)+q^{-1}u^{-1}\bar\eta(u^{-1})\right)}{(u^2-u^{-2})(q^2u^2-q^{-2}u^{-2})}  \label{eig}
\eeqa
is a meromorphic function in the variable $u$. To be equal to a constant (i.e. independent of $u$), we need to study the singular part of this function.
Using (\ref{Lap}) and the expressions in Appendix \ref{Sec:coefcommut}, one finds that the singular points are located at:
\beqa
u\in\{
\pm u_j ,\ \pm q^{-1}u_j^{-1},\ j=1,...,M\}\ .   
\eeqa 
The sum of all the residues at these points vanishes.  Extracting the constant part of (\ref{eig}), one obtains (\ref{specA}).
\end{proof}

Similarly, the eigenstates of the Heun-Askey-Wilson operator for the special case $\kappa=\kappa_+=\kappa_-=0$  are derived. The proof being analog to the previous case, we skip the details.
\begin{prop}\label{p33} Define 
\beqa\label{PsitildeA'}
|\Psi_{sp,+}^M(\bar u,m_0)\rangle = \bar\pi(B^{+}(\bar u,m_0,M))|\Omega^{+}\rangle\,.\label{bas2}
\eeqa
One has:
\beqa
 \bar\pi\left({\textsf I}(0,\kappa^*,0,0) \right)|\Psi_{sp,+}^M(\bar u,m_0)\rangle = 
\frac{\kappa^*}{2} q^{\frac{1}{2} ({ \nu +\nu'})}\left(e^{-\mu' }q^{2s-2 M}+e^{\mu' } q^{-2 s+2M}\right) |\Psi_{sp,+}^M(\bar u,m_0)\rangle \  \label{spm}
\eeqa
where the set $\bar u$ satisfies the Bethe equations:
\beqa
\prod_{j=1,j\neq i}^M\left(\frac{b(u_i/(qu_j))b(u_iu_j)}{b(qu_i/u_j)b(q^2u_iu_j)}\right)=
\frac{\left(q e^{-\mu} u_i+q^{-1}e^{-\mu'} u_i^{-1}\right) \left(q e^{\mu' }  u_i+q^{-1} e^{-\mu}  u_i^{-1}\right)
b\left(q^{\frac{1}{2}-s} v u_i \right) b\left( q^{\frac{1}{2}-s} v^{-1}u_i\right)}
{\left(e^{-\mu} u_i+e^{\mu' } u_i^{-1}\right) \left(e^{-\mu' } u_i+e^{-\mu} u_i^{-1}\right)
b\left(q^{s+\frac{1}{2}} v u_i \right) b\left( q^{s+\frac{1}{2}} v^{-1}u_i\right)}\nonumber\ \,
\eeqa
 for $i=1,\dots,M$.
\end{prop}
Let us point out that the spectra of the two specializations (\ref{Ispec}) do not depend on the Bethe roots. Furthermore, as expected the structure of the eigenvalues  matches with (\ref{st}) in agreement with \cite[Theorem 4.4 (case I)]{Ter03}.

\vspace{2mm}

\subsection{Diagonalization of the `diagonal' case  $\kappa\neq 0,\ \kappa^*\neq 0$ and $\kappa_\pm=0$}
This case is associated with the diagonal form of the $K^+(u)$ matrix (\ref{KP}).  The Heun-Askey-Wilson operator that we will diagonalize below using the algebraic Bethe ansatz is given by:
\beqa
{\textsf I}(\kappa,\kappa^*,0,0)=  \kappa\,\tA + \kappa^*\,\tA^*\ .\label{Idiag}
\eeqa
Similarly to the special case discussed in the previous subsection, the Heun-Askey-Wilson operator associated with (\ref{Idiag}) is first expressed in terms of the dynamical operators (\ref{Ae1})-(\ref{De1}).  If we choose the reference state to be $|\Omega^-\rangle$, according to the corresponding gauge transformation the element $\tA$ is given by (\ref{defAepm}) whereas  $\tA^*$ is given by:
\beqa\label{defAsepm}
\tA^*=\tilde{\mathbb{A}}^-(u,m) +\frac{\left(q\,u\,\bar\eta(u^{-1})+q^{-1}u^{-1}\bar\eta(u)\right)}{(u^2-u^{-2})(q^2u^2-q^{-2}u^{-2})} 
\eeqa
with
\beqa
\label{Asepm}\tilde{\mathbb{A}}^-(u,m)&=&\frac{u^{-1}}{(u^2-u^{-2})}     \left( \frac{\gamma^-\left(q^{-1} u^{-2},m  \right)}
{(qu^2-q^{-1}u^{-2}) \gamma^-(1,m+1)} 
\mathscr{A}^{-}(u,m) +
\frac{\gamma^- \left(q u^2,m\right)}{ (q^2u^2-q^{-2}u^{-2})  \gamma^-(1,m+1)}
\mathscr{D}^{-}(u,m)\right.
\nonumber\\
&&\qquad\qquad\qquad   \left.+
\frac{\alpha  q^{-m-1}}
{ \gamma^-(1,m)} \mathscr{B}^{-}(u,m)- \frac{\beta  q^{m-1}} {\gamma^- (1,m)} \mathscr{C}^{-}(u,m)\right)\ \nonumber\,
\eeqa
where the notation (\ref{gam}) is used.
On the other hand,  if the reference state is $|\Omega^+\rangle$,  the element $\tA^*$ is given by (\ref{defAsepp}) whereas  $\tA$ is now given by:
\beqa\label{defAepp}
\textsf{A}&=&\tilde{\mathbb{A}}^+(u,m)  +\frac{\left(q\,u\,\bar\eta(u)+q^{-1}u^{-1}\bar\eta(u^{-1})\right)}{(u^2-u^{-2})(q^2u^2-q^{-2}u^{-2})}\ ,\
\eeqa
where
\beqa
\label{Asepm}\tilde{\mathbb{A}}^+(u,m)&=&\frac{u}{(u^2-u^{-2})}     \left( \frac{\gamma^+\left(q^{-1} u^{-2},m  \right)}
{(qu^2-q^{-1}u^{-2}) \gamma^+(1,m+1)} 
\mathscr{A}^{+}(u,m) +
\frac{\gamma^+ \left(q u^2,m\right)}{ (q^2u^2-q^{-2}u^{-2})  \gamma^+(1,m+1)}
\mathscr{D}^{+}(u,m)\right.
\nonumber\\
&&\qquad\qquad\qquad   \left.+
\frac{\alpha  q^{m+1}}
{ \gamma^+(1,m)} \mathscr{B}^{+}(u,m)- \frac{\beta  q^{1-m}} {\gamma^+ (1,m)} \mathscr{C}^{+}(u,m)\right)\ .\nonumber\,
\eeqa
According to the gauge transformation chosen, combining the expressions for $\tA,\tA^*$ it follows:
\ben\label{diagonalHAW}
&&  \kappa\,\tA + \kappa^*\,\tA^*=
\frac{u^{\epsilon } \left(\alpha  u^{\epsilon } \left(\kappa  u+{\kappa^*} u^{-1}\right)-\beta   q^{-(2 m+2) \epsilon }u^{-\epsilon
   } \left({\kappa^*} u+\kappa  u^{-1}\right)\right)}
{\left(u^2-u^{-2}\right) \left(q u^2-q^{-1} u^{-2}\right) \left(\alpha -\beta  q^{-(2 m+2) \epsilon }\right)}\mathscr{A}^{\epsilon}(u,m)
\\&&\quad\quad\quad\quad\quad\quad
+
\frac{q^{-\epsilon } u^{\epsilon } \left(\alpha  u^{-\epsilon } \left(q{\kappa^*} 
   u+q^{-1}\kappa   u^{-1}\right)-\beta   q^{-2 m \epsilon }u^{\epsilon } \left(q \kappa  
   u+q^{-1}{\kappa^*}  u^{-1}\right)\right)}
{\left(u^2-u^{-2}\right) \left(q^2 u^2-q^{-2} u^{-2}\right) \left(\alpha -\beta  q^{-(2 m+2) \epsilon }\right)}\mathscr{D}^{\epsilon}(u,m)
\nonumber\\&&\quad\quad\quad\quad\quad\quad
-
\frac{\epsilon\,\alpha  q^{\epsilon }  \kappa ^{\frac{\epsilon +1}{2}} {\kappa^*}^{\frac{1-\epsilon }{2}}u^{\epsilon }}
{\left(u^2-u^{-2}\right) \left(\alpha -\beta  q^{-2 m \epsilon }\right)}
\mathscr{B}^{\epsilon}(u,m)+
\frac{\epsilon\,\beta     \kappa ^{\frac{\epsilon +1}{2}} {\kappa^*}^{\frac{1-\epsilon }{2}}
   q^{-(2 m-1) \epsilon }u^{\epsilon }}
{\left(u^2-u^{-2}\right) \left(\alpha -\beta  q^{-2 m \epsilon }\right)}\mathscr{C}^{\epsilon}(u,m)\, \nonumber\\
&&\quad\quad\quad\quad\quad\quad + \frac{(q+q^{-1})^2}{\rho  \left(u^2-u^{-2}\right) \left(q^2 u^2-q^{-2} u^{-2}\right)}     
\left(    \eta \kappa^* + \eta^*\kappa  +  ( \eta \kappa + \eta^*\kappa^*)\left(\frac{qu^2+q^{-1}u^{-2}}{q+q^{-1}} \right)\right)\ .\nonumber
\een
Recall that (\ref{tm}) holds for any choice of gauge transformation, i.e. any choice of parameters $\epsilon,\alpha$ and $\beta$. Having fixed the gauge parameter $\alpha$ according to the choice of reference state $|\Omega^+\rangle$ or $|\Omega^-\rangle$, to simplify the analysis of the spectral problem for (\ref{Idiag}), without loss of generality we choose to fix the gauge parameter $\beta$  in order to eliminate the term $\mathscr{C}^{\epsilon}(u,m)$ in (\ref{diagonalHAW}). To this end, we set:
\ben
\beta=0\,.\nonumber
\een

A crucial ingredient for the solution of the spectral problem of (\ref{Idiag}) is the following lemma\footnote{Similar relations have been obtained  in  \cite{Cra14,MABAq1}.}. For simplicity, the proof is reported Appendix \ref{apD}.   Recall the notation (\ref{SB}), (\ref{SB2}). For a generic set of parameters $\bar u=\{u_1,u_2,...,u_{2s}\}$,  define the Bethe vector:
\beqa
|\Psi_{{d},\epsilon}^{2s}(\bar u,m_0)\rangle &=& \bar\pi(B^\epsilon(\bar u,m_0,2s))|\Omega^\epsilon\rangle\, ,\nonumber\\
 |\Psi_{{d},\epsilon}^{2s}(\{u,\bar u_i\},m_0)\rangle &=& \bar\pi(B^\epsilon(\{u,\bar u_i\},m_0,2s))|\Omega^\epsilon\rangle\, . \nonumber
\eeqa
\begin{lem}\label{prop:diagonaloffshell}
For $M=2s$  and generic $\{u,u_i\}$, one has:
\ben
&&\bar\pi(\mathscr{B}^{\epsilon}(u,m_0+4s))|\Psi_{{d},\epsilon}^{2s}(\bar u,m_0)\rangle=
\label{eq:conj1}\\
&&\quad\quad
\delta_d\frac{u^{-\epsilon}b(u^2)\prod_{k=0}^{2s}b(q^{1/2+k-s}vu)b(q^{1/2+k-s}v^{-1}u)}{\prod_{i=1}^{2s}b(uu_i^{-1})b(q^{-1}u^{-1}u_i^{-1})}|\Psi_{{d},\epsilon}^{2s}(\bar u,m_0)\rangle
\nonumber\\
&&\quad\quad
-\delta_d
\sum_{i=1}^{2s}
\frac{
u_i^{-\epsilon}b(u_i^2)\prod_{k=0}^{2s}b(q^{1/2+k-s}vu_i)b(q^{1/2+k-s}v^{-1}u_i)}
{b(uu_i^{-1}) b(q^{-1} u^{-1} u_i^{-1})\prod_{j=1,j\neq i}^{2s}b(u_iu_j^{-1})b(q^{-1}u_i^{-1}u_j^{-1})}
|\Psi_{{d},\epsilon}^{2s}(\{u,\bar u_i\},m_0)\rangle \nonumber
\een
where we denote
\ben
\delta_d =- \frac{\epsilon(-1)^{2s+1}}{2}e^{-\mu  (1-\epsilon )/2- \mu' (1+\epsilon)/2} q^{(\nu+\nu')/2-\epsilon (2 s+2) }\,.\label{deltad}
\een
\end{lem}

We now turn to the solution of the spectral problem. 
\begin{prop}\label{p34}
 For $\epsilon=\pm 1$, one has:
\beqa
 \bar\pi\left({\textsf I}(\kappa,\kappa^*,0,0) \right)|\Psi_{{d},\epsilon}^{2s}(\bar u,m_0)\rangle= \Lambda_{d,\epsilon}^{2s}|\Psi_{{d},\epsilon}^{2s}(\bar u,m_0)\rangle \ \label{specdiag}
\eeqa
with
\beqa
\Lambda_{d,+}^{2s}&=& \kappa^* \theta^*_{2s} + \kappa  e^{\mu-\mu'}b \Big(  (v^2+v^{-2})[2s]_q + 2e^{\mu'}\cosh(\mu) -q\sum_{j=1}^{2s} (qu_j^2+q^{-1}u_j^{-2})\Big) \ ,\label{eigd1}\\
\Lambda_{d,-}^{2s}&=& \kappa \theta_{2s} + \kappa^* e^{\mu'-\mu}  c^*  \Big(   (v^2+v^{-2})[2s]_q + 2e^{\mu}\cosh(\mu') -q^{-1}\sum_{j=1}^{2s} (qu_j^2+q^{-1}u_j^{-2}) \Big) \ , \label{eigd2}
\eeqa
where the set $\bar u$ satisfies the (inhomogeneous) Bethe equations:
\beqa
&&\frac{b(u_i^2)}{b(qu_i^2)}(\kappa\, u_i+\kappa^* u_i^{-1})\prod_{j=1,j\neq i}^{2s}f(u_i,u_j)\Lambda_1^\epsilon(u_i)
-
q^{-\epsilon}u_i^{-2\epsilon}(q \kappa^* u_i+ q^{-1}\kappa u_i^{-1})\prod_{j=1,j\neq i}^{2s}h(u_i,u_j)\Lambda_2^\epsilon(u_i)
\label{BAEd}\\
&&
\  + \ \epsilon (q-q^{-1})^{-1} q^\epsilon \kappa^{(1+\epsilon)/2}{\kappa^*}^{(1-\epsilon)/2}\delta_d
\frac{u_i^{-2\epsilon}b(u_i^2)\prod_{k=0}^{2s}b(q^{1/2+k-s}vu_i)b(q^{1/2+k-s}v^{-1}u_i)}{\prod_{j=1,j\neq i}^{2s}b(u_iu_j^{-1}) b(qu_iu_j)}\, =0\ \nonumber
\eeqa
for $i=1,\dots,2s$.
\end{prop}

\begin{proof} For convenience, define the element:
\ben
\mathbb{W}_{{d}}(u,m)=  \kappa\,\tA + \kappa^*\,\tA^* - \frac{(q+q^{-1})^2}{\rho  \left(u^2-u^{-2}\right) \left(q^2 u^2-q^{-2} u^{-2}\right)}     
\left(    \eta \kappa^* + \eta^*\kappa  +  ( \eta \kappa + \eta^*\kappa^*)\left(\frac{qu^2+q^{-1}u^{-2}}{q+q^{-1}} \right)\right)\ .\nonumber
\een
Explicitly, using (\ref{diagonalHAW}) for $\beta=0$ it gives:
\ben
&&\mathbb{W}_{{d}}(u,m)=
 \frac{ u^{ \epsilon } \Delta_d(u)}{\left(u^2-u^{-2}\right) \left(q u^2-q^{-1} u^{-2}\right)}\mathscr{A}^{\epsilon}(u,m)
+
 \frac{ u^\epsilon \Delta_d(q^{-1}u^{-1})}{\left(u^2-u^{-2}\right) \left(q^2 u^2-q^{-2} u^{-2}\right)}\mathscr{D}^{\epsilon}(u,m)
\nonumber\\&&\quad\quad\quad\quad\quad\quad
-
\frac{\epsilon \, q^{\epsilon }  \kappa ^{\frac{\epsilon +1}{2}} {\kappa^*}^{\frac{1-\epsilon
   }{2}}u^{\epsilon }}{u^2-u^{-2}}\mathscr{B}^{\epsilon}(u,m)\, ,\nonumber
\een
where the notation
\beqa
\Delta_d(u) =  u^{\epsilon } \left(\kappa  u+{\kappa^*} u^{-1}\right)\label{Deltad}
\eeqa
is introduced.
From the multiple actions (\ref{AonSB}), (\ref{DonSB}), Lemma \ref{lem:diagonalaction} and Lemma \ref{prop:diagonaloffshell}, it follows that the action of $\mathbb{W}_{{d}}(u,m_0+4s)$ on the Bethe vector $|\Psi_{{d},\epsilon}^{2s}(\bar u,m_0)\rangle$ is given by:
\ben
 &&\quad \mathbb{W}_{{d}}(u,m_0+4s)|\Psi_{{d},\epsilon}^{2s}(\bar u,m_0)\rangle=\Lambda_{d,\epsilon}^{2s}(u,\bar u)|\Psi_{{d},\epsilon}^{2s}(\bar u,m_0)\rangle
+
\sum_{i=1}^{2s}\frac{(q-q^{-1})u^\epsilon E_d(u_i,\bar u_i)}{(u^2-u^{-2})b(uu_i^{-1})b(quu_i)}|\Psi_{{d},\epsilon}^{2s}(\{u,\bar u_i\},m_0)\rangle \nonumber
\een
where
\ben
&&\quad \Lambda_{d,\epsilon}^{2s}(u,\bar u)=\frac{u^\epsilon \Delta_d(u)}{\left(u^2-u^{-2}\right) \left(q u^2-q^{-1} u^{-2}\right)}
\prod_{j=1}^{2s}f(u,u_j)\Lambda_1^\epsilon(u)
+\frac{ u^\epsilon \Delta_d(q^{-1}u^{-1})}{\left(u^2-u^{-2}\right) \left(q^2 u^2-q^{-2} u^{-2}\right)}
\prod_{j=1}^{2s}h(u,u_j)\Lambda_2^\epsilon(u)
\nonumber\\&&\quad\quad\quad\quad\quad
+ (-1)^{2s+1}\epsilon q^\epsilon \kappa^{(1+\epsilon)/2}{\kappa^*}^{(1-\epsilon)/2}\delta_d
\frac{\prod_{k=0}^{2s}b(q^{1/2+k-s}vu)b(q^{1/2+k-s}v^{-1}u)}{\prod_{i=1}^{2s}b(uu_i^{-1})b(quu_i)}\,\label{lambddiag}
\een
and
\ben
\qquad E_{{d}}(u_i,\bar u_i)&=& \frac{b(u_i^2)}{b(qu_i^2)} \Delta_d(u_i)\prod_{j=1,j\neq i}^{2s}f(u_i,u_j)\Lambda_1^\epsilon(u_i)-
\Delta_d(q^{-1}u_i^{-1}) \prod_{j=1,j\neq i}^{2s}h(u_i,u_j)\Lambda_2^\epsilon(u_i)
\nonumber\\
&&
+ \epsilon (q-q^{-1})^{-1} q^\epsilon \kappa^{(1+\epsilon)/2}{\kappa^*}^{(1-\epsilon)/2}\delta_d
\frac{ u_i^{-\epsilon}b(u_i^2)\prod_{k=0}^{2s}b(q^{1/2+k-s}vu_i)b(q^{1/2+k-s}v^{-1}u_i)}{\prod_{j=1,j\neq i}^{2s}b(u_iu_j^{-1})b(qu_iu_j)}\,.\label{Ed}
\een

Requiring ${E}_{{d}}(u_i,\bar u_i)=0$ for $i=1,\dots,2s$, one gets the Bethe equations (\ref{BAEd}). To determine the eigenvalues $\Lambda_{d,\pm}^{2s}$ in (\ref{specdiag}), observe that 
\beqa
\Lambda_{d,\epsilon}^{2s}(u,\bar u)  + \frac{(q+q^{-1})^2}{\rho  \left(u^2-u^{-2}\right) \left(q^2 u^2-q^{-2} u^{-2}\right)}     
\left(    \eta \kappa^* + \eta^*\kappa  +  ( \eta \kappa + \eta^*\kappa^*)\left(\frac{qu^2+q^{-1}u^{-2}}{q+q^{-1}} \right)\right)\   \label{eigdiag}
\eeqa
is a meromorphic function in the variable $u$. To be equal to a constant, we study its singular part.
Using (\ref{Lap}) and the expressions in Appendix \ref{Sec:coefcommut}, one finds that the singular points are located at:%
\beqa
u\in\{
\pm u_j ,\ \pm q^{-1}u_j^{-1},\ j=1,...,M\}\ .   
\eeqa 
The sum of all the residues at these points vanishes.  Extracting the constant part of (\ref{eigdiag}), one obtains (\ref{eigd1}), (\ref{eigd2}).
\end{proof}

\subsection{Diagonalization of the generic case $\kappa\neq 0,\ \kappa^*\neq 0,\ \kappa_\pm\neq 0$} 

We now consider the most general case, associated with the non-diagonal  matrix $K^+(u)$ given by (\ref{KP}). Similarly to the special and diagonal cases, the first step is to express the Heun-Askey-Wilson operator (\ref{I}) in terms of the dynamical operators (\ref{Ae1})-(\ref{De1}). According to the gauge transformation parametrized by $\alpha,\beta$ applied in (\ref{tm}), recall the expressions (\ref{defAepm}), (\ref{defAsepp}), (\ref{defAsepm}), (\ref{defAepp}) for the elements  $\textsf{A}$ and $\textsf{A}^*$. In addition, from (\ref{monoA})-(\ref{monoc}) and (\ref{Ae1})-(\ref{De1})  one gets the following expressions for $[\textsf{A},\textsf{A}^*]_q$ and $[\textsf{A}^*,\textsf{A}]_q$ in terms of the dynamical operators:
\ben
&&[\textsf{A}^*,\textsf{A}]_q=
-\frac{\alpha\beta\rho\chi^{-1}q^{-\epsilon(m+1)}u^\epsilon}
{\alpha-q^{-2\epsilon(m+1)}\beta}
\left(\frac{1}{qu^2-q^{-1}u^{-2}}\mathscr{A}^{\epsilon}(u,m)-\frac{1}{u^2-u^{-2}}\mathscr{D}^{\epsilon}(u,m)\right)\\
&&\quad\quad\quad\quad\quad
+\frac{\rho \chi^{-1} u^\epsilon }{(\alpha-q^{-2\epsilon m}\beta)(u^2-u^{-2})}
\left(\alpha^2q^{\epsilon(m+2)}\mathscr{B}^{\epsilon}(u,m)
-\beta^2q^{\epsilon(-3m+2)}
\mathscr{C}^{\epsilon}(u,m)\right)
\nonumber\\
&&\quad\quad\quad\quad\quad
-\left(\rho\frac{qu^2+q^{-1}u^{-2}}{q^2-q^{-2}}+\frac{\omega}{q-q^{-1}}\right)
\,,
\nonumber\\
&&
[\textsf{A},\textsf{A}^*]_q=
\frac{\chi q^{-\epsilon(m+1)}u^\epsilon}
{\alpha-q^{-2\epsilon(m+1)}\beta}
\left(\frac{1}{qu^2-q^{-1}u^{-2}}\mathscr{A}^{\epsilon}(u,m)-\frac{1}{u^2-u^{-2}}\mathscr{D}^{\epsilon}(u,m)\right)\nonumber\\
&&\quad\quad\quad\quad\quad
-\frac{\chi e^{-m\epsilon}u^\epsilon }{(\alpha-q^{-2\epsilon m}\beta)(u^2-u^{-2})}
\left(\mathscr{B}^{\epsilon}(u,m)
-\mathscr{C}^{\epsilon}(u,m)\right)
\nonumber\\
&&\quad\quad\quad\quad\quad
-\left(\rho\frac{qu^2+q^{-1}u^{-2}}{q^2-q^{-2}}+\frac{\omega}{q-q^{-1}}\right)
\,.\nonumber
\een

For generic  parameters $\alpha$, $\beta$ and integer $m$, in terms of the dynamical operators recall that the transfer matrix is given by (\ref{tmgauge}). In order to simplify the analysis of the spectral problem, as a second step let us fix the gauge parameters $\alpha,\beta$ and $m$ such that the coefficients of $\mathscr{B}^{\epsilon}(u,m)$ and $\mathscr{C}^{\epsilon}(u,m)$ in (\ref{tmgauge}) vanish. For convenience and without loss of generality, let us choose the following parametrization:
\beqa
&&\kappa_+ = \frac{q^\varphi}{2}\,,\quad \kappa_- = -\frac{{q^\varphi} '}{2}\,,
\quad \kappa = \rho^{1/2}q^{(\varphi+\varphi')/2}\cosh(\xi)\,,\quad \kappa^* = \rho^{1/2}q^{(\varphi+\varphi')/2}\cosh(\xi')\,.\label{param2}
\eeqa
Then, for  the choice of gauge parameters
\ben
&&\alpha = -\epsilon \rho^{-1/2}q^{(\varphi-\varphi')/2-\epsilon(1+m_0+2M)}e^{-\xi(1+\epsilon)/2+\xi'(1-\epsilon)/2}\,,\label{gaugedl}\\
&&\beta = -\epsilon \rho^{-1/2}q^{(\varphi-\varphi')/2-\epsilon(1-m_0-2M)}e^{\xi(1+\epsilon)/2-\xi'(1-\epsilon)/2}\,,\nonumber
\een
from the expression of the transfer matrix (\ref{tmgauge}) at $m=m_0+2M$ and (\ref{trans}), for any choice of $\epsilon=\pm 1$ one finds the following expression for the Heun-Askey-Wilson element (\ref{I}):
\beqa
&& \quad {\textsf I}(\kappa,\kappa^*,\kappa_+,\kappa_-)=   \frac{1}{(u^2-u^{-2})}\left(
\frac{u^\epsilon\Delta_g(u)}{(q u^2-q^{-1} u^{-2})}
\mathscr{A}^{\epsilon}(u,m_0+2M) 
+
\frac{u^\epsilon\Delta_g(q^{-1}u^{-1})}{(q^2 u^2-q^{-2} u^{-2})}\mathscr{D}^{\epsilon}(u,m_0+2M)\right)  \label{Igen}\\
&&
\quad +\frac{q^{(\varphi+\varphi')/2}(q+q^{-1})^2}{\rho^{1/2}(u^2-u^{-2})(q^2u^2-q^{-2}u^{-2})}\left(\eta\cosh(\xi')+\eta^*\cosh(\xi)+(\eta\cosh(\xi)+\eta^*\cosh(\xi'))\frac{qu^2+q^{-1}u^{-2}}{q+q^{-1}}\right)
\nonumber\\
&& \quad
-\frac{(q^\varphi\chi^{-1}
-q^{\varphi'}\chi)}{2}\left(\rho\frac{qu^2+q^{-1}u^{-2}}{q^2-q^{-2}}+\frac{\omega}{q-q^{-1}} \right)\nonumber\ 
\eeqa
where we have denoted
\ben
\Delta_g(u)= \frac{1}{2}
e^{-\xi }  q^{(\varphi+\varphi')/2}\rho^{1/2}
(e^{-\xi'} u+e^{\xi }u^{-1})
 (e^{ \xi  (1-\epsilon )/2+ \xi' (1+\epsilon)/2} u + e^{ \xi  (1+\epsilon)/2+ \xi' (1-\epsilon )/2}u^{-1})\ . \label{Deltag}
\een

Below, we will consider the action of the dynamical operators on the reference state $|\Omega^+\rangle$ or  $|\Omega^-\rangle$.
For the choice of the gauge parameters $\alpha,\beta$ (\ref{gaugedl}), the action of the dynamical operators  $\mathscr{A}^{\epsilon}(u,m_0)$ and  $\mathscr{D}^{\epsilon}(u,m_0)$ is now considered\footnote{Similar relations were derived in \cite{MABAq3}}.
\begin{lem}\label{lem:modifiedaction} 
For the choice of gauge parameters (\ref{gaugedl}), the action of the operators $\mathscr{A}^{\epsilon}(u,m_0)$ and $\mathscr{D}^{\epsilon}(u,m_0)$ on $|\Omega^\epsilon\rangle$ is given by
\ben
&&\bar\pi(\mathscr{A}^{\epsilon}(u,m_0))|\Omega^\epsilon\rangle = \Lambda_1^\epsilon(u)|\Omega^\epsilon\rangle + c_{\epsilon}^M
\bar\pi(\mathscr{B}^{\epsilon}(u,m_0-2))|\Omega^\epsilon\rangle \,,\label{actA}\\
&&\bar\pi(\mathscr{D}^{\epsilon}(u,m_0))|\Omega^\epsilon\rangle = \Lambda_2^\epsilon(u)|\Omega^\epsilon\rangle
-\frac{(q^2u^2-q^{-2}u^{-2})}{(qu^2-q^{-1}u^{-2})} c_{\epsilon}^M \bar\pi(\mathscr{B}^{\epsilon}(u,m_0-2))|\Omega^\epsilon\rangle\,, \label{actD}
\een
where $\Lambda_{1}^\epsilon(u)$ and $\Lambda_{2}^\epsilon(u)$ are given by (\ref{Lap}) and
\beqa
\qquad c_{\epsilon}^M&=&e^{-\xi(1+\epsilon)+\xi'(1-\epsilon)}q^{-2\epsilon(2M+1)}
\frac{q^{\varphi/2}-e^{(\mu'+\xi)(1+\epsilon)/2+(\mu-\xi')(1-\epsilon)/2}q^{(\varphi'+\nu-\nu')/2+\epsilon(2M-2s+1)}\rho^{1/2}}
{q^{\varphi/2}-e^{(\mu'-\xi)(1+\epsilon)/2+(\mu+\xi')(1-\epsilon)/2}q^{(\varphi'+\nu-\nu')/2-\epsilon(2M+2s+1)}\rho^{1/2}}\ .\label{ceM}
\eeqa
\end{lem}
\begin{proof} We show (\ref{actA}) for $\epsilon=+1$. Recall $|\Omega^+\rangle\equiv |\theta^*_0\rangle $. Using the explicit expressions (\ref{Am}) and (\ref{Bm}) for $m\rightarrow m-2$ together with (\ref{tridAstar2}) it follows:
\beqa
\left( \bar\pi(\mathscr{A}^{\epsilon}(u,m_0))  - c_+^M  \bar\pi(\mathscr{B}^{\epsilon}(u,m_0-2))  \right)|\Omega^+\rangle = \frac{ u^{-1}(u^2-u^{-2})}{(\alpha q^{2m_0}-\beta q^2)}  const'_1( c_+^M)  a_{01}^*|\theta^*_1\rangle  +  eigen(u, c_+^M)|\theta^*_0\rangle \nonumber
\eeqa 
where 
\beqa
 const'_1 ( c_+^M) &=& (1-c_+^M)\frac{\chi q^{m_0}}{\rho}(q\theta_1^*-q^{-1}\theta_0^*) - \frac{(\alpha\beta q^{m_0+2}-c_+^M\beta^2q^{-m_0+4})}{\chi}(q\theta_0^*-q^{-1}\theta_1^*) \nonumber\\
&&\qquad +\  \beta q + \alpha q^{2m_0+1} - c_+^M \beta (q^3+q)\ .\nonumber
\eeqa
Requiring $ const'_1 ( c_+^M) \equiv 0$ determines  the coefficient $c_+^M$ in terms of the Leonard pair's data and the gauge parameters $\alpha,\beta$:
\beqa
c_+^M= \frac  { (\chi^2 +\alpha\beta\rho ) q^{m_0+2} \theta_1^*   -   (\chi^2 +\alpha\beta\rho q^4 ) q^{m_0} \theta_0^*   + ( \alpha q^{2m_0+2} + \beta q^2  ) \rho\chi}{  (\chi^2 q^{m_0+2}  +\beta^2\rho  q^{-m_0+4}) \theta_1^*   -   (\chi^2q^{m_0} +\beta^2\rho q^{-m_0+6}) \theta_0^*   +  (q^{4} + q^2 )\beta \rho\chi}\ .
\eeqa
 Inserting (\ref{st}) for $M=0,1$ in the above expression, for the choice of gauge parameters (\ref{gaugedl}) and using the parametrization (\ref{par}) one gets (\ref{ceM}). Then, inserting  (\ref{ceM}) in  $eigen(u, c_+^M)$, one finds  $eigen_0(u, c_+^M) = \Lambda_1^+(u)$. This completes the proof of (\ref{actA}) for $\epsilon=+1$. The other relations (\ref{actA}) for $\epsilon=-1$ and (\ref{actD}) are shown similarly, so we skip the details.
\end{proof}

As for the case of `diagonal' parameters studied in the previous section, a crucial ingredient for the solution of the spectral problem of (\ref{Igen}) is the conjecture below (see also  similar relations in \cite{MABAq2,MABAq3}) which is a generalization of Lemma \ref{prop:diagonaloffshell}. For a generic set of parameters $\bar u=\{u_1,u_2,...,u_{2s}\}$,  define the Bethe vector:
\beqa
|\Psi_{{g},\epsilon}^{2s}(\bar u,m_0)\rangle &=& \bar\pi(B^\epsilon(\bar u,m_0,2s))|\Omega^\epsilon\rangle\, ,\label{SBgen}\\
 |\Psi_{{g},\epsilon}^{2s}(\{u,\bar u_i\},m_0)\rangle &=& \bar\pi(B^\epsilon(\{u,\bar u_i\},m_0,2s))|\Omega^\epsilon\rangle\, . \label{SBgen2}
\eeqa
\begin{conj}\label{prop:genericoffshell}
For $M=2s$ and generic $\{u,u_i\}$, one has:
\ben
&&\tilde\delta_g\bar\pi(\mathscr{B}^{\epsilon}(u,m_0+4s))|\Psi_{g,\epsilon}^{2s}(\bar u,m_0)\rangle=
\label{eq:conj2}\\
&&\quad\quad
\delta_g\frac{u^{-\epsilon}b(u^2)\prod_{k=0}^{2s}b(q^{1/2+k-s}vu)b(q^{1/2+k-s}v^{-1}u)}{\prod_{i=1}^{2s}b(uu_i^{-1})b(q^{-1}u^{-1}u_i^{-1})}|\Psi_{g,\epsilon}^{2s}(\bar u,m_0)\rangle
\nonumber\\
&&\quad\quad
-\delta_g
\sum_{i=1}^{2s}
\frac{
u_i^{-\epsilon}b(u_i^2)\prod_{k=0}^{2s}b(q^{1/2+k-s}vu_i)b(q^{1/2+k-s}v^{-1}u_i)}
{b(uu_i^{-1}) b(q^{-1} u^{-1} u_i^{-1})\prod_{j=1,j\neq i}^{2s}b(u_iu_j^{-1})b(q^{-1}u_i^{-1}u_j^{-1})}
|\Psi_{g,\epsilon}^{2s}(\{u,\bar u_i\},m_0)\rangle \nonumber
\een
where 
\ben
&&\delta_g =\frac{(-1)^{2s}q^{\nu'}}{4}
(q^{\varphi/2}-e^{(\mu-\xi')(1-\epsilon)/2-(\mu'+\xi)(1+\epsilon)/2}q^{(\varphi'+\nu-\nu')/2-2s-1}\rho^{1/2})\label{deltag}\\
&&\quad\quad\times 
(q^{\varphi/2}-e^{-(\mu-\xi')(1-\epsilon)/2+(\mu'+\xi)(1+\epsilon)/2}q^{(\varphi'+\nu-\nu')/2+2s+1}\rho^{1/2})
\,,\nonumber
\een
\ben
&&\tilde\delta_g=
e^{-3\xi(1+\epsilon)/2+\xi'(1-\epsilon)/2}q^{(\varphi+\varphi')/2-(1+2\epsilon)(4s+1)}
(1-e^{\xi'(1-\epsilon)+\xi(1+\epsilon)}q^{2(4s+1)})\nonumber\\
&&\quad\quad\times
\frac{(q^{\varphi/2}-e^{(\mu-\xi')(1-\epsilon)/2+(\mu'+\xi)(1+\epsilon)/2}q^{(\varphi'+\nu-\nu')/2+\epsilon(2s+1)}\rho^{1/2})\rho^{1/2}}
{2(q^{\varphi/2}-e^{(\mu+\xi')(1-\epsilon)/2+(\mu'-\xi)(1+\epsilon)/2}q^{(\varphi'+\nu-\nu')/2-\epsilon(6s+1)}\rho^{1/2})}\,.\nonumber
\een
\end{conj}
This conjecture has been checked with Mathematica for small values of $s=1/2,1,3/2$. Note that the case $s=1/2$ has been proved in \cite{MABAq3} using the separation of variables (SoV) basis, and the method can be generalized for arbitrary $s$.
For generic $s$, it might be interesting to give a proof by analogy with Appendix \ref{apD} using the theory of Leonard pairs. This might be studied elsewhere.\vspace{1mm}

We now turn to the solution of the spectral problem.   Recall the parametrization (\ref{param2}).
\begin{prop}\label{p35}
 For $\epsilon=\pm 1$, one has:
\beqa\label{specgeneric}
 \bar\pi\left({\textsf I}(\kappa,\kappa^*,\kappa_+,\kappa_-) \right)|\Psi_{g,\epsilon}^{2s}(\bar u,m_0)\rangle= \Lambda_{g,\epsilon}^{2s}|\Psi_{g,\epsilon}^{2s}(\bar u,m_0)\rangle \ 
\eeqa
with
\beqa
&&\Lambda_{g,+}^{2s}=\kappa^*\theta_{2s}^* + \kappa^*\theta_{2s}^*|_{\mu'\rightarrow \xi}\frac{\cosh(\mu)}{\cosh(\xi')} 
+\left(-\frac{\kappa^*}{2\cosh(\xi')}\theta_{3s+1/2}^*|_{\mu'\rightarrow \mu' + \xi} +(-1)^{2s+1}\delta_g[2s+1]_q\right)(v^2+v^{-2})\nonumber\\
&&\qquad\quad-\omega\frac{(\chi^{-1}\kappa_+ + \chi\kappa_-)}{(q-q^{-1})}
+\left(\frac{\kappa^*q^{\nu+\nu'}}{4\cosh(\xi')}(q-q^{-1})(q^{2s}e^{\mu'+\xi}-q^{-2s}e^{-\mu'-\xi})-(-1)^{2s+1}\delta_g
\right)\sum_{j=1}^{2s}(qu_j^2+q^{-1}u_j^{-2})\,,\nonumber\\
&&\Lambda_{g,-}^{2s}=\kappa\theta_{2s} + \kappa\theta_{2s}|_{\mu\rightarrow -\xi'}\frac{\cosh(\mu')}{\cosh(\xi)} 
+\left(-\frac{\kappa}{2\cosh(\xi)}\theta_{3s+1/2}|_{\mu\rightarrow \mu- \xi'} +(-1)^{2s+1}\delta_g[2s+1]_q\right)(v^2+v^{-2})\nonumber\\
&&\qquad\quad-\omega\frac{(\chi^{-1}\kappa_+ + \chi\kappa_-)}{(q-q^{-1})}
+\left(\frac{\kappa}{2\cosh(\xi)}(q-q^{-1})\theta_{2s}|_{\mu\rightarrow \mu - \xi}-(-1)^{2s+1}\delta_g
\right)\sum_{j=1}^{2s}(qu_j^2+q^{-1}u_j^{-2})\,,\nonumber
\eeqa
where the set $\bar u$ satisfies the (inhomogeneous) Bethe equations: 
\beqa\label{BAEg}
&&
\frac{b(u_i^2)}{b(qu_i^2)}\frac{1}{2}
e^{-\xi }  q^{(\varphi+\varphi')/2}\rho^{1/2}
(e^{-\xi'} u_i+e^{\xi }u_i^{-1})
 (e^{ \xi  (1-\epsilon )/2+ \xi' (1+\epsilon)/2} u_i + e^{ \xi  (1+\epsilon)/2+ \xi' (1-\epsilon )/2}u_i^{-1})
\prod_{j=1,j\neq i}^{2s}f(u_i,u_j)\Lambda_1^\epsilon(u_i)
\nonumber\\&&
-
\frac{1}{2}e^{-\xi }  q^{(\varphi+\varphi')/2}\rho^{1/2}
(qe^{\xi} u_i+q^{-1}e^{-\xi'}u_i^{-1})
 (q e^{ \xi  (1+\epsilon )/2+ \xi' (1-\epsilon)/2} u_i +q^{-1} e^{\xi (1-\epsilon )/2+ \xi'  (1+\epsilon)/2}u_i^{-1})
\prod_{j=1,j\neq i}^{2s}h(u_i,u_j)\Lambda_2^\epsilon(u_i)
\nonumber\\
&&
-(-1)^{2s}\delta_g(q-q^{-1})^{-1}
\frac{u_i^{-\epsilon}b(u_i^2)\prod_{k=0}^{2s}b(q^{1/2+k-s}vu_i)b(q^{1/2+k-s}v^{-1}u_i)}{\prod_{j=1,j\neq i}^{2s}b(u_iu_j^{-1})b(qu_iu_j)}=0\,\nonumber
\eeqa
for $i=1,\dots,2s$.
\end{prop}
\begin{proof}
For convenience, define the element:
\ben
&&\mathbb{W}_g(u,m)={\textsf I}(\kappa,\kappa^*,\kappa_+,\kappa_-)+\frac{(q^\varphi\chi^{-1}
-q^{\varphi'}\chi)}{2}\left(\rho\frac{qu^2+q^{-1}u^{-2}}{q^2-q^{-2}}+\frac{\omega}{q-q^{-1}} \right)
\nonumber\\
&&\quad-\frac{ \rho^{-1/2}q^{(\varphi+\varphi')/2}(q+q^{-1})^2}{(u^2-u^{-2})(q^2u^2-q^{-2}u^{-2})}\left(\eta\cosh(\xi')+\eta^*\cosh(\xi)+(\eta\cosh(\xi)+\eta^*\cosh(\xi'))\frac{qu^2+q^{-1}u^{-2}}{q+q^{-1}}\right)\ .\nonumber
\een
Explicitly, using (\ref{Igen}) evaluated at the point $m=m_0+2M$ we have:
\ben\label{Wg}
&& \quad \mathbb{W}_g(u,m_0+2M)= \frac{1}{(u^2-u^{-2})}\left(
\frac{u^\epsilon\Delta_g(u)}{(q u^2-q^{-1} u^{-2})}
\mathscr{A}^{\epsilon}(u,m_0+2M) 
+
\frac{u^\epsilon\Delta_g(q^{-1}u^{-1})}{(q^2 u^2-q^{-2} u^{-2})}\mathscr{D}^{\epsilon}(u,m_0+2M)\right)\,. 
\een
From the multiple actions (\ref{AonSB}), (\ref{DonSB}), Lemma \ref{lem:modifiedaction} and Conjecture \ref{prop:genericoffshell}, it follows that the action of $\mathbb{W}_{g}(u,m_0+4s)$ on the Bethe vector $|\Psi_{g,\epsilon}^{2s}(\bar u,m_0)\rangle$ is given by:
\ben
&& \ \ \mathbb{W}_{g}(u,m_0+4s)|\Psi_{g,\epsilon}^{2s}(\bar u,m_0)\rangle=\Lambda_{g,\epsilon}^{2s}(u,\bar u)|\Psi_{g,\epsilon}^{2s}(\bar u,m_0)\rangle
-
\sum_{i=1}^{2s}\frac{(q+q^{-1})u^\epsilon E_g(u_i,\bar u_i)}{(u^2-u^{-2})b(uu_i^{-1})b(quu_i)}|\Psi_{g,\epsilon}^{2s}(\{u,\bar u_i\},m_0)\rangle
\een
where
\ben
&& \Lambda_{g,\epsilon}^{2s} (u,\bar u)=
\frac{ u^\epsilon\Delta_g(u)}{(u^2-u^{-2})(q u^2-q^{-1} u^{-2})}
\prod_{j=1}^{2s}f(u,u_j)\Lambda_1^\epsilon(u)
+
\frac{ u^\epsilon\Delta_g(q^{-1}u^{-1})}{(u^2-u^{-2})(q^2 u^2-q^{-2} u^{-2})}
\prod_{j=1}^{2s}h(u,u_j)\Lambda_2^\epsilon(u)
\nonumber\\
&& \qquad \qquad 
+ (-1)^{2s}\delta_g
\frac{\prod_{k=0}^{2s}b(q^{1/2+k-s}vu)b(q^{1/2+k-s}v^{-1}u)}{\prod_{i=1}^{2s}b(uu_i^{-1})b(quu_i)}\,,\label{lambdgen}
\een
and
\ben
&& E_{g}(u_i,\bar u_i)=
- \frac{b(u_i^2)}{b(qu_i^2)}\Delta_g(u_i)
\prod_{j=1,j\neq i}^{2s}f(u_i,u_j)\Lambda_1^\epsilon(u_i)
+ \Delta_g(q^{-1}u_i^{-1})
\prod_{j=1,j\neq i}^{2s}h(u_i,u_j)\Lambda_2^\epsilon(u_i)
\label{Eg}\\
&& \qquad \qquad 
+ (-1)^{2s}\delta_g(q-q^{-1})^{-1}
\frac{ u_i^{-\epsilon}b(u_i^2)\prod_{k=0}^{2s}b(q^{1/2+k-s}vu_i)b(q^{1/2+k-s}v^{-1}u_i)}{\prod_{j=1,j\neq i}^{2s}b(u_iu_j^{-1})b(qu_iu_j)}\,.\nonumber
\een

Requiring ${E}_{\textrm{g}}(u_i,\bar u_i)=0$ for $i=1,\dots,2s$,  one gets the Bethe equations below (\ref{specgeneric}). 
To determine the eigenvalues $\Lambda_{g,\epsilon}^{2s}$ in (\ref{specgeneric}), observe that 
\beqa
&&\Lambda_{g,\epsilon}^{2s}(u,\bar u)  +\frac{q^{(\varphi+\varphi')/2}(q+q^{-1})^2}{\rho^{1/2} (u^2-u^{-2})(q^2u^2-q^{-2}u^{-2})}\left(\eta\cosh(\xi')+\eta^*\cosh(\xi)+(\eta\cosh(\xi)+\eta^*\cosh(\xi'))\frac{qu^2+q^{-1}u^{-2}}{q+q^{-1}}\right)
\nonumber\\&&\quad\quad
-\frac{(q^\varphi\chi^{-1}
-q^{\varphi'}\chi)}{2}\left(\rho\frac{qu^2+q^{-1}u^{-2}}{q^2-q^{-2}}+\frac{\omega}{q-q^{-1}} \right)\,\nonumber\\&& \label{eiggeneric}
\eeqa
is a meromorphic function in the variable $u$. To be equal to a constant, we study its singular part.
The singular points are located at:
\beqa
u\in\{\pm u_j ,\ \pm q^{-1}u_j^{-1},\ j=1,...,M\}\ .   
\eeqa 
The sum of all the residues at these points vanishes.  Extracting the constant part of (\ref{eiggeneric}), one obtains the eigenvalues $\Lambda_{g,\pm}^{2s}$ below (\ref{specgeneric}).
\end{proof}

\subsection{Bethe ansatz equations and Bethe states revisited}
In this subsection, it is shown that each system of Bethe ansatz equations previously derived can be rewritten in terms of the polynomials (\ref{Polya})  in the `symmetrized' Bethe roots (\ref{sBr}), see Proposition \ref{propBAU}. Thus, solving the Bethe ansatz equations is  reduced to  finding the solutions of a system of polynomial equations instead of Laurent polynomial equations, see Proposition \ref{propP}. Also, an expansion formula for the Bethe states in the Poincar\'e-Birkhoff-Witt basis of the Askey-Wilson algebra is given, exhibiting the explicit dependence in the variables $\{U_i|i=1,...,M\}$, see Corollary \ref{propBS}. This leads to the construction of two different eigenbases for the Leonard pairs, see Proposition \ref{baseLP}.  
\subsubsection{An alternative presentation for the Bethe ansatz equations}\label{subs351}
 For the three cases studied in the previous  subsections, observe that each system of Bethe ansatz equations  in Propositions \ref{p31}, \ref{p33}, \ref{p34} and \ref{p35} enjoys the symmetries
\beqa
&& u_i \longleftrightarrow \pm q^{-1}u_i^{-1} \ ,\quad  u_i \longleftrightarrow - u_i\ \label{t1} \\
&& u_j \longleftrightarrow \pm q^{-1}u_j^{-1} \ ,\quad  u_j \longleftrightarrow - u_j\ \quad \mbox{for}\quad j\neq i \ .\label{t2}
\eeqa
This suggests that each system of Bethe equations admits an alternative presentation as a  system of equations written solely in terms of the `symmetrized' variables 
\beqa
U_i = \frac{ qu_i^2 + q^{-1}u_i^{-2}}{q+q^{-1}} \quad \mbox{with} \quad i=1,...,M, \label{sBr}
\eeqa
as we now show.  Recall the notation $\bar U_i=\{U_1,U_2,...,U_{i-1},U_{i+1}, ...,U_M\}$ with (\ref{sBr}). Define the polynomial:
\beqa
\quad P_a^M(U_i,\bar U_i) &=&  \sum_{k=0}^{M-1} \frac{(-1)^k(q+q^{-1})^{M-1}}{2^{M-1-k}} \textsf{e}_{k}(\bar U_i) \left(\sum_{l=0}^{M-1-k}\bin {M-1-k} {l}
\frac{(q-q^{-1})^{l}U_i^{M-1-k-l}}{(q^2+q^{-2})^{1+k+l-M}} g_0^{2[\frac{l}{2}]}(U_i) g_{a,\epsilon}^{( p[l])}(U_i) \right) \ \nonumber\\
&&+\bar\Delta_a H(U_i)\,\label{Polya}
\eeqa
with (\ref{bexp}), (\ref{geven}), (\ref{godd}), (\ref{defpl}), (\ref{bbp}) and
\beqa
\bar\Delta_a =  \left\{\begin{array}{c}
  0 \qquad\qquad\qquad \qquad \qquad \quad \mbox{for}\quad   a=sp , \nonumber \\
  -2(-1)^{2s}\epsilon \frac{q^{-(\nu+\nu')/2+\epsilon}}{(q-q^{-1})}{\kappa^*}^{\frac{1-\epsilon}{2}}\kappa^{\frac{1+\epsilon}{2}}\delta_d \qquad \quad \mbox{for} \quad a=d ,    \nonumber\\
2\frac{(-1)^{2s}q^{-(\nu+\nu')/2}}{(q-q^{-1})}\delta_g  \qquad \qquad \qquad \quad \mbox{for} \quad a=g  
 \end{array}  \right. \ .
\eeqa

Note that $P_a^M(U_i,\bar U_i)$ is of maximal degree $M$ in the variable $U_i$ for the case $a=sp$, and of maximal degree $2s+1$ for $a\in\{d,g\}$.
\begin{example}
For the special case $a=sp$ and $\epsilon=+$, for $M=1,2,3$, the polynomials are given by:
\ben
P_{sp}^1(U_1)&=&(q+q^{-1})s_{1-2s}U_1-r_0-s_1(v^2+v^{-2}) \,\ ,\nonumber\\
P_{sp}^2(U_1,U_2)&=&(q+q^{-1})^2 \left(s_{3-2s}U_1^2-s_{1-2s}U_1U_2\right)\nonumber\\
&&\quad
-(q+q^{-1})\left((r_2+s_3(v^2+v^{-2}))U_1 
-(r_0+s_1(v^2+v^{-2}))U_2\right)\nonumber\\
&&\quad+(q^2-q^{-2})(q^{2s}-q^{-2s})s_1+(q+q^{-1}) r_{2s+1}(v^2+v^{-2}))\,\ ,\nonumber
\een
\ben
P_{sp}^3(U_1,\{U_2,U_3\})&=&
 (q+q^{-1})^3 (s_{5-2 s} U_1^3
-s_{3-2 s} U_1^2(U_2 +U_3)+s_{1-2 s} U_1U_2 U_3)\non\\
&&+
U_1 \left(q^{-2 s-3} (q-q^{-1}) (q+q^{-1})^2 (e^{\mu'} (q^{4
   s+2}+q^{4 s+6}-2 q^6)+e^{-\mu'} (-2 q^{4 s}+q^4+1))\right.\non\\
&&\qquad \qquad \left.+(q+q^{-1})
   r_{2 s+4} (v^2+v^{-2}) \right)
\non\\ &&
+(U_2+U_3) (q+q^{-1}) (s_1
   (q^{-2 s}-q^{2 s}) (q^2-q^{-2})-r_{2 s+2}
   (v^2+v^{-2}))
\non\\ &&
-U_2 U_3 (q+q^{-1})^2
   (r_0+s_1 (v^2+v^{-2}))
\non\\&&
+(U_1 U_2+U_1 U_3) (q+q^{-1})^2
   (r_2+s_3 (v^2+v^{-2}))
\non\\&&
-U_1^2 (q+q^{-1})^2 (r_4+s_5
   (v^2+v^{-2}))
\non\\&&
-(q^2-q^{-2}) ((q^{2 s}-q^{-2 s}) r_{2s+2}
-s_1 (q^2-q^{-2}) (v^2+v^{-2})) 
\non
\een
where
\ben
&&r_j=2\cosh(\mu)(q^{2s-j}-q^{-2s+j})\,,\qquad s_j = e^{\mu'}q^j-e^{-\mu'}q^{-j}\ .\nonumber
\een
\end{example}
\begin{example}
For the diagonal case $a=d$ and $\epsilon=+$, for $M=1$ ($s=1/2$) and $M=2$ ($s=1$), the polynomials are given by: 
\ben
P_{d}^1(U_1)&=&-\frac{(q+q^{-1})^2}{(q-q^{-1})}e^{-\mu'}\kappa U_1^2
\non\\&&
+
2(q+q^{-1})(\sinh(\mu')\kappa^*+(-\cosh(\mu)+e^{-\mu'}q^{-1}(q-q^{-1})^{-1}(v^2+v^{-2}))\kappa)U_1 
\non\\&&
+
(-2 \cosh (\mu ) (q-q^{-1})- (qe^{\mu'} -q^{-1}e^{-\mu'}
   )(v^2+v^{-2}))\kappa^*
\non\\&&
-(e^{\mu'}(q-q^{-1})^2-2q^{-1}(q-q^{-1})\cosh(\mu)(v^2+v^{-2})+q^{-2}e^{-\mu'}(v^2+v^{-2})^2)(q-q^{-1})^{-1}\kappa \ ,
\nonumber\\
P_{d}^2(U_1,U_2)&=&
-\frac{(q+q^{-1})^3}{(q-q^{-1})}q^{-1}e^{-\mu'}\kappa (U_1^3+(1-q^2)U_1^2U_2)
\non\\&&
+U_2 ((-2 q^{-1}(q+q^{-1})\cosh (\mu )(v^2+v^{-2})
+q^{-1}(q-q^{-1})(q+q^{-1})^2(e^{\mu'}q-e^{-\mu'}q^{-1}))\kappa
\non\\&&\qq 
+(2(q-q^{-1})(q+q^{-1})^2\cosh(\mu)+(q+q^{-1})(e^{\mu'}q-e^{-\mu'}q^{-1})(v^2+v^{-2}))\kappa^*)
\non\\&&
+U_1U_2 (q+q^{-1})^2( (2q \cosh (\mu )-q^{-2}e^{-\mu'}(v^2+v^{-2}))\kappa- (e^{\mu'}q^{-1}-e^{-\mu'}q)\kappa^*)
\non\\&&
+ U_1^2(q+q^{-1})^2 (q^{-2} (\frac{ e^{-\mu'}
   (q+2q^{-1})}{(q-q^{-1})}(v^2+v^{-2})-2 q \cosh (\mu ))\kappa+
     (e^{\mu'} q-e^{-\mu'}q^{-1})\kappa^*)
\non\\&&
+ U_1 (  (2\cosh(\mu)q^{-3}(q+q^{-1})(v^2+v^{-2})
-\frac{e^{-\mu'} \left(1+2 q^{-2}+2 q^{-4}+q^{-6}\right)}{(q-q^{-1})}(v^4+v^{-4})
\non\\&&\qq\qq
-\frac{e^{-\mu'} (q+q^{-1}) (q^3+2 q^{-3}+3q^{-5})}{
   (q-q^{-1})})\kappa
-(q+q^{-1})
   (e^{\mu'} q^3-e^{-\mu'}q^{-3})(v^2+v^{-2})\kappa^*)
\non\\&&
+ (q^2-q^{-2}) ( (q^2-q^{-2}) (e^{\mu'}
   q-e^{-\mu'}q^{-1})-2\cosh (\mu )(v^2+v^{-2}))\kappa^*
\non\\&&
+
   (-2 q^{-1}(q^2-q^{-2})^2 \cosh (\mu )+e^{-\mu'}q^{-3}(q-q^{-1})^{-1}(v^6 +v^{-6} )
- e^{\mu'} (q^2-q^{-2})(v^2+v^{-2})
\non\\&&\qq
+e^{-\mu'} (2 q-q^{-1}+q^{-5}+q^{-7})(q-q^{-1})^{-1}(v^2 +v^{-2} ))\kappa \ .
\nonumber
\een
\end{example}

The proof of the following proposition is given in Appendix \ref{prP}.  Recall (\ref{Esp}), (\ref{Ed}) or (\ref{Eg}).
\begin{prop}\label{propBAU}
\beqa
{E}_a(u_i,\bar u_i)=\frac{u_i^{-\epsilon}b(u_i^2)q^{(\nu+\nu')/2}}{2\prod_{j\neq i}^Mb(u_i/u_j)b(qu_iu_j)}P_a^M(U_i,\bar U_i)\ \label{BAU}
\eeqa
with
\beqa
\left\{\begin{array}{c}
  M=0,1,...,2s \quad \mbox{for}\quad   a=sp , \nonumber \\
  M=2s \qquad \quad \mbox{for} \quad a\in\{d,g\}     \nonumber
 \end{array}  \right. \  .
\eeqa
\end{prop}
Recall that the Bethe ansatz equations follow from requiring
\beqa
E_a(u_i,\bar u_i)= 0 \label{Ea0}
\eeqa
for all $i=1,...,M$.
Let $\{u_1,u_2,...,u_M\}$ denote the Bethe roots that satisfy the  Bethe ansatz equations  in Propositions \ref{p31}, \ref{p33}, \ref{p34} and  \ref{p35}. According to  (\ref{BAU}) and  (\ref{Ea0}), it follows:
\begin{prop}\label{propP}
The Bethe roots $\{u_1,u_2,...,u_M\}$ are determined from (\ref{sBr}) where  $\{U_1,U_2,..., U_M\}$ are solutions of the system of polynomial equations:
\beqa
P_a^M(U_i,\bar U_i)=0  \quad \mbox{for} \quad i=1,...,M.\ \label{PBA}
\eeqa 
\end{prop}
In general, there may be solutions such that $U_i=U_j$ for $i\neq j$. However, these solutions are not compatible with (\ref{BAU}). 
In the following, a solution is called  {\it admissible} if the condition $U_i\neq U_j$ for any $i\neq j$ is fulfilled. We now study the subset of admissible solutions  $\bar U= \{U_1,U_2,..., U_M\}$ of (\ref{PBA}) for the special $a=sp$, diagonal $a=d$ and generic $a=g$ cases.\vspace{1mm}

\underline{Admissible solutions for $a=sp$}:
 For $M=1$, it is clear that $P_{sp}^1(U_1)$ has a unique solution. For $M=2$, we have to solve the set of equations
\beqa
P_{sp}^2(U_1,U_2)=0 \ , \qquad P_{sp}^2(U_2,U_1)=0.\label{P2}
\eeqa
Numerically, we find that this system admits $4=2^2$ solutions. Two solutions are such that $U_1=U_2$, which are discarded (not admissible).  The other two solutions are distinct and related by permutation. So, up to permutation, the admissible solution of (\ref{P2}) associated with  (\ref{BAU})  is unique.  For $M=3$, we have to solve the set of equations
\beqa
P_{sp}^3(U_1,\{U_2,U_3\})=0 \ , \qquad  P_{sp}^3(U_2,\{U_1,U_3\})=0 \ , \qquad  P_{sp}^3(U_3,\{U_1,U_2\})=0. \label{P3}
\eeqa
 Numerically, we find $27=3^3$ solutions. Among them, only $6=3!$ are admissible. Furthermore, they are all related by a permutation. So, up to permutation, again we conclude that the  admissible  solution of  (\ref{P3})  is unique. For $M=4$, numerically we find $256=4^4$ solutions to (\ref{PBA}). Among these, only $24=4!$ are admissible and related by permutation. Again, this manifests the fact that, up to permutation, the admissible solution of the polynomial equations  (\ref{PBA})  is unique. \vspace{1mm}

More generally, by Bezout's theorem \cite[p. 670]{Grif} the  total number of solutions of the system (\ref{PBA}) for $a=sp$ is $M^M$. Previous numerical analysis suggests that the number of admissible solutions is $M!$, all related by permutation.
Although we have no proof at the moment for $M$ generic, we formulate the following conjecture:
\begin{conj} The system of polynomial equations (\ref{PBA}) for $a=sp$ admits a unique admissible solution $\bar U= \{U_1,U_2,..., U_M\}$ up to permutation.
\end{conj}

\underline{Admissible solutions for $a=d$ and $a=g$}:
 For $M=1$ ($s=1/2$), $P_{a}^1(U_1)$ is a polynomial of degree $2$ with two distinct solutions. For $M=2$ ($s=1$), we have to solve the set of equations
\beqa
P_{a}^2(U_1,U_2)=0 \ , \qquad P_{a}^2(U_2,U_1)=0\label{Pd2}
\eeqa
where the two polynomials are of total degree $3$.
Numerically, we find  $9=3^2$ solutions. Among those, only $6$ solutions are admissible.  Up to permutation, we have $3$ distinct solutions.
 For $M=3$ ($s=3/2$), we have to solve the set of equations
\beqa
P_{a}^3(U_1,\{U_2,U_3\})=0 \ , \qquad  P_{a}^3(U_2,\{U_1,U_3\})=0 \ , \qquad  P_{a}^3(U_3,\{U_1,U_2\})=0 \label{Pd3}
\eeqa
where the three polynomials are of total degree $3$. Numerically, we find $64=4^3$ solutions. Among those, only $4$ are admissible and distinct.

For generic $M=2s$, by Bezout's theorem the total number of solutions of the system (\ref{PBA}) for $a=d$ or $a=g$ is $(2s+1)^{2s}$. Previous numerical analysis suggests that the number of admissible solutions is $2s+1$ which matches with the dimension of the vector space ${\bcV}$, as expected. We formulate the following conjecture:
\begin{conj} The system of polynomial equations (\ref{PBA}) for $a=d,g$ with $M=2s$ admits $2s+1$ distinct admissible solutions.
\end{conj}
For all three cases $a=\{sp,d,g\}$, we wish to observe that the numerical analysis of (\ref{PBA})  is simpler compared with the usual analysis of the Bethe ansatz equations in terms of the original Bethe roots $\{u_i\}$.

\subsubsection{Bethe states and the PBW basis of Askey-Wilson algebra}
In the framework of the algebraic Bethe ansatz applied to the $K-$matrix given in Proposition \ref{Kd}, by construction any Bethe state is a polynomial in the elements $\tA,\tA^*$  of the Askey-Wilson algebra acting on a  certain reference state, $|\Omega^+\rangle$ or $|\Omega^-\rangle$.  In general, this polynomial is not written in terms of linearly independent monomials in $\tA,\tA^*$.  Furthermore, according to the previous subsection, it is expected that any Bethe state can be written, up to an overall factor, in terms of the `symmetrized' variables (\ref{sBr}). 
Below, an expansion formula for any  (off-shell or on-shell) Bethe state is given in the linear basis of the Askey-Wilson algebra, and the dependency in the variables $U_i$ is exhibited.\vspace{1mm}

As a preliminary, recall that a linear basis for the Askey-Wilson algebra is known, see e.g. \cite[Theorem 4.1]{T11}. For convenience, we introduce the element  $\tB$ such that:
\beqa
\tB = \big [\tA,\tA^* \big]_q \quad ( \Rightarrow \quad \tA^* \tA = q^2 \tA\tA^* -q\tB) \ .\label{o1}
\eeqa
In terms of the elements $\tA,\tA^*,\tB$, from the Askey-Wilson relations (\ref{aw1}), (\ref{aw2}) one obtains:
\beqa
 \big[\tA,\tB \big]_{q^{-1}} &=&  \rho \,\tA^*+\omega \,\tA+\eta\mathcal{I}\quad\quad \Rightarrow\quad \tB \tA = q^{-2} \tA\tB - \rho q^{-1}\tA^* -\omega q^{-1}\tA - \eta q^{-1}\ \ ,\label{o2}\\
 \big[\tA^*,\tB \big]_q &=&  -\rho \,\tA - \omega \,\tA^* - \eta^*\mathcal{I}\quad \Rightarrow \quad \tB \tA^* = q^{2} \tA^* \tB + \rho q\tA +\omega q\tA^* + \eta^*q\ .\label{o3}
\eeqa
Thus, using the ordering relations (\ref{o1}),(\ref{o2}),(\ref{o3}), any polynomial in the elements $\tA,\tA^*,\tB$ can be written in terms of the linearly independent monomials
\beqa
\tA^i {\tA^*}^j \tB^k  \quad \mbox{for any $i,j,k\geq 0$}\label{pbw}
\eeqa
that form a Poincar\'e-Birkhoff-Witt basis (see definition in \cite[ p. 299]{Dam}) of the Askey-Wilson algebra. We refer the reader to \cite{T11} for more details.

Consider the dynamical operators  $\mathscr{B}^{\pm}(u,m)$. Recall their expressions in terms of the elements $\tA,\tA^*$ of the Askey-Wilson algebra given in Appendix \ref{apA}.  Define
\beqa
U= \frac{ qu^2 + q^{-1}u^{-2}}{q+q^{-1}}.\label{sV}
\eeqa
According to the ordering prescription (\ref{pbw}), in the PBW basis of the Askey-Wilson algebra (\ref{pbw}) we obtain (the notation (\ref{b}) is used): 
\ben
\mathscr{B}^{+}(u,m)&=\frac{\beta b(u^2) (q+q^{-1})}{u(\alpha  q^{2 m+2}-\beta)}\left(\frac{\chi q^{m+1}b(q)}{\beta\rho}\tA\tA^*  -\frac{(\chi^2q^{m+2}+\beta^2\rho q^{-m})}{(q+q^{-1})\beta\rho\chi} \tB   - U\tA^* + \tA + h^+_0(U,m) \mathcal{I}\right)\ ,\label{bp} \\
\mathscr{B}^{-}(u,m)&=\frac{\beta u b(u^2) (q+q^{-1})}{(\alpha  q^{-2 m-2}-\beta)}\left(\frac{\chi q^{-m+1}b(q)}{\beta\rho}\tA\tA^*  -\frac{(\chi^2q^{-m+2}+\beta^2\rho q^{m})}{(q+q^{-1})\beta\rho\chi} \tB   +U\tA - \tA^* + h^-_0(U,m) \mathcal{I}\right)\ \label{bm}
\een
where
\ben
h^\epsilon_0(U,m)=\frac{(\chi^2 q^{\epsilon m} - \beta^2\rho q^{-\epsilon m})}{\beta\chi(q^2-q^{-2})}\left( U+ \frac{\omega}{\rho}   \right)  + \left\{ \begin{array}{cc} \  \frac{\eta^*}{\rho}\ \mbox{for}\ \epsilon=+1\ ,\\
 -\frac{\eta}{\rho}\ \mbox{for}\ \epsilon=-1 \end{array} \right.\nonumber \ .
\een
For an arbitrary product of the dynamical operators $\mathscr{B}^{\pm}(u,m)$, it is easy to extract the general structure in terms of the ordered monomials (\ref{pbw}). For instance, consider a product of two operators $\mathscr{B}^{\epsilon}(u_1,m')$, $\mathscr{B}^{\epsilon}(u_2,m)$. In its ordered form, it reads as a polynomial of total degree four in $\tA,\tA^*$ and $\tB$ with the prescription that the power of each element is at most two. More generally, by induction it follows:
\begin{lem}\label{l1} For any integer $M\geq 1$ and any set $\bar u=\{u_1,u_2,...,u_M\}$:
\beqa
\mathscr{B}^{\epsilon}(u_1,m+2(M-1))\cdots \mathscr{B}^{\epsilon}(u_M,m) = \left(  \left(\frac{\beta}{\alpha}\right)^M\frac{ (q+q^{-1})^M q^{2M(M-m-2)}}{    (\frac{\beta}{\alpha}q^{-2m-2};q^4)_M}\prod_{l=1}^M\frac{(u_l^2-u_l^{-2})}{u_l^\epsilon}\right)\!\!\!\!\!\!\!\!\!\!\!\! \sum_{\footnotesize{ \begin{array}{cc} i,j,k  \leq M \\ i+j+k \leq 2M\end{array}}} \!\!\!\!\!\!\!\!\!\!\!\! \zeta_{i,j,k}^{[M]}(\bar U) \tA^i {\tA^*}^j \tB^k \ ,\nonumber
\eeqa
where  $\zeta_{i,j,k}^{[M]}(\bar U) $ are polynomials of total degree $M$  in the variables $U_i=\frac{qu_i^2+q^{-1}u_i^{-2}}{q+q^{-1}}$.
\end{lem}
\begin{rem} Note that both expressions (\ref{bp}), (\ref{bm}) are regular in the parameter $\beta$. Thus, for the special case $\beta=0$ the product of $\mathscr{B}^{\epsilon}(u_i,m_i)$ in  Lemma \ref{l1} is well-defined.
\end{rem}

Applying the ordering prescription (\ref{pbw}) to any product of dynamical operators $\mathscr{B}^{\epsilon}(u,m)$, the polynomials $\zeta_{i,j,k}^M(\bar u) $ can be derived recursively.
\begin{example}\label{e1} For $\epsilon=+$, the non-vanishing coefficients $\zeta_{i,j,k}^{[M]}(\bar u)$ are given by:
\ben
\mbox{For $M=1$:}\qquad \zeta^{[1]}_{110}&=&\zeta_{110}(m)=\frac{\chi q^{m+1}(q-q^{-1})}{\beta\rho},\quad \zeta^{[1]}_{001}=\zeta_{001}(m)=  -\frac{(\chi^2q^{m+2}+\beta^2\rho q^{-m})}{(q+q^{-1})\beta\rho\chi} ,\nonumber\\
    \zeta^{[1]}_{100}&=&\zeta_{100}= 1,  \quad \zeta^{[1]}_{010}(U_1)=\zeta_{010}(U_1)= -U_1  , \quad  \zeta^{[1]}_{000}(U_1)=\zeta_{000}(U_1,m)= h_0^{+}(U_1,m).\nonumber\\  
\nonumber\\
\mbox{For $M=2$:}\qquad \zeta^{[2]}_{220}&=& q^2\zeta_{110}(m-2)\zeta_{110}(m),\nonumber\\
 \zeta^{[2]}_{111}&=& \zeta_{110}(m-2)\zeta_{001}(m) + \zeta_{110}(m)\zeta_{001}(m-2)- q^3\zeta_{110}(m-2)\zeta_{110}(m)\nonumber\\
\zeta^{[2]}_{101}&=& \zeta_{100}\left(\zeta_{001}(m) + q^{-2}\zeta_{001}(m-2) -q\zeta_{110}(m-2)\right)  ,\nonumber\\
 \zeta^{[2]}_{011}&=& q^2\zeta_{001}(m-2)\zeta_{010}(U_2) + \zeta_{001}(m)\zeta_{010}(U_1) -q^3 \zeta_{110}(m)\zeta_{010}(U_1),\nonumber 
\een

\ben
 \zeta^{[2]}_{120}&=& q^2\zeta_{110}(m) \zeta_{010}(U_1) +  \zeta_{110}(m-2) \zeta_{010}(U_2) , \nonumber\\
 \zeta^{[2]}_{210}&=& \zeta_{100}(q^2 \zeta_{110}(m-2)+\zeta_{110}(m)) , \nonumber\\
 \zeta^{[2]}_{110}&=& \zeta_{110}(m-2)\zeta_{000}(U_2,m) +\zeta_{110}(m)\zeta_{000}(U_1,m-2)\nonumber\\
&&\qquad + \zeta_{100}(q^2\zeta_{010}(U_1)+ \zeta_{010}(U_2)) -  q^2\omega\zeta_{110}(m-2) \zeta_{110}(m)
\nonumber\\
 \zeta^{[2]}_{200}&=&  \rho \zeta_{110}(m)  (q^{-1}\zeta_{001}(m-2)  -q^2\zeta_{110}(m-2))+   \zeta_{100}^2  , \nonumber\\
 \zeta^{[2]}_{020}&=& \zeta_{010}(U_1) \zeta_{010}(U_2)- \rho q^{-1} \zeta_{001}(m-2) \zeta_{110}(m)  , \nonumber\\
 \zeta^{[2]}_{002}&=& \zeta_{001}(m-2)\zeta_{001}(m),\nonumber\\ 
 \zeta^{[2]}_{100}&=&   \rho q   \zeta_{001}(m-2) \zeta_{010}(U_2) - \rho q^2  \zeta_{110}(m) \zeta_{010}(U_1) -\eta^* q^2\zeta_{110}(m-2) \zeta_{110}(m)\nonumber\\ && -\omega q^{-1} \zeta_{001}(m-2) \zeta_{100}  +\eta^* q^{-1} \zeta_{001}(m-2) \zeta_{110}(m) +  ( \zeta_{000}(U_1,m-2) + \zeta_{000}(U_2,m)) \zeta_{100},\nonumber
\een
\ben
 \zeta^{[2]}_{010}&=&   -\rho q^{-1} \zeta_{001}(m-2) \zeta_{100} - \eta  \zeta_{001}(m-2)  \zeta_{110}(m) + \omega q \zeta_{001}(m-2) \zeta_{010}(U_2) \nonumber\\ &&   - \omega q^2 \zeta_{010}(U_1) \zeta_{110}(m)   + \zeta_{000}(U_1,m-2) \zeta_{010}(U_2)  + \zeta_{000}(U_2,m) \zeta_{010}(U_1)  ,  \nonumber\\ 
 \zeta^{[2]}_{001}&=&  - q  \zeta_{010}(U_1)\zeta_{100} + \zeta_{000}(U_1,m-2)  \zeta_{001}(m) + \zeta_{000}(U_2,m)  \zeta_{001}(m-2),\nonumber\\
 \zeta^{[2]}_{000}&=&      \zeta_{000}(U_1,m-2) \zeta_{000}(U_2,m) - \eta q^{-1} \zeta_{001}(m-2) \zeta_{100}\nonumber\\ &&  - \eta^*\left(q^2\zeta_{010}(U_1) \zeta_{110}(m) -  q \zeta_{010}(U_2)\ \zeta_{001}(m-2)\right)  \ . \nonumber
\een
\end{example}

Let ${\bcV}$ be the irreducible  finite dimensional vector space on which the Leonard pair $\bar\pi(\tA),\bar\pi(\tA^*)$  act. Denote the (off-shell or on-shell) Bethe state:
\beqa
|\Psi_{\epsilon}(\bar u,m)\rangle= \bar\pi(\mathscr{B}^{\epsilon}(u_1,m+2(M-1))\cdots \mathscr{B}^{\epsilon}(u_M,m))|\Omega^\epsilon\rangle\ .
\eeqa
For $M \leq 2s$, there are no relations besides the Askey-Wilson relations (\ref{aw1}), (\ref{aw2}), whereas for $M\geq \dim({\bcV})=2s+1$, additional relations occur from the characteristic polynomials of $\bar\pi(\tA)$ and $\bar\pi(\tA^*)$, see (\ref{polyc}). By Lemma \ref{l1}, it follows:
\begin{cor}\label{propBS} For $ 0\leq M \leq 2s$, any Bethe state admits the expansion formula: 
\beqa
&&|\Psi^M_{\epsilon}(\bar u,m)\rangle= \left(  \left(\frac{\beta}{\alpha}\right)^M\frac{ (q+q^{-1})^M q^{2M(M-m-2)}}{    (\frac{\beta}{\alpha}q^{-2m-2};q^4)_M}\prod_{l=1}^M\frac{(u_l^2-u_l^{-2})}{u_l^\epsilon}\right)\!\!\!\!\!\!\!\!\!\!\!\! \sum_{\footnotesize{ \begin{array}{cc} i,j,k  \leq M \\ i+j+k \leq 2M\end{array}}} \!\!\!\!\!\!\!\!\!\!\!\! \zeta_{i,j,k}^M(\bar U) \bar\pi\left(\tA^i {\tA^*}^j \tB^k\right)|\Omega^\epsilon\rangle \ .\label{expf}
\eeqa
\end{cor}

\subsubsection{Eigenbases for Leonard pairs}\label{eigLP}
The results of previous subsections suggest  that the two families of on-shell Bethe states (\ref{bas1}) and  (\ref{bas2}) for $M=0,1,2,...,2s$, provide two explicit examples of bases on which the Leonard pair $\bar\pi(\tA),\bar\pi(\tA^*)$ act as (\ref{tridAstar}),  (\ref{tridAstar2}). Recall that solving the  system of Bethe equations in Proposition \ref{p31} or \ref{p33} reduces to solving (\ref{PBA}) for $a=sp$.
\begin{prop}\label{baseLP} For each $M=0,1,...,2s$, assume the system of polynomial equations  (\ref{PBA}) for $a=sp$ admits at least one admissible solution $\overline{U}=\{U_1,...,U_M\}$ with (\ref{sBr}).  Given the Leonard pair $\bar\pi(\tA),\bar\pi(\tA^*)$ defined in subsection \ref{ss31},
the set of Bethe states (\ref{bas1}) (resp. (\ref{bas2})) with $M=0,1,...,2s$, forms a basis of the vector space ${\bcV}$. 
One has:
\beqa
|\theta_M\rangle 
&\propto &  \left(\prod_{l=1}^M\frac{u_l^\epsilon}{  (u_l^2-u_l^{-2})}\right) | \Psi_{sp,-}^M(\bar u,m_0)\rangle \, \quad \mbox{for} \quad \alpha \quad \mbox{according to Lemma \ref{lem:g2}}\ ,\nonumber\\
|\theta_M^*\rangle 
&\propto &  \left( \prod_{l=1}^M\frac{u_l^\epsilon}{  (u_l^2-u_l^{-2})}\right) | \Psi_{sp,+}^M(\bar u,m_0)\rangle \,   \quad \mbox{for} \quad \alpha \quad \mbox{according to Lemma \ref{lem:g1}}\nonumber\ .
\eeqa
\end{prop}
\begin{proof} By definition,
recall that $\dim({\bcV})=2s+1$, see subsection \ref{ss31}. For $q$ not a root of unity by (\ref{st})  the ($2s+1$)-eigenvalues $\theta_M$ and $\theta_M^*$ are multiplicity-free, with $M=0,1,...,2s$. By Propositions \ref{p31} and \ref{p33}, the Bethe states (\ref{bas1}) and (\ref{bas2}) are the eigenvectors of $\bar\pi(\tA)$ and $\bar\pi(\tA^*)$, respectively. Given $M$ fixed and according to previous comments, each eigenspace has dimension one.  Also, in view of the form  (\ref{expf}) for different values of $M$ the eigenvectors in the set  $\{|\Psi_{sp,\epsilon}^M(\bar u,m_0)\rangle\}$   are linearly independent (no additional relations besides the Askey-Wilson relations). 
\end{proof}
Note that this is consistent with the numerical analysis of the polynomial equations (\ref{PBA}) for $a=sp$: up to permutation, (\ref{PBA}) admits a unique admissible solution $\{U_1,...,U_M\}$.\vspace{1mm}

Finally, let us mention that the numerical analysis for $s=1/2,1,3/2$ done in subsection \ref{subs351} suggests that it should be possible to  generalize Proposition \ref{baseLP} to the cases $a=d$ or $a=g$. Indeed, in these cases and $M=1,2,3$, it was found that the total number of distinct admissible solutions is exactly $2s+1=\dim({\bcV})$. However, for generic values of $M=2s$ the proof that the spectrum in (\ref{specdiag}) or (\ref{specgeneric}) is multiplicity-free - a key ingredient in Proposition \ref{baseLP} - is missing.

\section{Baxter T-Q relations and the Heun-Askey-Wilson $q$-difference operator}\label{s4}
In this section, homogeneous and inhomogeneous Baxter T-Q relations are deduced from the results of the previous section.  Independently,   for the special and diagonal cases  we construct a q-difference operator realization of the Heun-Askey-Wilson element (\ref{I}) acting on an infinite dimensional representation $(\pi,{\cal V})$. For each choice of parameters, $\pi({\textsf I}(\kappa,\kappa^*,0,0))$ gives a specialization of the Heun-Askey-Wilson q-difference operator
introduced in \cite[Proposition 5]{BTVZ}.  Its action on the Q-polynomial produces the T-Q relations. 
 We also  briefly comment on the similar result for the generic case. For completeness, the action of the dynamical operators on the unit is given. These results suggest an interpretation of the Q-polynomials as transition matrix coefficients.\vspace{1mm} 

\subsection{(In)homogeneous Baxter T-Q relations}

\subsubsection{Special case $\kappa^*\neq 0$ and $\kappa=\kappa_\pm=0$} For the special case $\kappa^*=\kappa_\pm=0$, $\kappa\neq 0$ and the reference state $|\Omega^-\rangle$, a functional relation for the spectrum of $\bar\pi(\tA)$ is given by (\ref{LsM}) with (\ref{Amm}), (\ref{defAepm}).  For the special case $\kappa=\kappa_\pm=0$, $\kappa^*\neq 0$ and the reference state $|\Omega^+\rangle$, the functional relation for the spectrum of $\bar\pi(\kappa^*\tA^*)$ - denoted below ${\Lambda^*}_{sp,+}^M$ - that occurs in the proof of Proposition \ref{p33} is given by: 
\ben
&&{\Lambda^*}_{sp,+}^M=  \frac{\kappa^* u}{(u^2-u^{-2})}\left(
\frac{\Lambda_1^+(u)}{(qu^2-q^{-1}u^{-2})}
\prod_{j=1}^Mf(u,u_j)+
  \frac{\Lambda_2^+(u)}{(q^2u^2-q^{-2}u^{-2})}\prod_{j=1}^Mh(u,u_j)\right)\label{LsMplus} \\
&& \qquad \qquad +\kappa^*\frac{\left(q\,u\,\bar\eta(u^{-1})+q^{-1}u^{-1}\bar\eta(u)\right)}{(u^2-u^{-2})(q^2u^2-q^{-2}u^{-2})} \ . \nonumber
\een
Introduce the q-difference operator $T_\pm$ such that $T_\pm(f(u^2))=f(q^{\pm 2}u^2)$. By elementary computations, one finds:
\beqa
\prod_{j=1}^M h(u,u_j) = \frac{T_+Q_M(U)}{Q_M(U)} \quad \mbox{and} \quad  \prod_{j=1}^M f(u,u_j) = \frac{T_-Q_M(U)}{Q_M(U)} \quad \mbox{with} \quad  Q_M(U)=\prod_{j=1}^{M}(U- U_j)\ \label{QM}
\eeqa
where the notations (\ref{sV}), (\ref{sBr}) are used. Furthermore
\beqa
\frac{\left(q\,u\,\bar\eta(u^{-1})+q^{-1}u^{-1}\bar\eta(u)\right)}{(u^2-u^{-2})(q^2u^2-q^{-2}u^{-2})}  = \frac{(q+q^{-1})^2 (\eta + \eta^* U)}{\rho(u^2-u^{-2})(q^2u^2-q^{-2}u^{-2})}\ .\nonumber
\eeqa
Replacing these expressions in (\ref{LsMplus}),  the proof of the following proposition is immediate.
\begin{prop}\label{propTQsp} The eigenvalues ${\Lambda^*}_{sp,+}^M$ of the Heun-Askey-Wilson operator $\bar\pi(\textsf{I}(0,\kappa^*,0,0))$ are given by the homogeneous Baxter T-Q relation
\ben
&&\left((u^2-u^{-2})(q^2u^2-q^{-2}u^{-2})\right){\Lambda^*}_{sp,+}^MQ_M(U)= \kappa^*u  \Lambda_2^+(u)T_+Q_M(U)  +  \kappa^*u
\Lambda_1^+(u)\frac{(q^2u^2-q^{-2}u^{-2})}{(qu^2-q^{-1}u^{-2})}
T_-Q_M(U) \nonumber
\\
&& \qquad \qquad  \qquad \qquad \qquad \qquad  \qquad \qquad \qquad \qquad  +\kappa^*\frac{(q+q^{-1})^2}{\rho} (\eta + \eta^* U)Q_M(U)\  \nonumber
\een
with (\ref{sc1}), (\ref{sc3}), (\ref{sc4}).
\end{prop}
According to the results of Section 3, the roots of the Q-polynomial are determined by (\ref{PBA}) and the spectrum is given by  
\beqa
{\Lambda^*}_{sp,+}^M = \frac{\kappa^*}{2} q^{\frac{1}{2} ({ \nu +\nu'})}\left(e^{-\mu' }q^{2s-2 M}+e^{\mu' } q^{-2 s+2M}\right)\ .\nonumber
\eeqa

\subsubsection{Special case $\kappa\neq 0$ and $\kappa^*=\kappa_\pm=0$} As mentioned previously, eq.  (\ref{LsM}) together with (\ref{Amm}), (\ref{defAepm}) determine the spectrum  of $\bar\pi(\kappa\tA)$. Following the same arguments as for the special case $\kappa^*\neq 0$ just described, an {\it homogenous} Baxter T-Q relation characterizing the spectrum of $\bar\pi(\kappa\tA)$ can be written. However, it is interesting to observe that specializing the results of the diagonal case  allows to exhibit an {\it inhomogeneous} T-Q relation  characterizing the spectrum of $\bar\pi(\kappa\tA)$ too. Consider (\ref{lambddiag}) for $\epsilon=+$ and $\kappa^*=0$. It reads:
\ben
&&\quad \Lambda_{d,+}^{2s}(u,\bar u)|_{\kappa^*=0}=\frac{\kappa u}{(u^2-u^{-2})}\left( \frac{u^2\Lambda_1^+(u)}{(q u^2-q^{-1} u^{-2})}\prod_{j=1}^{2s}f(u,u_j) +     \frac{q^{-2}u^{-2}\Lambda_2^+(u)}{(q^2 u^2-q^{-2} u^{-2})}\prod_{j=1}^{2s}h(u,u_j) \right)
\label{lambddiagspec}\\ 
&& \quad\quad\quad\quad\quad
+ (-1)^{2s+1} q \kappa\delta_d
\frac{\prod_{k=0}^{2s}b(q^{1/2+k-s}vu)b(q^{1/2+k-s}v^{-1}u)}{\prod_{i=1}^{2s}b(uu_i^{-1})b(quu_i)}\ .\nonumber
\een
Using (\ref{bbp}), 
\beqa
\prod_{ i=1}^{2s}  b(u/u_i)b(quu_i) &=& (q+q^{-1})^{2s}Q_{2s}(U)\ \nonumber
\eeqa
and
\beqa
\frac{\left(q\,u\,\bar\eta(u)+q^{-1}u^{-1}\bar\eta(u^{-1})\right)}{(u^2-u^{-2})(q^2u^2-q^{-2}u^{-2})} = \frac{(q+q^{-1})^2 (\eta^* + \eta U)}{\rho(u^2-u^{-2})(q^2u^2-q^{-2}u^{-2})}\  ,
\eeqa
it follows:
\begin{prop} The eigenvalues ${\Lambda}_{sp,+}^{2s}$ of the Heun-Askey-Wilson operator $\bar\pi(\textsf{I}(\kappa,0,0,0))$ are given by the inhomogeneous Baxter T-Q relation
\ben
&&\left((u^2-u^{-2})(q^2u^2-q^{-2}u^{-2})\right){\Lambda}_{sp,+}^MQ_{2s}(U)= \kappa q^{-2}u^{-1}  \Lambda_2^+(u)T_+Q_{2s}(U)  +  \kappa u^3
\Lambda_1^+(u)\frac{(q^2u^2-q^{-2}u^{-2})}{(qu^2-q^{-1}u^{-2})}
T_-Q_{2s}(U) \nonumber
\\
&& \qquad \qquad  \qquad \qquad \qquad \qquad  \qquad \qquad \qquad \qquad  +\kappa\frac{(q+q^{-1})^2}{\rho} (\eta^* + \eta U)Q_{2s}(U)\   \nonumber\\
&&  \qquad \qquad  \qquad \qquad \qquad \qquad  \qquad \qquad \qquad \qquad + \kappa q\delta_d\frac{(-1)^{2s+1}}{(q+q^{-1})^{2s-2}}(U^2-1) H(U) \nonumber
\een
with (\ref{sc1}), (\ref{sc3}), (\ref{sc4}), (\ref{deltad}) and (\ref{bbp}).
\end{prop}
According to the results of Section 3, the roots of the Q-polynomial are now determined by the specialization of (\ref{PBA}) for $a=d,\epsilon=+$ and $\kappa^*=0$. Also, the spectrum follows from ${\Lambda}_{d,+}^{2s}|_{\kappa^*=0}$. Contrary to (\ref{spm}), it depends on the Bethe roots satisfying the inhomogeneous Bethe ansatz equation (\ref{PBA}):
\beqa
{\Lambda}_{sp,+}^{2s} =  \frac{\kappa}{2}q^{\frac{1}{2} ({ \nu +\nu'})}\   e^{-\mu'}q^{-2s} \Big(  (v^2+v^{-2})[2s]_q + 2e^{\mu'}\cosh(\mu) -(1+q^2)\sum_{i=1}^{2s} U_i\Big) \ .\label{speckappa}
\eeqa
Combined with the fact that the spectrum of $\bar \pi(\kappa\tA)$ is always of the form (\ref{st}) with (\ref{par}), it implies that for each set of symmetrized Bethe roots satisfying (\ref{PBA}), there exists an integer $M\in\{0,1,...,2s\}$ such that the following equality holds:
\beqa
e^{-\mu}q^{-2s+2M} + e^{\mu}q^{2s-2M} =  e^{-\mu'}q^{-2s} \Big(  (v^2+v^{-2})[2s]_q + 2e^{\mu'}\cosh(\mu) -(1+q^2)\sum_{i=1}^{2s} U_i\Big)  \ . 
\eeqa
This equality has been checked numerically for small values of $s=1/2,1,3/2,2$. Given $s$ fixed, for each set of symmetrized Bethe roots the eigenvalues ${\Lambda}_{sp,+}^{2s}$ are displayed in Table \ref{table:speckappa}. The corresponding value of $M$ is given in parenthesis. 

\begin{table}[]
\begin{tabular}{|c|c|c|}	
\hline
spin $s$ & ${\Lambda}_{sp,+}^{2s}$      $(M)$                                                      & Bethe roots $\{U_1,\dots,U_{2s}\}$                                     \\ \hline
1/2  &
\begin{tabular}[c]{@{}c@{}}2.16667 $(0)$\\ 6.16667 $(1)$\end{tabular}                   & \begin{tabular}[c]{@{}c@{}}$\{0.36\}$\\ $\{0.04\}$\end{tabular}                   \\ \hline
1    & \begin{tabular}[c]{@{}c@{}}2.08333 $(0)$\\ 3.33333 $(1)$\\ 12.0833 $(2)$\end{tabular}             & \begin{tabular}[c]{@{}c@{}}$\{0.4\, -0.646529 i,0.4\, +0.646529 i\}$\\ $\{0.3\, +0.676757 i,0.3\, -0.676757 i\}$\\ $\{-0.4+0.859069 i,-0.4-0.859069 i\}$\end{tabular}             \\\hline
3/2  & \begin{tabular}[c]{@{}c@{}}2.16667 $(1)$\\ 3.04167 $(0)$\\ 6.16667 $(2)$\\ 24.0417 $(3)$\end{tabular}       & \begin{tabular}[c]{@{}c@{}}$\{0.249828,0.645086\, +1.62544 i,0.645086\, -1.62544 i\}$\\ $\{0.201668,0.529166\, -1.67841 i,0.529166\, +1.67841 i\}$\\ $\{0.0420546,0.108973\, +1.83288 i,0.108973\, -1.83288 i\}$\\ $\{-2.42474+2.13318 i,-2.42474-2.13318 i,-0.610518\}$\end{tabular}       \\\hline
2    & \begin{tabular}[c]{@{}c@{}}2.08333 $(1)$\\ 3.33333 $(2)$\\ 5.52083 $(0)$\\12.0833 $(3)$\\ 48.0208 $(4)$\end{tabular} & \begin{tabular}[c]{@{}c@{}}$\{0.299681\, +0.596429 i,0.299681\, -0.596429 i,1.22532\, -3.42503 i,1.22532\, +3.42503 i\}$\\ $\{0.214101\, +0.612971 i,0.214101\, -0.612971 i,0.910899\, -3.58263 i,0.910899\, +3.58263 i\}$\\ $\{0.0723844\, -0.630621 i,0.0723844\, +0.630621 i,0.352616\, +3.81015 i,0.352616\, -3.81015 i\}$\\ $\{-1.3877-4.22879 i,-1.3877+4.22879 i,-0.287303+0.617224 i,-0.287303-0.617224 i\}$\\ $\{-11.8331-3.15869 i,-11.8331+3.15869 i,-1.84146,-0.842389\}$
\end{tabular}\\ \hline
\end{tabular}
\vspace{2mm}

\caption{Spectrum of $\bar \pi(\kappa\tA)$ for $\kappa=\nu=\nu'=v=1,q=2,\mu=\log(1/3),\mu'=\log(1/5)$.}
\label{table:speckappa}
\end{table}

\subsubsection{Diagonal case $\kappa,\kappa^*\neq 0$ and $\kappa_\pm=0$}
Following a similar analysis, for the diagonal case it is straightforward to derive the inhomogeneous Baxter T-Q relation from (\ref{lambddiag}). It yields to:
\begin{prop}
The eigenvalues ${\Lambda}_{d,+}^{2s}$ of the Heun-Askey-Wilson operator $\bar\pi(\textsf{I}(\kappa,\kappa^*,0,0))$ are given by the inhomogeneous Baxter T-Q relation
\ben
&&\left((u^2-u^{-2})(q^2u^2-q^{-2}u^{-2})\right){\Lambda}_{d,+}^{2s}Q_{2s}(U)=  \label{LsMdiag}\\
&& \qquad \quad = u \Delta_d(q^{-1}u^{-1}) \Lambda_2^+(u)T_+Q_{2s}(U)  +   u\Delta_d(u)
\Lambda_1^+(u)\frac{(q^2u^2-q^{-2}u^{-2})}{(qu^2-q^{-1}u^{-2})}
T_-Q_{2s}(U)\nonumber
\\
&& \qquad \quad  +\frac{(q+q^{-1})^2}{\rho} (\kappa\eta^* +\kappa^*\eta + (\kappa\eta + \kappa^*\eta^*) U)Q_{2s}(U)\     + \kappa q\delta_d(-1)^{2s+1}\frac{(U^2-1)}{(q+q^{-1})^{2s-2}} H(U) \nonumber
\een
with (\ref{sc1}), (\ref{sc3}), (\ref{sc4}), (\ref{deltad}), (\ref{Deltad}) and (\ref{bbp}).
\end{prop}
According to the results of Section 3, the roots of the Q-polynomial are determined by (\ref{PBA}) for $a=d,\epsilon=+$ and the spectrum of $\bar\pi(\textsf{I}(\kappa,\kappa^*,0,0))$  is given by (\ref{eigd1}).

\subsubsection{Generic case $\kappa,\kappa^*,\kappa_\pm\neq 0$}
For the generic case it is also straightforward to derive the corresponding inhomogeneous Baxter T-Q relation from (\ref{lambdgen}):
\begin{prop}
The eigenvalues ${\Lambda}_{g,+}^{2s}$ of the Heun-Askey-Wilson operator $\bar\pi(\textsf{I}(\kappa,\kappa^*,\kappa_+,\kappa_-))$ are given by the inhomogeneous Baxter T-Q relation
\ben
&&\left((u^2-u^{-2})(q^2u^2-q^{-2}u^{-2})\right){\Lambda}_{g,+}^{2s}Q_{2s}(U)=  \label{LsMgen}\\
&& \qquad \quad = u \Delta_g(q^{-1}u^{-1}) \Lambda_2^+(u)T_+Q_{2s}(U)  +   u\Delta_g(u)
\Lambda_1^+(u)\frac{(q^2u^2-q^{-2}u^{-2})}{(qu^2-q^{-1}u^{-2})}
T_-Q_{2s}(U)\nonumber
\\
&& \qquad  \qquad + \zeta_g(U)  Q_{2s}(U)  + \delta_g(-1)^{2s}\frac{(U^2-1)}{(q+q^{-1})^{2s-2}} H(U)  \nonumber
\een
where
\beqa
\zeta_g(U) &=& (q+q^{-1})^2\left(\frac{\big(\kappa\eta^* +\kappa^*\eta + (\kappa\eta + \kappa^*\eta^*) U
\big)}{\rho} - \frac{(\kappa_+\chi^{-1} + \kappa_-\chi)}{(q-q^{-1})}(U^2-1)\big(\rho U +\omega\big)\right)\ 
\eeqa
with (\ref{sc1}), (\ref{sc3}), (\ref{sc4}), (\ref{deltag}), (\ref{Deltag}), (\ref{param2}) and (\ref{bbp}).
\end{prop}
Here the roots of the Q-polynomial are determined by (\ref{PBA}) for $a=g,\epsilon=+$ and the spectrum of $\bar\pi(\textsf{I}(\kappa,\kappa^*,\kappa_+,\kappa_-))$  is given below (\ref{specgeneric}).

\subsection{Heun-Askey-Wilson q-difference operator and the Q-polynomial}
The starting point is a realization of the Askey-Wilson algebra (\ref{aw1}), (\ref{aw2}) in terms of q-difference operators. Define the elementary q-difference operators $T_\pm$ such that $T_\pm(f(z))=f(q^{\pm 2}z)$. In Appendix \ref{appqdiff}, examples are given. Applying the invertible transformation (\ref{dirA}) to the realization  (\ref{AWo1}), (\ref{AWo2}) of the Askey-Wilson algebra given in \cite[Section 5]{Ter03}, we obtain the linear transformation denoted $\pi$: AW $\mapsto$ $\mathbb{C}[z,z^{-1}]$ such that:\\
\beqa
\pi(\tA) &=&
 q^{-1}z^{-1}\phi(z)(T_+-1) + q^{-1}z\phi(z^{-1})(T_--1)  \label{AWop1a}\\
&& \qquad + \ \frac{1}{2}q^{\frac{\nu+\nu'}{2}}e^{-\mu'}q^{2s}\left( 2e^{\mu'}\cosh(\mu) - (v^2+v^{-2}) q^{-2s-1} + q^{-1}(z+z^{-1})\right)
\ ,\nonumber\\
\pi(\tA^*)  &=&      \phi(z)(T_+ -1)  + \phi(z^{-1}) (T_- -1) + \frac{1}{2}q^{(\nu+\nu')/2} (e^{\mu'}q^{-2s} + e^{-\mu'}q^{2s}) \ \label{AWop2a}
\eeqa
where
\beqa
\phi(z)= \frac{1}{2}q^{\frac{\nu+\nu'}{2}} e^{-\mu'} q^{2s} \frac{(1+qe^{-\mu + \mu'} z)(1+qe^{\mu + \mu'} z)(1- q^{-2s}v^2 z)(1- q^{-2s}v^{-2} z)}{(1-z^2)(1-q^2z^2)}\ .\label{phidef}
\eeqa
Note that $\pi(\textsf{A}^*)$  is a specialization of  the Askey-Wilson second-order q-difference operator, see e.g.   \cite{KS} for details. By straightforward calculations, one finds  that  the corresponding structure constants in (\ref{aw1}), (\ref{aw2}) are given by (\ref{sc1})-(\ref{sc4}). Using (\ref{AWop1a}),  (\ref{AWop2a}),  the realization of the q-commutators in $\tA,\tA^*$ is also computed. For instance, one finds:
\ben
\label{qc1}\\
\pi([\tA^*,\tA]_q) = -\frac{q^{\nu+\nu'}(q-q^{-1})}{4}
\left((q+q^{-1})(z+z^{-1})-(q^{2s+1}+q^{-2s-1})(v^2+v^{-2})+4\cosh(\mu)\cosh(\mu')\right)\ .\nonumber
\een
For simplicity, the expression of $\pi([\tA,\tA^*]_q)$ is reported at the end of Appendix \ref{appqdiff}.\vspace{1mm}

In general, the image of the element (\ref{I}) by the map $\pi$ reads as a fourth order q-difference operator that generalizes the Askey-Wilson operator (\ref{AWop2a}). However, for the special and diagonal cases this expression reduces to a second-order q-difference operator.
In these cases,  it is easy to see that it is a specialization of the  Heun-Askey-Wilson operator  introduced in \cite[Proposition 5]{BTVZ}. Below, we consider the action of $\pi(\textsf{I}(\kappa,\kappa^*,0,0))$ on the Q-polynomial $Q_M(Z)$ with (\ref{QM}) where we denote:
\beqa
Z = \frac{(z+z^{-1})}{(q+q^{-1})}\ .
\eeqa
For convenience, we denote respectively by $Q_M(Z)$, $\bar Q_{2s}(Z)$ and $\tilde Q_{2s}(Z)$  the Q-polynomials associated with the symmetrized Bethe roots (\ref{sBr}) satisfying (\ref{PBA}) for $a=sp$, for  $a=d$ and $\kappa^*=0$,   for $a=d$, respectively.
\begin{lem} According to the parameters $\kappa,\kappa^*$, the action of the Heun-Askey-Wilson q-difference operator $\pi(\textsf{I}(\kappa,\kappa^*,0,0))$ on the corresponding  Q-polynomial is such that:
\beqa
\pi(\textsf{I}(0,\kappa^*,0,0)) Q_{M}(Z) &=& \Lambda_{sp,+}^{M} Q_{M}(Z) \ ,\label{act1}\\
\pi(\textsf{I}(\kappa,0,0,0)) \bar Q_{2s}(Z) &=& {\Lambda^*}^{2s}_{sp,+} \bar  Q_{2s}(Z)  -   \kappa q\delta_d\frac{(-1)^{2s+1}}{(q+q^{-1})^{2s}}H(Z) \ ,\label{act2}\\
\pi(\textsf{I}(\kappa,\kappa^*,0,0)) \tilde Q_{2s}(Z) &=& \Lambda_{d,+}^{2s} \tilde Q_{2s}(Z)  -   \kappa q\delta_d\frac{(-1)^{2s+1}}{(q+q^{-1})^{2s}}H(Z) \ .\label{act3}
\eeqa
\end{lem}
\begin{proof} Firstly, we show (\ref{act1}). Consider the Baxter T-Q relation (\ref{LsMplus}). Recall (\ref{Lap}). Observe that:
\beqa
\Lambda_1^+(u) &=& u^{-1}(u^2-u^{-2})(qu^2-q^{-1}u^{-2}) \phi(q^{-1}u^{-2})\ ,\label{L1p}\\
\Lambda_2^+(u) &=& u^{-1}(u^2-u^{-2})(q^2u^2-q^{-2}u^{-2}) \phi(qu^{2})\ \label{L2p}
\eeqa
with (\ref{phidef}). Also, one easily show that:
\beqa
 \frac{(q+q^{-1})^2 (\eta + \eta^* U)}{\rho(u^2-u^{-2})(q^2u^2-q^{-2}u^{-2})} =  \frac{1}{2} q^{\frac{1}{2} ({ \nu +\nu'})}\left(e^{-\mu' }q^{2s}+e^{\mu' } q^{-2 s}\right) -  \phi(qu^{2})-\phi(q^{-1}u^{-2})\ .\nonumber
\eeqa
Inserting these expressions in (\ref{LsMplus}),  one obtains (\ref{act1}) using (\ref{AWop2a}) through the identification $z=qu^2$. \vspace{1mm}

Secondly, we show   (\ref{act3}). Inserting (\ref{L1p}), (\ref{L2p}) and
\beqa
 \frac{(q+q^{-1})^2 (\eta^* + \eta U)}{\rho(u^2-u^{-2})(q^2u^2-q^{-2}u^{-2})} &=&  \frac{1}{2} q^{\frac{1}{2} ({ \nu +\nu'})}e^{-\mu' }q^{2s} \left((1+q^{-2})U  +   2e^{\mu'}\cosh(\mu)   -(v^2+v^{-2})q^{-2s-1} \right)\nonumber\\
&&   -  q^{-2}u^{-2} \phi(qu^{2})-u^2\phi(q^{-1}u^{-2})\ \nonumber
\eeqa
into (\ref{LsMdiag}), one obtains (\ref{act3}) using  (\ref{AWop1a}),  (\ref{AWop2a}) through the identification $z=qu^2$. The proof of (\ref{act2}) is the specialization $\kappa^*=0$ of the proof of (\ref{act3}).
\end{proof}

Although  not reported here, the connection between the Baxter T-Q relation for the generic case (\ref{LsMgen}) and  the Heun-Askey-Wilson q-difference operator introduced in \cite{BTVZ} can be established using a different realization of the Askey-Wilson algebra.
In this case, the action of the Heun-Askey-Wilson q-difference operator $\pi({\textsf I}(\kappa,\kappa^*,\kappa_+,\kappa_-))$ on the Baxter Q-polynomial 
also takes a form similar to (\ref{act3}).

\subsection{The Q-polynomial for the special case and the Askey-Wilson polynomials}
It is well-known that the eigenfunctions of the second-order q-difference operator (\ref{AWop2a}) are given by the Askey-Wilson orthogonal polynomials \cite[eq. (3.1.6)]{KS}. Recall that the Askey-Wilson polynomials   $P_M(x)$, $M=0,1,2,...$ are defined by:
\beqa
P_M(x;\fa,\fb,\fc,\fd)= \fpt{q^{-2M}}{\fa\fb\fc\fd q^{2M-2}}{\fa z}{\fa z^{-1}}{\fa\fb}{\fa\fc}{\fa\fd }\quad \mbox{with}\quad x=z+z^{-1}\ . \label{awpoly}
\eeqa
Denoting the Askey-Wilson  second-order q-difference by ${\mathbb D}$ according to \cite[eq. (3.1.6)]{KS},  the Askey-Wilson polynomials solve the bispectral problem:
\beqa
{\mathbb D} P_M(x) &=&   (q^{-2M} + \fa\fb\fc\fd q^{2M-2}) P_M(x)  ,\label{awd}\\
x P_M(x) &=& b_M  P_{M+1}(x)  + a_M  P_{M}(x) + c_M  P_{M-1}(x)  \label{rec}
\eeqa
where $P_{-1}\equiv 0$ and
\beqa
b_M&=& \frac{ (1-\fa\fb q^{2M}) (1-\fa\fc q^{2M}) (1-\fa\fd q^{2M}) (1-\fa\fb\fc\fd q^{2M-2})}{\fa (1-\fa\fb\fc\fd q^{4M-2}) (1-\fa\fb\fc\fd q^{4M})} \ ,\nonumber\\
 c_M&=&   \frac{ \fa(1- q^{2M}) (1-\fb\fc q^{2M-2}) (1-\fb\fd q^{2M-2}) (1-\fc\fd q^{2M-2})}{ (1-\fa\fb\fc\fd q^{4M-4}) (1-\fa\fb\fc\fd q^{4M-2})}  \ ,\nonumber\\
 a_M&=&  \fa + \fa^{-1} - b_M - c_M . \nonumber
\eeqa
From the three-term recurrence relation (\ref{rec}), one extracts the leading term of degree $M$ in the variable $x$:
\beqa
P_M(x) =  \left(\frac{(\fa\fb;q^2)_M (\fa\fc;q^2)_M (\fa\fd;q^2)_M (\fa\fb\fc\fd q^{-2};q^2)_M}{ \fa^M (\fa\fb\fc\fd q^{-2};q^4)_M (\fa\fb\fc\fd;q^4)_M} \right)^{-1} x^M + \cdots\label{asymPM}
\eeqa 

For generic parameters $q,\fa,\fb,\fc,\fd$, the spectrum in (\ref{awd})  is non-degenerate, of the form (\ref{st}). Based on this observation, a comparison of  (\ref{awd}) and (\ref{asymPM}) with (\ref{act1})  immediately yields to:
\begin{prop}\label{pQM} For the special case $\kappa=\kappa_\pm=0$, the Q-polynomial  (\ref{QM}) of Proposition \ref{propTQsp} is given by
\beqa
Q_M(Z)= \frac{(\fa\fb;q^2)_M (\fa\fc;q^2)_M (\fa\fd;q^2)_M (\fa\fb\fc\fd q^{-2};q^2)_M}{(q+q^{-1})^M \fa^M (\fa\fb\fc\fd q^{-2};q^4)_M (\fa\fb\fc\fd;q^4)_M }  P_{M}(z + z^{-1};\fa,\fb,\fc,\fd )\label{Qf}
\eeqa
with  
\beqa
\fa = -q e^{-\mu + \mu'}\ ,\quad \fb = -q e^{\mu + \mu'}\ ,\quad \fc = q^{-2s}v^2 \ ,\quad \fd = q^{-2s}v^{-2}\ .\label{pfix} 
\eeqa
\end{prop}

Expanding the Askey-Wilson polynomials in the variable $Z=(z+z^{-1})/(q+q^{-1})$ using (\ref{awpoly}) (see \cite{KS} for definitions), for the Q-polynomial  (\ref{Qf}) we obtain:
\beqa
Q_M(Z) = \sum_{l=0}^M (-1)^l\textsf{Q}_{l,M} Z^{M-l}\ \label{Qpower}
\eeqa
with
\beqa
\textsf{Q}_{l,M}&=& \frac{(-1)^M(\fa\fb;q^2)_M (\fa\fc;q^2)_M (\fa\fd;q^2)_M (\fa\fb\fc\fd q^{-2};q^2)_M}{(q+q^{-1})^l \fa^M (\fa\fb\fc\fd q^{-2};q^4)_M (\fa\fb\fc\fd;q^4)_M } \label{coeffQlm}\\
&&\times \ \sum_{k=M-l}^M \frac{(q^{-2M};q^2)_k(\fa\fb\fc\fd q^{2M-2};q^2)_kq^{k(k+1)} \fa^k} {(\fa\fb ;q^2)_k (\fa\fc ;q^2)_k (\fa\fd ;q^2)_k(q^2;q^2)_k} \textsf{e}_{l+k-M}(x_1,x_2,...,x_k)\ \nonumber
\eeqa
where
\beqa
x_l = \fa q^{2l-2} + \fa^{-1}q^{2-2l}\ . \nonumber
\eeqa
\begin{example} With the identification (\ref{pfix}):
\beqa
&&Q_1(Z)=   Z-{\frac {(\fa + \fb + \fc + \fd - \fa\fb\fc - \fa\fb \fd - \fa\fc\fd -\fb\fc\fd ) }{\left( {q}+q^{-1} \right) \left( 1-\fa\fb\fc\fd \right)   }}
\ ,\nonumber\\
&&Q_2(Z)= Z^2 -  {\frac {q(\fa + \fb + \fc + \fd - q^2(\fa\fb\fc + \fa\fb \fd + \fa\fc\fd +\fb\fc\fd )) }{ \left( 1-\fa\fb\fc\fd q^4 \right)   }} Z \nonumber\\
&&+ \frac{\left( \begin{array}{cc} \left( {q}^{2}-1 \right)  \left( 1-{\fa}^{2}{\fb}^{2}{\fc}^{2}{\fd}^{2}{q}^{4
} \right) + \left( \fa\fb    +\fa\fc    +\fa\fd    +\fb\fd    +\fb\fc  +\fc\fd \right)  \left( {q}^{2}+1
 \right) + \left( {\fa}^{2}+{\fb}^{2}+{\fc}^{2}+{\fd}^{2} \right) {q}^{2}  \nonumber\\
+
 \left( {\fa}^{2}{\fb}^{2}{\fc}^{2}+{\fa}^{2}{\fb}^{2}{\fd}^{2}+{\fa}^{2}{\fc}^{2}{\fd}^
{2}+{\fb}^{2}{\fc}^{2}{\fd}^{2} \right) {q}^{4}+\fa\fb\fc\fd  \left( \fa\fb    +\fa\fc    +\fa\fd    +\fb\fd    +\fb\fc  +\fc\fd \right)  \left( {q}^{4}+{q
}^{6} \right)   \nonumber\\
- \left( {q}^{4}+{q}^{2}+1 \right)  \left( \fa\fb    {\fc}^{2}+\fa{\fb
}^{2}\fc+\fa\fb    {\fd}^{2}+\fa{\fb}^{2}\fd+\fa\fc    {\fd}^{2}+\fa{\fc}^{2}\fd+{\fa}^{2}\fc\fd+{\fa}^{2}\fb\fd    +{\fa}
^{2}\fb\fc  +{\fb}^{2}\fc\fd+\fb{\fc}^{2}\fd+\fb\fc    {\fd}^{2} \right)  \nonumber\\ 
-\fa\fb\fc\fd  \left( {q}^{2}+1
 \right) ^{3} +\fa\fb\fc   \left( \fa+\fb+\fc
 \right) +\fa\fb\fd    \left( \fa+\fb+\fd \right) +\fa\fc\fd     \left( \fa+\fc+\fd \right) +\fb\fc\fd    
 \left( \fb+\fc+\fd \right)  \end{array} \right)}{(q+q^{-1})^2(1-\fa\fb\fc\fd q^2)(1-\fa\fb\fc\fd q^4)}
 .\nonumber
\eeqa
\end{example}

According to Proposition \ref{pQM}, the roots of the Q-polynomial can be computed for large values of $M$, a regime in which usually they are difficult to access by solving directly the Bethe equations of Proposition \ref{p33} (or similarly (\ref{PBA}) for $a=sp$).
\vspace{1mm}

To conclude this subsection, let us remark that the identification of the roots of the Q-polynomial (\ref{Qpower}) with the admissible Bethe roots satisfying Proposition \ref{p33}  (or (\ref{PBA}) for $a=sp$)  imply the existence of certain relations that are now described for completeness. For instance, given $M$ fixed a comparison between $Q_M(Z)$ (see (\ref{QM}))  and (\ref{Qpower}) leads to a system of equations for the symmetrized Bethe roots (\ref{sBr}). Equating the coefficients of both polynomials,  one finds that the symmetrized Bethe roots (\ref{sBr}) satisfy the following set of relations:
\beqa
 \textsf{e}_{l}(U_1,U_2,...,U_M)= \textsf{Q}_{l,M}\ \quad \mbox{for} \quad l=0,1,\cdots, M \ .\label{poly2}
\eeqa
For small values of $M=1,2,3$, these relations are equivalent to the following polynomial equations.
\begin{example}
\beqa
\mbox{For $M=1$:}&& U_1 = \textsf{Q}_{1,1}\ , \nonumber\\
\mbox{For $M=2$:}&& U_2 = \textsf{Q}_{1,2}- U_1\ ,\nonumber\\
&&  U_1^2 - U_1 \textsf{Q}_{1,2} + \textsf{Q}_{2,2}=0 \ , \nonumber\\
\mbox{For $M=3$:}&& U_3 = \textsf{Q}_{1,3}- U_1-U_2\ , \nonumber\\
&& U_2^2 +U_2(U_1- \textsf{Q}_{1,3}) + U_1^2  -U_1\textsf{Q}_{1,3} + \textsf{Q}_{2,3}=0 \ , \nonumber\\
&& U_1^3 - U_1^2 \textsf{Q}_{13} + U_1\textsf{Q}_{2,3} - \textsf{Q}_{3,3}=0\ .\nonumber
\eeqa
\end{example}

Let us make some comments about the polynomial equations (\ref{poly2}) compared with the polynomial equations  (\ref{PBA}).
For small values of $M$, using (\ref{pfix}) and (\ref{coeffQlm}) we have computed numerically the solutions of   (\ref{poly2}). It is found that they correspond to the subset of adminissible solutions of (\ref{PBA}). A direct proof of this fact for generic values of $M$ remains to be done.

As a second consequence of the relation between the Q-polynomial and the Askey-Wilson polynomials, an explicit relation between  sets of Bethe roots associated with different values of $M$ can be extracted from the recurrence relation (\ref{rec}). For convenience, introduce the notation 
\beqa
\textsf{Q}_{l,M} \equiv \textsf{e}_{l}(U^{[M]}_1,U^{[M]}_2,...,U^{[M]}_M)
\eeqa
where the indices are added in order to distinguish between different root systems $\{U^{[M]}_1,U^{[M]}_2,...,U^{[M]}_M\}$ that solve (\ref{PBA}).  From (\ref{rec}) and (\ref{Qf}) it follows:
\beqa
\textsf{Q}_{1,M+1}&=& \textsf{Q}_{1,M} +  \frac{a_M}{(q+q^{-1})}\ ,\quad M>0\nonumber\\
\textsf{Q}_{l+2,M+1}&=& \textsf{Q}_{l+2,M}  + \frac{a_M}{(q+q^{-1})} \textsf{Q}_{l+1,M} -   \frac{c_Mb_{M-1}}{(q+q^{-1})^2} \textsf{Q}_{l,M-1} \ \quad \mbox{for}\quad l=0,1,\cdots , M-2,\quad M>1\nonumber\\
\textsf{Q}_{M+1,M+1} &=&  \frac{a_M}{(q+q^{-1})} \textsf{Q}_{M,M} -   \frac{c_Mb_{M-1}}{(q+q^{-1})^2} \textsf{Q}_{M-1,M-1} \ .\nonumber
\eeqa

\subsection{The Q-polynomial as a transition coefficient}
For the dynamical operators (\ref{Bm}), consider the choice of gauge parameter $\beta=0$. Using (\ref{qc1}), after straightforward simplifications one gets:
\ben
&&\pi(\mathscr{B}^{+}(u,m))
=\frac{\chi b(u^2)}{\alpha (q-q^{-1})q^{2+m}u}
\left(U-\frac{z+z^{-1}}{q+q^{-1}}\right) \ .
\een
Considering in particular the product of dynamical operators (\ref{SB}) entering in the definition of the Bethe states of Propositions \ref{p31}, \ref{p33} and \ref{p34},  the following lemma is easily shown.
\begin{lem}\label{lem5} Assume the gauge parameter $\alpha$ satisfies (\ref{ab}). For $a=sp\ (M=0,1,...,2s)$ or  $a=d \ (M=2s)$, the Baxter Q-polynomial (\ref{QM}) is given by
\beqa
Q_M(Z) = {\cal N}_M(\bar u)^{-1} \langle z |\Psi_{a,+}^{M}(\bar u)\rangle \quad \mbox{with}\quad   {\cal N}_M(\bar u)  = (-1)^M \frac{(q+q^{-1})^M}{2^M} \left(q^{\frac{\nu+\nu'}{2}} e^{-\mu'}q^{2s-M-1}\right)^M \prod_{i=1}^{M} \frac{b(u_i^2)}{u_i}\ \label{norm}\nonumber
\eeqa
for (\ref{PBA}) and the following notation:
\beqa
 \langle z |\Psi_{a,+}^{M}(\bar u)\rangle = \pi(\mathscr{B}^{+}(u_1,m_0+2(M-1))\cdots \mathscr{B}^{+}(u_M,m_0))|_{\beta=0}\ \label{Qtrans} {\bf 1} \ .
\eeqa
\end{lem}

Here we have  denoted
\beqa
 \langle z |\Psi_{a,+}^{0}(\bar u)\rangle = \langle z  |\Omega_+\rangle = Q_0(Z) \equiv  {\bf 1} \  . \nonumber
\eeqa
For completeness, the action of the other dynamical operators on ${\bf 1}$ is computed in a straightforward manner using $T_\pm  {\bf 1} = {\bf 1} $. Using the expressions reported in Appendix \ref{apF}, for which $\alpha$ satisfies (\ref{ab}) and $\beta=0$, for any integer $m_0$ one finds\footnote{Note that after acting on ${\bf 1} $ with the expressions in Appendix \ref{apF} the dependence on $z$ vanishes and the results drastically simplify.}
\ben
\pi(\mathscr{A}^{+}(u,m_0))  {\bf 1} &=& u^{-1}b(u^2)b(qu^2)\phi(q^{-1}u^{-2})\,,\nonumber\\
 \pi(\mathscr{D}^{+}(u,m_0)) {\bf 1} &=& u^{-1}b(u^2)b(q^2u^2)\phi(qu^2)\,,\nonumber\\
 \pi(\mathscr{C}^{+}(u,m_0)) {\bf 1} &=& 0\, .\nonumber
\een

Suppose a rigorous mathematical definition of a basis $\{|z\rangle \}$ within the framework of rigged Hilbert spaces (see e.g. \cite{Madr1}) is given for the Askey-Wilson algebra, extending the case of the quantum harmonic oscillator and Hermite functions \cite{Ce} to the realm of Askey-Wilson orthogonal polynomials \cite[p. 50]{KS}. To our knowledge, this problem has not been considered yet in the literature.
If solved, then the notation  $\langle z |\Psi_{sp,+}^{M}(\bar u)\rangle$ would find a natural interpretation as a transition coefficient connecting the continuous basis   $\{|z\rangle \}$ and the discrete basis $\{|\theta^*_{M}\rangle\}$  given in  Proposition \ref{baseLP}.

\section{Applications}\label{App}

In this Section, we apply the algebraic Bethe ansatz solution for the Heun-Askey-Wilson operator of Section 3 to the diagonalization of the q-analog of the quantum Euler top as well as to various examples of Hamiltonians of $3$-sites Heisenberg spin chains in a magnetic field, inhomogeneous couplings, three-body terms and boundary interactions.

\subsection{Algebraic Bethe ansatz solution for the q-analog of the quantum Euler top}
The relation between the Hamiltonian of a quantum Euler top in a magnetic field built from $sl_2({\mathbb R})$ and the Heun operator has been recently studied  in \cite{Tu16} (see also \cite{WZ}). It is a quantum version of the Zhukovsky-Volterra gyrostat  of classical mechanics \cite{Ba09,LOZ06}, and arises in spin systems with anisotropy \cite{ZU87}. Considering the realization of the Askey-Wilson algebra given in Example \ref{ex1}, it is straightforward to derive a q-deformed analog of
 the Euler top \cite[eq. (2)]{Tu16} generated from $U_q(sl_2)$ starting from (\ref{I}). Replacing (\ref{Aex1}), (\ref{Astarex1}) in (\ref{I}) one gets the bilinear  expression considered in \cite[Section 3.2]{WZ}. It gives a Hamiltonian of the form:
\beqa
\bar\pi({\textsf I}(\kappa,\kappa^*,\kappa_+,\kappa_-)) &=& t_{+-}S_+S_- + t_{00}q^{2s_3} + t'_{00}q^{-2s3} + t_{++}S_+^2  + t_{--}S_-^2 \nonumber\\
&& + t_{0+}S_+ q^{s_3} + t_{0-}S_- q^{s_3} + t'_{0+}S_+ q^{-s_3} +
 t'_{0-}S_- q^{-s_3} + I_0\ ,\nonumber
\eeqa
where the coupling constants $t_{+-},t_{00},t'_{00}, t_{\pm \pm},t_{0\pm},t'_{0\pm}$ are expressed in terms of the parameters
$k_\pm,\epsilon_\pm,v,\kappa,\kappa^*,\kappa_\pm,\chi$ introduced  in Example \ref{ex1}  and $I_0$ central in $U_q(sl_2)$.\vspace{1mm}

For an irreducible finite dimensional representation of $U_q(sl_2)$ of dimension $2s+1$ on which this Hamiltonian acts, we  consider the parametrization
(\ref{parbis}) and apply the results of Section \ref{s3}. For generic parameters $\kappa,\kappa^*,\kappa_\pm$, the  spectrum and Bethe eigenstates are
 given by Proposition \ref{p35}. In the Tables \ref{table:top1}, \ref{table:top2} and \ref{table:top3}, we give the numerical results for $s=1/2,s=1,s=3/2$, respectively.

\begin{table}[h]
\begin{tabular}{|c|c|c|c|}	
\hline
{\bf Spin $s=\frac{1}{2}$}
 & \begin{tabular}[c]{c} Direct diagonalization \end{tabular} & 
\begin{tabular}[c]{c}Diagonalization via ABA\\${\Lambda}_{a,+}^{1}$        \end{tabular}
& 
\begin{tabular}[c]{c}Bethe roots\\ $\{U_1,\dots,U_{2s}\}$   \end{tabular}
 \\ \hline
\begin{tabular}[c]{c}$\kappa=1$, $\kappa^*=0.25$,\\ $\kappa_\pm=0$ \end{tabular}&  \begin{tabular}[c]{c}6.40701 + 3.99187 i\\ 6.40701 - 3.99187 i\end{tabular}  &
\begin{tabular}[c]{c}6.40701 + 3.99187 i\\ 6.40701 - 3.99187 i\end{tabular}
& 
\begin{tabular}[c]{c}\{-0.8157 - 1.07769 i\}\\ \{-0.8157 + 1.07769 i\}\end{tabular}
\\ \hline
\begin{tabular}[c]{c} $\kappa=\frac{10i(1+\sqrt{5})}{\sqrt{3}}\,, \kappa^*=\frac{20i}{\sqrt{3}}$, \\$\kappa_+=\frac{\sqrt{3}}{2}\,,\kappa_-=-\frac{3}{2}$\\$\chi=-\frac{40}{3\sqrt{3}}$ \end{tabular} &  \begin{tabular}[c]{c}294.909 - 337.018 i\\ -1006.7 + 746.697 i\end{tabular}  &
\begin{tabular}[c]{c}294.909 - 337.018 i\\ -1006.7 + 746.697 i\end{tabular}
& 
\begin{tabular}[c]{c}\{0.454184 + 0.509566 i\}\\ \{1.36299 - 0.412627 i\}\end{tabular}
\\ \hline
\end{tabular}
\vspace{2mm}
\caption{Numerical results for the parameters $q=3$, $\nu=\nu'=1$,  $\mu=0.2$, $\mu'=0.3$, $v=1.1$.}
\label{table:top1}
\end{table}

\begin{table}[h]
\begin{tabular}{|c|c|c|c|}	
\hline
{\bf Spin $s=1$}
 & \begin{tabular}[c]{c} Direct diagonalization \end{tabular} & 
\begin{tabular}[c]{c}Diagonalization via ABA\\${\Lambda}_{a,+}^{2}$        \end{tabular}
& 
\begin{tabular}[c]{c}Bethe roots\\ $\{U_1,\dots,U_{2s}\}$   \end{tabular}
 \\ \hline
\begin{tabular}[c]{c}$\kappa=1$, $\kappa^*=0.25$,\\ $\kappa_\pm=0$ \end{tabular}&  \begin{tabular}[c]{c}\begin{tabular}[c]{c}17.8556\\{}\end{tabular}\\ \begin{tabular}[c]{c}10.5068 + 9.82751 i\\{}\end{tabular}\\ \begin{tabular}[c]{c}10.5068 - 9.82751 i\\{}\end{tabular}\end{tabular}  &
\begin{tabular}[c]{c}\begin{tabular}[c]{c}17.8556\\{}\end{tabular}\\ \begin{tabular}[c]{c}10.5068 + 9.82751 i\\{}\end{tabular}\\ \begin{tabular}[c]{c}10.5068 - 9.82751 i\\{}\end{tabular}\end{tabular}
& 
\begin{tabular}[c]{c}\begin{tabular}[c]{c}\{-7.53525,\\-2.25731\}\end{tabular}\\ \begin{tabular}[c]{c}\{-2.89915 - 7.58381 i,\\-0.941451 - 0.375642 i\}\end{tabular}\\ \begin{tabular}[c]{c}\{-2.89915 + 7.58381 i,\\-0.941451 + 0.375642 i\}\end{tabular}\end{tabular}
\\ \hline
\begin{tabular}[c]{c} $\kappa=\frac{10i(1+\sqrt{5})}{\sqrt{3}}\,, \kappa^*=\frac{20i}{\sqrt{3}}$, \\$\kappa_+=\frac{\sqrt{3}}{2}\,,\kappa_-=-\frac{3}{2}$\\$\chi=-\frac{40}{3\sqrt{3}}$ \end{tabular} &  \begin{tabular}[c]{c}\begin{tabular}[c]{c}-2394.67 + 986.732 i\\{}\end{tabular}\\ \begin{tabular}[c]{c}-6079.21 + 1505.54 i\\{}\end{tabular}\\ \begin{tabular}[c]{c}-1543.12 - 1249.58 i\\{}\end{tabular}\end{tabular}  &
\begin{tabular}[c]{c}\begin{tabular}[c]{c}-2394.67 + 986.732 i\\{}\end{tabular}\\ \begin{tabular}[c]{c}-6079.21 + 1505.54 i\\{}\end{tabular}\\ \begin{tabular}[c]{c}-1543.12 - 1249.58 i\\{}\end{tabular}\end{tabular}
& 
\begin{tabular}[c]{c}\begin{tabular}[c]{c}\{2.98826 - 0.846233 i,\\-0.155658 + 1.20672 i\}\end{tabular}\\ \begin{tabular}[c]{c}\{4.06015 + 0.244047 i,\\1.69724 - 0.997537 i\}\end{tabular}\\ \begin{tabular}[c]{c}\{2.43438 + 1.09148 i,\\0.117738 + 1.26215 i\}\end{tabular}\end{tabular}
\\ \hline
\end{tabular}
\vspace{2mm}
\caption{Numerical results for the parameters $q=3$, $\nu=\nu'=1$,  $\mu=0.2$, $\mu'=0.3$, $v=1.1$.}
\label{table:top2}
\end{table}

\begin{table}[h]
\begin{tabular}{|c|c|c|c|}	
\hline
{\bf Spin $s=\frac{3}{2}$}
 & \begin{tabular}[c]{c} Direct diagonalization \end{tabular} & 
\begin{tabular}[c]{c}Diagonalization via ABA\\${\Lambda}_{a,+}^{3}$        \end{tabular}
& 
\begin{tabular}[c]{c}Bethe roots\\ $\{U_1,\dots,U_{2s}\}$   \end{tabular}
 \\ \hline
\begin{tabular}[c]{c}$\kappa=1$, $\kappa^*=0.25$,\\ $\kappa_\pm=0$ \end{tabular}&  \begin{tabular}[c]{c}\begin{tabular}[c]{c}37.2756\\{}\\{}\end{tabular}\\ \begin{tabular}[c]{c}46.446\\{}\\{}\end{tabular}\\ \begin{tabular}[c]{c}16.5142 + 19.0709 i\\{}\\{}\end{tabular}\\ \begin{tabular}[c]{c}16.5142 - 19.0709 i\\{}\\{}\end{tabular}\end{tabular}  &
\begin{tabular}[c]{c}\begin{tabular}[c]{c}37.2756\\{}\\{}\end{tabular}\\ \begin{tabular}[c]{c}46.446\\{}\\{}\end{tabular}\\ \begin{tabular}[c]{c}16.5142 + 19.0709 i\\{}\\{}\end{tabular}\\ \begin{tabular}[c]{c}16.5142 - 19.0709 i\\{}\\{}\end{tabular}\end{tabular}
& 
\begin{tabular}[c]{c}\begin{tabular}[c]{c}\{-13.1346,\\-1.2149,\\-40.6532\}\end{tabular}\\ \begin{tabular}[c]{c}\{-14.4407,\\-1.493,\\-61.3506\}\end{tabular}\\ \begin{tabular}[c]{c}\{-3.616 - 3.04982 i,\\-0.0635861 - 43.0976 i,\\-0.878068 - 0.190033 i\}\end{tabular}\\ \begin{tabular}[c]{c}\{-3.616 + 3.04982 i,\\-0.0635861 + 43.0976 i\\-0.878068 + 0.190033 i\}\end{tabular}\end{tabular}
\\ \hline
\begin{tabular}[c]{c} $\kappa=\frac{10i(1+\sqrt{5})}{\sqrt{3}}\,, \kappa^*=\frac{20i}{\sqrt{3}}$, \\$\kappa_+=\frac{\sqrt{3}}{2}\,,\kappa_-=-\frac{3}{2}$\\$\chi=-\frac{40}{3\sqrt{3}}$ \end{tabular} &  \begin{tabular}[c]{c}\begin{tabular}[c]{c}-20905.2 + 2621.26 i\\{}\\{}\end{tabular}\\ \begin{tabular}[c]{c}-15139.6 + 3210.92 i\\{}\\{}\end{tabular}\\ \begin{tabular}[c]{c}-9508.34 - 2492.84 i\\{}\\{}\end{tabular}\\ \begin{tabular}[c]{c}-6679.41 + 393.291 i\\{}\\{}\end{tabular}
\end{tabular}  &
\begin{tabular}[c]{c}\begin{tabular}[c]{c}-20905.2 + 2621.26 i\\{}\\{}\end{tabular}\\ \begin{tabular}[c]{c}-15139.6 + 3210.92 i\\{}\\{}\end{tabular}\\ \begin{tabular}[c]{c}-9508.34 - 2492.84 i\\{}\\{}\end{tabular}\\ \begin{tabular}[c]{c}-6679.41 + 393.291 i\\{}\\{}\end{tabular}\end{tabular}
& 
\begin{tabular}[c]{c}\begin{tabular}[c]{c}\{1.04709 - 1.16555 i,\\8.06887 - 1.51037 i\\9.48236 + 1.52911 i\}\end{tabular}\\ \begin{tabular}[c]{c}\{-0.255185 + 3.99139 i,\\3.64641 - 1.51464 i\\11.0049 - 0.370951 i\}\end{tabular}\\ \begin{tabular}[c]{c}\{0.658762 + 0.810534 i,\\4.03758 + 7.37438 i,\\9.56462 + 1.26855 i\}\end{tabular}\\ \begin{tabular}[c]{c}\{2.80458 + 9.49603 i,\\0.718949 + 0.525899 i,\\7.03427 - 0.69735 i\}\end{tabular}\end{tabular}
\\ \hline
\end{tabular}
\vspace{2mm}
\caption{Numerical results for the parameters $q=3$, $\nu=\nu'=1$,  $\mu=0.2$, $\mu'=0.3$, $v=1.1$.}
\label{table:top3}
\end{table}

\subsection{Algebraic Bethe ansatz solution of $3$-sites Heisenberg spin chains}
The construction of $3$-sites Heisenberg spin chains follow from Example \ref{ex2} and Example \ref{ex4}, by inserting (\ref{Aex2}), (\ref{Astarex2}) into (\ref{I}). In the next subsections, we compare numerically the direct diagonalization of these Hamiltonians to the algebraic Bethe ansatz solution given in Propositions \ref{p34}, \ref{p35}. Successively, we consider the cases of a spin-1/2 chain and a spin-1 chain. Introduce the anisotropy parameter $\Delta = \frac{q+q^{-1}}{2}$.

\subsubsection{The Heisenberg chain for $j_1=j_2=j_3=1/2$}
Denote $S_i^a$, $a=x,y,z$ with $i=1,2,3$ as the operators acting on the representation $V(j_i)$ such that 
\beqa\label{Pauli}
S^x = \frac{1}{2}\left(
\begin{array}{cc}
 0    & 1 \\
 1 & 0 
\end{array} \right)\ ,\quad S^y=\frac{1}{2}\left(
\begin{array}{cc}
 0    & -i \\
 i & 0 
\end{array} \right)  \ , \quad S^z = \frac{1}{2}\left(
\begin{array}{cc}
 1    & 0 \\
 0 & -1 
\end{array} \right)\ .
\eeqa
Using
\ben
&& \bar \pi (q^{\pm s_3}) =\pm \frac{q-1}{\sqrt{q}}S^z + \frac{q+1}{2\sqrt{q}} \,,\quad  \bar \pi (S_\pm) = S^x \pm i S^y\, ,\nonumber
\een
it follows that the Heun-Askey-Wilson operator is given by
\beqa\label{Hs12}
\bar\pi( {\textsf I}(\kappa,\kappa^*,\kappa_+,\kappa_-)) &=& \displaystyle {2(q-q^{-1})^2\Big( \sum_a \sum_{i,j}} J_{ij}^a S_i^a S_j^a + \displaystyle{ \sum_{ a \neq b } \sum_{i,j}} K_{ij}^{ab} S_i^aS_j^b + \displaystyle{\sum_{\{a,b,c\}} \sum_{i,j,k}} L_{ijk}^{abc} S_i^a S_j^b S_k^c + \displaystyle{\sum_i} b_i^z S_i^z \Big)+J_0\non\\
\eeqa
with the coupling constants,
\beqa
	J_{12}^x&=&J_{12}^y=\kappa + \frac{3q^{4}-2q^2+2q^{-2}-3q^{-4}}{4}(\kappa_+ \chi^{-1} + \kappa_- \chi ),\nonumber \\ J_{23}^x&=&J_{23}^y=\kappa^* + \frac{3q^{4}-2q^2+2q^{-2}-3q^{-4}}{4}(\kappa_+ \chi^{-1} + \kappa_- \chi ),\nonumber \\
	J_{12}^z&=&\Delta \kappa +   \frac{q-q^{-1}}{2}(q^2+q^{-2})^2 (\kappa_+ \chi^{-1} + \kappa_- \chi ),\nonumber \\
	J_{23}^z&=&\Delta \kappa^* +   \frac{q-q^{-1}}{2}(q^2+q^{-2})^2 (\kappa_+ \chi^{-1} + \kappa_- \chi ) ,\nonumber\\
	J_{13}^x&=&J_{13}^y=J_{13}^z =   \frac{(q-q^{-1})^3}{2}(\kappa_+ \chi^{-1} + \kappa_- \chi ), \nonumber\\
 K_{12}^{xy}&=&K_{23}^{yx}=-K_{12}^{yx}=-K_{23}^{xy}=-i\frac{(q-q^{-1})^3}{2} (\kappa_+ \chi^{-1} - \kappa_- \chi ) \Delta,\nonumber \\
	L_{123}^{xzy}&=&-L_{123}^{yzx}=2i(q-q^{-1})^2 (\kappa_+ \chi^{-1} - \kappa_- \chi )\Delta, \nonumber\\
 L_{123}^{yxz}&=&L_{123}^{zyx}=-L_{123}^{xyz}=-L_{123}^{zxy}=\Delta L_{123}^{xzy},\nonumber \\
	L_{123}^{xxz} &=& L_{123}^{yyz}=-L_{123}^{zxx}=-L_{123}^{zyy}= -\frac{(q-q^{-1})^4}{4} (\kappa_+ \chi^{-1} + \kappa_- \chi ),\nonumber \\
	b_2^z&=&-\frac{q-q^{-1}}{2}(\kappa-\kappa^*), \nonumber\\
b_1^z&=&\frac{q-q^{-1}}{4} (\kappa + (q^3+q^{-3})(q-q^{-1}) (\kappa_+ \chi^{-1} + \kappa_- \chi ) ),\nonumber \\ 
b_3^z&=&-\frac{q-q^{-1}}{4} (\kappa^* + (q^3+q^{-3})(q-q^{-1}) (\kappa_+ \chi^{-1} + \kappa_- \chi ) )
, \nonumber\\
J_0 &=& (q+q^{-1}+3 q^3+3 q^{-3})\frac{(\kappa+\kappa^*)}{4}\nonumber\\&&+
(12 q-12 q^{-1}-5 q^3+5 q^{-3}+q^5-q^{-5}+2 q^7-2 q^{-7})\frac{(\kappa^+\chi^{-1}+\kappa^-\chi)}{4}\,.
\nonumber
\eeqa
For the diagonal case $\kappa_\pm=0$, the three-body terms  vanish and the Heun-Askey-Wilson operator simplifies to
\beqa 
\bar\pi({\textsf I}(\kappa,\kappa^*,0,0))  &=&\frac{(q-q^{-1})^2}{2}\Big(  \kappa ( S_1^x S_2^x + S_1^y S_2^y + \Delta S_1^z S_2^z) + \kappa^* (  S_2^x S_3^x + S_2^y S_3^y + \Delta S_2^z S_3^z)\non\\&&- \frac{q-q^{-1}}{2} ( \kappa^* S_3^z - \kappa S_1^z + (\kappa -\kappa^*) S_2^z)\Big)+\frac{(q+q^{-1})}{2}(\kappa+\kappa^*)\Big(\frac{3}{2}(q^2+q^{-2})-1\Big).\non
\eeqa

Following Example \ref{ex4}, we have set $j_1=j_2=j_3=1/2$. Then the set $\Sigma=(\ell,k)$ such that: 
\beqa
0 \leq \ell \leq 1 \ ,\qquad 0 \leq k  \leq 3 - 2\ell \  .
\eeqa
In the algebraic Bethe ansatz results of Proposition \ref{p33} and \ref{p34}, we fix:
\beqa
q^{\nu+\nu'}=4\ ,\quad  e^{-\mu}= e^{\mu'}= -v^2= q^{\min(1,\ell)-3}  \quad \mbox{and} \quad 2s=3\min(1,\ell) - 2\ell\ .
\eeqa

In Table \ref{table:j1j2j3a}, for numerical values of the scalar parameters $\kappa,\kappa^*,\kappa_\pm$ the eigenvalues of the Hamiltonian $\bar\pi({\textsf I}(\kappa,\kappa^*,\kappa_+,\kappa_-))$ obtained by direct diagonalization and the ones obtained from the algebraic Bethe ansatz (ABA) are displayed.
We remark that the Hamiltonian (\ref{Hs12}) is Hermitian if
\ben
\operatorname{Im}(q)=\operatorname{Im}(\kappa)=\operatorname{Im}(\kappa^*)=0\,,\quad \kappa_+ = \bar \kappa_- \lvert \chi\rvert ^2\nonumber\,,
\een
where $\bar \kappa_-$ denotes the complex conjugate of $\kappa_-$. In addition, for $\kappa=\kappa^*=1$ and $\kappa_{\pm}=0$, the Hamiltonian (\ref{Hs12}) reduces to the $U_q(sl_2)$-invariant XXZ chain \cite{Alcaraz:1987uk,Pasquier:1989kd}, and thus has real spectrum for $q$ in the unit circle \cite{Morin-Duchesne:2015afa}.
\begin{table}[h]
\begin{tabular}{|c|c|c|c|}	
\hline
{\bf Spin-$\frac{1}{2}$ chain}
 & \begin{tabular}[c]{c} Direct diagonalization \\ (degeneracy) \end{tabular} & 
\begin{tabular}[c]{c}Diagonalization via ABA\\${\Lambda}_{a,+}^{2s}$      $(s)$   \end{tabular}
& 
\begin{tabular}[c]{c}Bethe roots\\ $\{U_1,\dots,U_{2s}\}$   \end{tabular}
 \\ \hline
\begin{tabular}[c]{c}$\kappa=3\,, \kappa^*=1$ \,,\\ $\kappa_\pm=0$   \end{tabular}
 &  \begin{tabular}[c]{c}32.5 (4)\\ 14.4069 (2)\\ 28.0931 (2)\end{tabular}  &
\begin{tabular}[c]{c}32.5 (0)\\ 14.4069 (1/2)\\ 28.0931 (1/2)\end{tabular}
& 
\begin{tabular}[c]{c}-\\ \{-1.0344\}\\ \{-1.4906\}\end{tabular}
\\ \hline
\begin{tabular}[c]{c} $\kappa=-\frac{5}{4 \sqrt{2}}\,, \kappa^*=-\frac{9}{4 \sqrt{2}}$, \\$\kappa_+=\frac{1}{8}\,,\kappa_-=-\frac{1}{16}$\\$\chi=-\frac{15}{4}$ \end{tabular}
&  \begin{tabular}[c]{c}-0.200512 (4)\\ -6.25895 + 3.32745 i (2)\\ -6.25895 - 3.32745 i (2)\end{tabular}  &
\begin{tabular}[c]{c}-0.200512 (0)\\ -6.25895 + 3.32745 i (1/2)\\ -6.25895 - 3.32745 i (1/2)\end{tabular}
& 
\begin{tabular}[c]{c}-\\ \{-0.793147 - 1.40509 i\}\\ \{-0.793147 + 1.40509 i\}\end{tabular}
\\ \hline
\end{tabular}
\vspace{2mm}
\caption{Numerical results for the parameters $q=2$, $\nu=\nu'=1$.}
\label{table:j1j2j3a}
\end{table}

\subsubsection{The Heisenberg chain for $j_1=j_2=j_3=1$}
Denote $S_i^a$, $a=x,y,z$ with $i=1,2,3$ as the operators acting on the representation $V(j_i)$ such that 
\beqa\label{Pauli1}
S^x = \frac{1}{\sqrt{2}}\left(
\begin{array}{ccc}
 0 & 1 & 0 \\
 1 & 0 & 1 \\
 0 & 1 & 0 \\
\end{array}
\right)\ ,\quad S^y=\frac{i}{\sqrt{2}}\left(
\begin{array}{ccc}
 0 & -1 & 0 \\
 1 & 0 & -1 \\
 0 & 1 & 0 \\
\end{array}
\right)  \ , \quad S^z = \frac{1}{2}\left(
\begin{array}{ccc}
 1 & 0 & 0 \\
 0 & 0 & 0 \\
 0 & 0 & -1 \\
\end{array}
\right)\ .
\eeqa
We now have
\ben
&& \bar \pi (q^{\pm s_3}) =\pm \frac{(q-q^{-1})}{2}S^z + \frac{(q-1)^2}{2q}(S^z)^2+1 \,,\quad  \bar \pi (S_\pm) = \sqrt{\frac{q+q^{-1}}{2}}\left(S^x \pm i S^y\right)\, .\nonumber
\een
Note that the term $(S^z)^2$ will lead to a proliferation of higher-order terms in the basis $\{S_i^x,S_i^y,S_i^z\}$, and for that reason we only write the expression of the Heun-Askey-Wilson operator for the diagonal case.
It is given by
\beqa 
\bar\pi({\textsf I}(\kappa,\kappa^*,0,0))  &=&
\frac{\kappa}{2} \Big(
q(q^2-q^{-2})^2(S_1^xS_2^x+q^{-2}S_1^yS_2^y+q^{-1}\Delta S_1^zS_2^z)
\non\\&&+ (q-q^{-1})^3 (q+q^{-1})(-q(S_1^xS_2^xS_2^z+iS_1^xS_2^zS_2^y-S_1^xS_1^zS_2^x+iS_1^zS_1^yS_2^x)+
\non\\&&
\qq\qq\qq\qq\q +q^{-1}(S_1^zS_1^yS_2^y+iS_1^xS_1^zS_2^y-S_1^yS_2^zS_2^y+iS_1^yS_2^xS_2^z))
\non\\&&+ (q-q^{-1})^4 (q+q^{-1})(-S_1^xS_1^zS_2^xS_2^z-iS_1^xS_1^zS_2^zS_2^y-S_1^zS_1^yS_2^zS_2^y+iS_1^zS_1^yS_2^xS_2^z
\non\\&&
\qq\qq\qq\qq\q+(S_1^z)^2+(S_2^z)^2-\frac{1}{2}(S_1^z)^2(S_2^z)^2)
\non\\&&+ (q-q^{-1})^3 (q+q^{-1})^2(S_1^z-S_2^z+\frac{1}{2}(S_1^z)^2S_2^z-\frac{1}{2}S_1^z(S_2^z)^2)
+ 2(2q^3+2q^{-3}-q-q^{-1})
\Big) 
\non\\&+&
\frac{\kappa^*}{2} \Big(
q(q^2-q^{-2})^2(S_2^xS_3^x+q^{-2}S_2^yS_3^y+q^{-1}\Delta S_2^zS_3^z)
\non\\&&+ (q-q^{-1})^3 (q+q^{-1})(-q(S_2^xS_3^xS_3^z+iS_2^xS_3^zS_3^y-S_2^xS_2^zS_3^x+iS_2^zS_2^yS_3^x)+
\non\\&&
\qq\qq\qq\qq\q +q^{-1}(S_2^zS_2^yS_3^y+iS_2^xS_2^zS_3^y-S_2^yS_3^zS_3^y+iS_2^yS_3^xS_3^z))
\non\\&&+ (q-q^{-1})^4 (q+q^{-1})(-S_2^xS_2^zS_3^xS_3^z-iS_2^xS_2^zS_3^zS_3^y-S_2^zS_2^yS_3^zS_3^y+iS_2^zS_2^yS_3^xS_3^z
\non\\&&
\qq\qq\qq\qq\q+(S_2^z)^2+(S_3^z)^2-\frac{1}{2}(S_2^z)^2(S_3^z)^2)
\non\\&&+ (q-q^{-1})^3 (q+q^{-1})^2(S_2^z-S_3^z+\frac{1}{2}(S_2^z)^2S_3^z-\frac{1}{2}S_2^z(S_3^z)^2)
+ 2(2q^3+2q^{-3}-q-q^{-1})
\Big) \,.
\eeqa
Following Example \ref{ex4}, we have set $j_1=j_2=j_3=1$, and the set $\Sigma=(\ell,k)$ such that: 
\beqa
0 \leq \ell \leq 3 \ ,\qquad 0 \leq k  \leq 6 - 2\ell \  .
\eeqa
In the algebraic Bethe ansatz results of Proposition \ref{p33} and \ref{p34}, we fix:
\beqa
q^{\nu+\nu'}=4\ ,\quad  e^{-\mu}= e^{\mu'}= -v^2= q^{\min(2,\ell)-5}  \quad \mbox{and} \quad 2s=3\min(2,\ell) - 2\ell\ .
\eeqa
In Table \ref{table:j1j2j3b}, for numerical values of the scalar parameters $\kappa,\kappa^*,\kappa_\pm$ the eigenvalues of the Hamiltonian $\bar\pi({\textsf I}(\kappa,\kappa^*,\kappa_+,\kappa_-))$ obtained by direct diagonalization and the ones obtained from the algebraic Bethe ansatz (ABA) are displayed.  We remark that for generic $q$ the Hamiltonian $\bar\pi({\textsf I}(\kappa,\kappa^*,\kappa_+,\kappa_-))$ for $j_1=j_2=j_3=1$ is not Hermitian. However, similarly to the $j_1=j_2=j_3=1/2$ case, we observe that for $\kappa=\kappa^*=1$, $\kappa_{\pm}=0$ and $q$ in the unit circle, the spectrum is real.

\begin{table}[h]
\begin{tabular}{|c|c|c|c|}	
\hline
{\bf Spin-$1$ chain}
 & \begin{tabular}[c]{c} Direct diagonalization \\ (degeneracy) \end{tabular} & 
\begin{tabular}[c]{c}Diagonalization via ABA\\${\Lambda}_{a,+}^{2s}$      $(s)$   \end{tabular}
& 
\begin{tabular}[c]{c}Bethe roots\\ $\{U_1,\dots,U_{2s}\}$   \end{tabular}
 \\ \hline
\begin{tabular}[c]{c}$\kappa=3\,, \kappa^*=1$\,,\\ $\kappa_\pm=0$   \end{tabular} &  \begin{tabular}[c]{c}128.125 (7)\\ 106.128 (5)\\ 54.4972 (5)\\  100.033 (3)\\   45.8128 (3) \\ 24.7794 (3) \\ 32.5 (1)\end{tabular}  &
\begin{tabular}[c]{c}128.125 (0)\\ 106.128 (1/2)\\ 54.4972(1/2)\\ 100.033 (1) \\ 45.8128 (1) \\ 24.7794 (1) \\ 32.5 (0)\end{tabular}
& 
\begin{tabular}[c]{c} - \\ \{-3.82841\}\\ \{-3.39815\} \\ \{-5.84346, -1.26701\}\\ \{-4.28657,-1.01656\} \\\{-3.59681, -1.00521\} \\ -\end{tabular}
\\ \hline
\begin{tabular}[c]{c} $\kappa=-\frac{5}{4 \sqrt{2}}\,, \kappa^*=-\frac{9}{4 \sqrt{2}}$, \\$\kappa_+=\frac{1}{8}\,,\kappa_-=-\frac{1}{16}$\\$\chi=-\frac{15}{4}$ \end{tabular} &  \begin{tabular}[c]{c}230.13 (7)\\ 33.5628 + 33.3326 i (5)\\ 33.5628 - 33.3326 i (5)\\ -2.26847 + 22.221 i (3)\\ \\ -2.26847 - 22.221 i (3)\\ \\ -13.7437 (3)\\ \\ -0.200512 (1)\end{tabular}  &
\begin{tabular}[c]{c}230.13 (0)\\ 33.5628 + 33.3326 i (1/2)\\ 33.5628 - 33.3326 i (1/2)\\ -2.26847 + 22.221 i (1)\\ \\ -2.26847 - 22.221 i (1)\\ \\ -13.7437 (1)\\ \\ -0.200512 (0)\end{tabular}
& 
\begin{tabular}[c]{c} - \\ \{ -3.51786 - 5.42098 i\}\\ \{-3.51786 + 5.42098 i\}\\ \{-8.41833 - 8.60097 i,\\-0.252667 - 0.78235 i\} \\ \{-8.41833 + 8.60097 i,\\-0.252667 + 0.78235 i\}\\  \{-1.91267 + 5.49269 i,\\-1.91267 - 5.49269 i\}\\ - \end{tabular}
\\ \hline
\end{tabular}
\vspace{2mm}
\caption{Numerical results for the parameters $q=2$, $\nu=\nu'=1$.}
\label{table:j1j2j3b}
\end{table}

\newpage

\section{Perspectives}\label{persp}
Besides the generalization to the Askey-Scheme of the basic quantum harmonic oscillator construction described in the Introduction, there are four main motivations for the present paper.\vspace{1mm}

$\bullet$ The Askey-Wilson algebra with generators $\tA,\tA^*$ provides the algebraic framework for all orthogonal polynomials of the Askey-scheme \cite{Z91}: the bispectral problem with respect to $\tA,\tA^*$ produces the well-known recurrence and second-order q-difference, difference or differential equations satisfied by these polynomials. If  $\tA,\tA^*$ act as a Leonard pair (irreducible finite dimensional representation of the Askey-Wilson algebra are considered), the entries of the transition matrix between the two respective eigenbasis of the Leonard pair are given in terms  of the orthogonal polynomials \cite{T03,T04}. By analogy, the Heun-Askey-Wilson algebra \cite[Definition 2.1]{BTVZ} with generators $\tA,{\textsf I}$ is a generalization of the Askey-Wilson algebra. So, it should provide the algebraic framework for a class of special functions beyond the Askey-scheme. Thus, investigating  the spectral problem with respect to $\tA,{\textsf I}$ is an important issue in this direction. Related works in these directions are e.g. \cite{Ta17,KST18,Ta19}. \vspace{1mm}

$\bullet$ The q-Onsager algebra  with generators $\tW_0,\tW_1$ introduced in \cite{Ter03} (see also \cite{Bas2})  is known to be
 a homomorphic pre-image of  the Askey-Wilson algebra with $\tW_0 \rightarrow \tA $, $\tW_1 \rightarrow \tA^* $. From that point of view, the theory of tridiagonal pairs developed by Terwilliger {\it et al.}  generalizes the theory of Leonard pairs \cite{Ter03}. In \cite{BK,BK2}, elements denoted $\{{\textsf I}_{2k+1},k\in{\mathbb Z}_+\}$ that generate a commutative subalgebra of the q-Onsager algebra have been constructed.  In general, they read as polynomials in $\tW_0,\tW_1$ of maximal degree $2k+2$ \cite{BB}. In particular, the image of the element ${\textsf I}_1$ in the Askey-Wilson algebra is the Heun-Askey-Wilson element (\ref{I}). Thus, the analysis presented here can be viewed as a warm up for the diagonalization of the mutually commuting elements  $\{{\textsf I}_{2k+1},k\in{\mathbb Z}_+\}$ within the algebraic Bethe ansatz. In quantum integrable systems, it is important to stress that the elements $\{{\textsf I}_{2k+1},k\in{\mathbb Z}_+\}$ are the basic building quantities for mutually commuting quantities, for instance the Hamiltonian of the open XXZ spin chain with  generic integrable boundary conditions. 
\vspace{1mm}

$\bullet$ The Askey-Wilson algebra admits an embedding into $U_q(sl_2)\otimes U_q(sl_2) \otimes U_q(sl_2)$ \cite{GZ93b} (see also \cite{Huang}), where the generators map as $ \tA \rightarrow \Delta(C) \otimes 1$,  $ \tA^* \rightarrow 1\otimes \Delta(C)$ with $C,\Delta$, respectively the Casimir element and coproduct of  $U_q(sl_2)$. In the recent literature \cite{Post,DDV}, a generalization of the Askey-Wilson algebra indexed by $N$ has been introduced. For $N=3$, it reduces to the Askey-Wilson algebra. Importantly, an embedding of this algebra in terms of `intermediate' Casimir elements of $(U_q(sl_2))^{\otimes N}$ has been given. Thus, the spectral problem solved in the present paper can be viewed as a toy model for studying generalizations of (\ref{I}) to $N>3$.  In the context of quantum integrable spin chains, this approach differs from the usual one based on Sklyanin's construction \cite{Skly88}. Indeed,  for $N=3$ considering an irreducible finite dimensional representation for each $U_q(sl_2)$  labeled by $j_1,j_2,j_3$,  for $j_1=j_2=j_3=j$, the Heun-Askey-Wilson operator associated with (\ref{I}) gives the Hamiltonian of a three-sites spin$-j$ chain in a magnetic field with inhomogeneous couplings, three-body  and boundary interactions. Thus, generalizations of the Askey-Wilson algebra should naturally generate Heisenberg spin chains with possible  inhomogeneous couplings and long range interactions.  For instance, see \cite{Ku}.\vspace{1mm}

$\bullet$ In the context of signal treatment, the optimal reconstruction of a signal from limited observational data is a central problem. In particular,  the diagonalization of Heun type operators play a crucial role in the so-called band-time limiting problem \cite{P87, G94,GVZ17}. For a recent review see \cite{bvz19} and references therein. For the q-deformed case, it is straigthforward to derive the q-analogs of the conditions \cite[eqs. (5.31), (5.32), (5.33)]{GVZ17} such that $\bar\pi(\textsf{I}(\kappa,\kappa^*,\kappa_+,\kappa_-)$ commutes with projectors on certain eigenspaces of $\bar\pi(\tA),\bar\pi(\tA^*)$. Thus, the results here presented apply in this context.
\vspace{1mm}

Other perspectives may be also briefly mentioned. For instance, for the realization (\ref{AWop1a}), (\ref{AWop2a}), the Heun-Askey-Wilson operator gives a fourth-order q-difference operator. Possible connections with q-Krall polynomials \cite{VYZ} could be explored. Another interesting direction is to study the relation with the separation of variables framework, see e.g. \cite{sov} and references therein. In particular, the eigenvectors  $|\theta^\diamond_{M}\rangle$ introduced in subsection \ref{ss31} diagonalize the off-diagonal entry $\bar\pi({\cal B}(u))$ given by (\ref{monoB}). In addition, it should be possible to study the spectrum of the Heun-Askey-Wilson q-difference operator using an {\it homogeneous} Baxter T-Q relation with non-polynomial solution \cite{LP14}. Also, for  q a root of unity and cyclic representations  of the Askey-Wilson algebra, the results here presented may be extended in light of \cite{BGV,Huang3}.
\vspace{1mm}

\vspace{0.7cm}

\noindent{\bf Acknowledgments:}  We thank S. Belliard, O. Brodier, N. Cramp\'e and P. Terwilliger for discussions and comments. We thank H-W. Huang for explanations on \cite{Huang}. P.B. thanks H. Saleur for pointing out reference \cite{Fen}, C. Fewster for pointing out reference \cite{Madr2} and comments on rigged Hilbert spaces, R. de la Madrid for comments on the construction of an $`|x\rangle '$ basis for the Askey-scheme using rigged Hilbert space, G. Roux for comments on potential applications of the Hamiltonians of Section \ref{App}, and L. Vinet and A. Zhedanov for joint collaborations around the Heun operator and comments on applications of the results here presented to the time-band limiting problem in signal processing. R.P. thanks the Institut Denis Poisson for hospitality. R.P. is supported by
the S\~ao Paulo Research Foundation (FAPESP) and by the Coordination for the Improvement of Higher Education Personnel (CAPES), grants \# 2017/02987-8 and \#88881.171877/2018-01.
P.B.  is supported by C.N.R.S. 
\vspace{0.2cm}

\newpage

\begin{appendix}

\section{Dynamical operators in the Askey-Wilson presentation}\label{apA}
The dynamical operators are polynomials of maximum degree two in the elements $\tA,\tA^*$ of the Askey-Wilson algebra.
From (\ref{Ae1})-(\ref{De1}), one finds:
\ben
&& \label{Am}\\
 &&  \mathscr{A}^{+}(u,m)=\frac{b(u^2)}{u (\alpha  q^{2 m}-q^2\beta  )} \Big(
\frac{\chi  q^m }{\rho }[\textsf{A}^*,\textsf{A}]_q-\frac{\alpha  \beta  q^{m+2} }{\chi
   }[\textsf{A},\textsf{A}^*]_q+ (\beta +\alpha  q^{2
   m})q\textsf{A}  -\frac{ (\alpha  q^{2 m} +\beta  q^4 u^4)}{q u^2}\textsf{A}^*
\nonumber\\&&\quad
-\frac{1}{\rho  q^3 u^2 \chi 
   b(q^2) b(u^2)}\Big(     \rho q^{{m}+2}(\alpha\beta q^2\rho -\chi^2)(u^6q^2-u^{-2})  \nonumber\\
&& \qquad \qquad \qquad \qquad \qquad  + (q^2+1)q^{{m}}( \alpha \eta^* \chi q^{m}(q^4-1) -q^2\omega (\alpha\beta q^2\rho -\chi^2)   )                          \nonumber\\
&& \qquad \qquad \qquad \qquad \qquad  - (q^2+1)q^{2}u^4( \beta \eta^* \chi (q^4-1) -q^{m}\omega (\alpha\beta q^2\rho -\chi^2)   ) \nonumber\\         
&& \qquad \qquad \qquad \qquad \qquad  + (q^2-1)u^2 (   \eta\chi( \alpha q^{2{m}} -  \beta q^2) (1+q^2)^2 -\rho q^{{m}+2} (\alpha\beta q^2\rho -\chi^2)   ) \Big)\Big),\nonumber\\  \nonumber
\een
\ben
&&\mathscr{A}^{-}(u,m)=
\frac{b(u^2)}{\alpha  q^2-\beta 
   q^{2 m}}
 \Big(\frac{u \chi  q^{m+2} }{\rho }[\textsf{A}^*,\textsf{A}]_q-\frac{\alpha 
   \beta  u q^m }{\chi }[\textsf{A},\textsf{A}^*]_q+\frac{ (\beta  q^{2 m}+\alpha  q^4
   u^4)}{q u}\textsf{A} -q u (\alpha +\beta  q^{2 m})\textsf{A}^* 
\nonumber\\&&\quad
+
\frac{1}{
\rho  q^3 u^2 \chi  b(q^2) b(u^2)}\Big(\rho  q^{m+2} (q^2 u^7-u^{-1})  (q^2 \chi ^2-\alpha  \beta  \rho )-(q^2-1)
   u^3 ( (q^2+1)^2 \eta^*\chi  (\alpha  q^2-\beta  q^{2 m})
\nonumber\\&&\quad\quad\quad\quad\quad\quad\quad\quad\quad\quad
+\rho  q^{m+2}
   (q^2 \chi ^2-\alpha  \beta  \rho ))-(q^2+1) q^2 u^5 (\omega  q^m
   (\alpha  \beta  \rho -q^2 \chi ^2)+\alpha   (q^4-1)\eta  \chi
   )
\nonumber\\&&\quad\quad\quad\quad\quad\quad\quad\quad\quad\quad
+(q^2+1) u q^m (\beta    (q^4-1) \eta\chi  q^m+\omega  (\alpha  \beta
    q^2 \rho -q^4 \chi ^2))\Big)\Big),\nonumber\\ \nonumber
\een
\ben
&&\mathscr{B}^{+}(u,m)=\frac{\beta  b(u^2)}{\alpha  q^{2 m+2}-\beta }\Big(\frac{\chi  q^m }{\beta  \rho  u}[\textsf{A}^*,\textsf{A}]_q-\frac{\beta  q^{-m} }{u \chi
   }[\textsf{A},\textsf{A}^*]_q+\frac{ (q^2+1)}{q
   u}\textsf{A}-(\frac{1}{q u^3}+q u)\textsf{A}^* 
\label{Bm}\\&&\quad
+\frac{q^{-m-3}}{\beta  \rho 
   \chi  b(q^2)}\big((q^2+1) u^{-1} (\omega  (\chi ^2 q^{2 m+2}-\beta ^2 q^2 \rho )+\beta   (q^4-1) \eta^*\chi  q^m)-q^2 \rho  (q^2 u+u^{-3})  (\beta ^2 \rho -\chi ^2
   q^{2 m})\big)\Big),\nonumber\\ \nonumber
\een
\ben
&&\mathscr{B}^{-}(u,m)=
\frac{\beta  b(u^2) q^{2 m+1}}{\alpha q^{-2}-\beta  q^{2 m}}
\Big(\frac{u \chi  q^{-m-1} }{\beta 
   \rho }[\textsf{A}^*,\textsf{A}]_q-\frac{\beta  u q^{m-1} }{\chi }[\textsf{A},\textsf{A}^*]_q+\frac{ (q^2 u^4+1)}{q^2
   u}\textsf{A}-\frac{ (q^2+1) u}{q^2}\textsf{A}^*
\nonumber\\&&\quad
-\frac{q^{-m-4} }{\beta  \rho  \chi 
   b(q^2)}\big(\rho  q^2 (q^2 u^3+u^{-1})  (\beta ^2 \rho  q^{2 m}-\chi ^2)+(q^2+1) u
   (\omega  (\beta ^2 \rho  q^{2 m+2}-q^2 \chi ^2)+\beta    (q^4-1)\eta \chi 
   q^m)\big)\Big),\nonumber\\ \nonumber
\een
\ben
&&\mathscr{C}^{+}(u,m)=
\frac{\alpha  b(u^2)}{\alpha
   -\beta  q^{2-2 m}}
 \Big(\frac{\alpha  q^m }{u \chi }[\textsf{A},\textsf{A}^*]_q-\frac{\chi  q^{-m} }{\alpha  \rho 
   u}[\textsf{A}^*,\textsf{A}]_q-\frac{ (q^2+1)}{q u}\textsf{A}+
   (\frac{1}{q u^3}+q u)\textsf{A}^*
\label{cm}\\&&
+\frac{q^{-m-3}}{\alpha  \rho  \chi  b(q^2)}\Big(
  \rho  q^2 (q^2 u+u^{-3})   (\alpha ^2 \rho  q^{2 m}-\chi ^2)+
   (q^2+1) u^{-1} (\omega  (\alpha ^2 \rho  q^{2 m+2}-q^2 \chi
   ^2)-\alpha (q^4-1) \eta^*\chi  q^m)
\Big)\Big),\nonumber
\een
\ben
&&\mathscr{C}^{-}(u,m)=
\frac{\alpha q b(u^2)}{\alpha q^2-\beta q^{2m}}
\Big(
\frac{\alpha  u q^{1-m}
   }{\chi }[\textsf{A},\textsf{A}^*]_q-\frac{u \chi  q^{m+1} }{\alpha  \rho }[\textsf{A}^*,\textsf{A}]_q-\frac{ (q^2 u^4+1)}{u}\textsf{A}+ u(q^2+1)  \textsf{A}^*
\nonumber\\
&&\quad+
\frac{q^{-m}}{\alpha\rho\chi (q^4-1) }
\Big(
\rho  q^2 (q^2 u^3+u^{-1})  (\alpha ^2 \rho -\chi ^2 q^{2 m})+(q^2+1) u
   (\omega  (\alpha ^2 q^2 \rho -\chi ^2 q^{2 m+2})+\alpha   (q^4-1) \eta \chi 
   q^m)
\Big)
\Big),\nonumber\\ \nonumber
\een
\ben
&&\mathscr{D}^{+}(u,m)=
\frac{b(u^2) b(q^2 u^2)}{u b(q u^2) (\alpha  q^{2 m}-\beta  q^2)}\Big(\frac{\alpha  \beta  q^{m+2}
   }{\chi }[\textsf{A},\textsf{A}^*]_q-\frac{\chi  q^m }{\rho }[\textsf{A}^*,\textsf{A}]_q- q
   (\beta +\alpha  q^{2 m})\textsf{A} +\frac{ q (\beta +\alpha  u^4 q^{2
   m})}{u^2}\textsf{A}^*
\label{dm}\\
&&\quad
+
\frac{1}{\rho  q^3 u^2  b(q^2)
   b(q^2 u^2)}\Big(
(q^2-1) u^2 \chi ^{-1} (\rho  q^{m+2} (\alpha  \beta  q^2 \rho -\chi ^2)+\eta  \chi 
   (\beta  (q^3+q)^2-\alpha  (q^2+1)^2 q^{2 m}))
\nonumber\\&&\quad\quad\quad\quad\quad\quad\quad
+\rho  \chi ^{-1} q^m (q^6
   u^6-u^{-2})   (\alpha  \beta  q^2 \rho -\chi ^2)+(q^2+1) \chi ^{-1}
   (\omega  q^m (\chi ^2-\alpha  \beta  q^2 \rho )+\beta   (q^4-1)
  \eta^* \chi )
\nonumber\\&&\quad\quad\quad\quad\quad\quad\quad
-(q^2+1) u^4 q^{m+2} \chi ^{-1} (\alpha   (q^4-1)\eta^* \chi
    q^m+q^2 \omega  (\chi ^2-\alpha  \beta  q^2 \rho ))
\Big)\Big),\nonumber\\ \nonumber
\een
\ben
&&\mathscr{D}^{-}(u,m)=
\frac{b(u^2) b(q^2 u^2)}{b(q u^2) (\alpha  q^2-\beta  q^{2
   m})}
\Big(\frac{\alpha  \beta  u q^m }{\chi
   }[\textsf{A},\textsf{A}^*]_q -\frac{u \chi  q^{m+2}
   }{\rho }[\textsf{A}^*,\textsf{A}]_q -\frac{ q (\alpha +\beta  u^4 q^{2 m})}{u}\textsf{A} + q u (\alpha +\beta  q^{2
   m})\textsf{A}^*
\nonumber\\
&&\quad
- \frac{1}{\rho  q^3 u^2 \chi 
   b(q^2) b(q^2 u^2)}
\Big(
(q^2-1) u^3 (\rho  q^{m+2} (q^2 \chi ^2-\alpha  \beta  \rho )+\eta^* \chi 
   (\alpha  (q^3+q)^2-\beta  (q^2+1)^2 q^{2 m}))
\nonumber\\&&\quad\quad\quad\quad\quad\quad\quad
+\rho  q^m (q^6
   u^7-u^{-1})  (q^2 \chi ^2-\alpha  \beta  \rho )+(q^2+1) u (\omega  q^m
   (\alpha  \beta  \rho -q^2 \chi ^2)+\alpha    (q^4-1)\eta \chi
   )
\nonumber\\&&\quad\quad\quad\quad\quad\quad\quad
-(q^2+1) u^5 q^{m+2} (\beta   (q^4-1)  \eta\chi  q^m+\omega  (\alpha 
   \beta  q^2 \rho -q^4 \chi ^2))
\Big)\Big).\nonumber\\ \nonumber
\een

\vspace{2mm}

\section{Coefficients of the commutation relations}\label{Sec:coefcommut}
The coefficients
of the dynamical commutation relations (\ref{comAdBd}), (\ref{comDdBd}), (\ref{comcdBd}) are given by 
\ben
\nonumber&&f(u,v)= \frac{b(qv/u)b(uv)}{b(v/u)b(quv)}\,,\quad h(u,v)= \frac{b(q^2uv)b(qu/v)}{b(quv)b(u/v)},\\
\nonumber&&g(u,v,m)=\frac{\gamma(u/v,m+1)}{\gamma(1,m+1)}\frac{b(q) b\left(v^2\right)}{b\left(q v^2\right) b\left(\frac{u}{v}\right)},
\quad w(u,v,m)=-\frac{\gamma(uv,m)}{\gamma(1,m+1)}\frac{b(q)}{b(q u v)},\\
\nonumber &&k(u,v,m)=\frac{ \gamma(v/u,m+1)}{\gamma(1,m+1)}\frac{b(q) b\left(q^2 u^2\right)}{b\left(q u^2\right) b\left(\frac{v}{u}\right)}, \quad
n(u,v,m)=\frac{\gamma(1/(uv),m+2)}{\gamma(1,m+1)} \frac{b(q) b\left(v^2\right) b\left(q^2 u^2\right)}{b\left(q u^2\right) b\left(q v^2\right)
   b(q u v)}\,\,,\\
\nonumber&&q(u,v,m)=\frac{ \gamma \left(u/v,m\right)b(q) b(u v)}{\gamma (1,m+1) b\left(u/v\right) b(q u v)}\,,\quad
r(u,v,m)=\frac{b(q) b\left(u^2\right) \gamma (1,m) \gamma \left(v/u,m+1\right)}{\gamma (1,m+1)^2 b\left(q
   u^2\right) b\left(v/u\right)}\,,\\
\nonumber&&s(u,v,m)=\frac{b(q)^2 b\left(u^2\right) \gamma \left(v^{-2},m+1\right) \gamma \left(v/u,m+1\right)}{\gamma
   (1,m+1)^2 b\left(q u^2\right) b\left(q v^2\right) b\left(\frac{v}{u}\right)}\,,
\quad
x(u,v,m)=\frac{b(q) b\left(u^2\right) b\left(q u/v\right) \gamma \left(1/(uv),m+1\right)}{\gamma (1,m+1)
   b\left(q u^2\right) b\left(u/v\right) b(q u v)}\,,\\
\nonumber&&y(u,v,m)=-\frac{b(q)^2 \gamma \left(v^{-2},m+1\right) \gamma (u v,m)}{\gamma (1,m+1)^2 b\left(q v^2\right) b(q u v)}\,,
\quad
z(u,v,m)=-\frac{b(q) \gamma (1,m) \gamma (u v,m)}{\gamma (1,m+1)^2 b(q u v)}\,.\nonumber
\een

\vspace{2mm}

\section{Proof of Lemma   \ref{prop:diagonaloffshell}}\label{apD}
In this section, we show Lemma  \ref{prop:diagonaloffshell} for the choice $\epsilon =+$. The proof for $\epsilon=-$ is done along the same line, so we omit the details. As a preliminary, consider (\ref{Bm}) for $\beta=0$. It reads:
\beqa
\mathscr{B}^{+}(u,m) = \frac{b(u^2)}{u\alpha q^{m+2}}\frac{\chi}{(q-q^{-1})} \left( U - \tilde{\tA}^\diamond \right)
\eeqa
where the notation $\tilde{\tA}^\diamond$
\beqa
\tilde{\tA}^\diamond = \frac{i(q-q^{-1})}{\sqrt \rho}\tC \label{tCop}
\eeqa
with (\ref{Cop}) is introduced for convenience. Recall the notation  (\ref{SB}), (\ref{SB2}).  Successively, one has:
\beqa
\mathscr{B}^{+}(u,m_0+4s)  B^{+}(\bar u,m_0,2s) &=& \left( \frac{b(u^2)}{u}\prod_{i=1}^{2s}\frac{b(u_i^2)}{u_i}\right) \frac{\chi^{2s+1}q^{-(m_0+2)(2s+1)}q^{-2s(2s+1)} }{\alpha^{2s+1}(q-q^{-1})^{2s+1}}(U-\tilde{\tA}^\diamond)\prod_{i=1}^{2s}(U_i-\tilde{\tA}^\diamond) \ ,\nonumber\\
 B^{+}(\bar u,m_0,2s) &=& \left(\prod_{i=1}^{2s}\frac{b(u_i^2)}{u_i}\right) \frac{\chi^{2s}q^{-(m_0+2)2s}q^{-2s(2s-1)} }{\alpha^{2s}(q-q^{-1})^{2s}}\prod_{i=1}^{2s}(U_i-\tilde{\tA}^\diamond) \ ,\nonumber\\
 B^{+}(\{u,\bar u_i\},m_0,2s) &=& \left( \frac{b(u^2)}{u}\prod_{j=1,j\neq i}^{2s}\frac{b(u_j^2)}{u_j}\right) \frac{\chi^{2s}q^{-(m_0+2)2s}q^{-2s(2s-1)} }{\alpha^{2s}(q-q^{-1})^{2s}}(U-\tilde{\tA}^\diamond)\prod_{j=1,j\neq i}^{2s}(U_j-\tilde{\tA}^\diamond) \ .\nonumber
\eeqa
Also, one has:
\beqa
\prod_{k=0}^{2s} b(q^{1/2+k-s} uv)b(q^{1/2+k-s} u v^{-1})&=& (q+q^{-1})^{2s+1} \prod_{k=0}^{2s} \left(U - X_k\right) \quad \mbox{with} \quad X_k = \frac{v^{2}q^{2k-2s} +v^{-2}q^{-2k+2s} }{q+q^{-1}}  \ ,\nonumber\\
b(uu_i^{-1})b(q^{-1}u^{-1}u_i^{-1})&=& -(q+q^{-1})(U-U_i) \ .\nonumber
\eeqa
Consider now the combination:
\beqa
(*)&=& \mathscr{B}^{+}(u,m_0+4s)  B^{+}(\bar u,m_0,2s) -\delta_d \frac{b(u^2)}{u}\frac{\prod_{k=0}^{2s} b(q^{1/2+k-s} uv)b(q^{1/2+k-s} u v^{-1})}{\prod_{i=1}^{2s}b(uu_i^{-1})b(q^{-1}u^{-1}u_i^{-1})} B^{\epsilon}(\bar u,m_0,2s) \nonumber\\
&& \qquad \qquad \qquad \qquad +\delta_d
\sum_{i=1}^{2s}
\frac{
u_i^{-1}b(u_i^2)\prod_{k=0}^{2s}b(q^{1/2+k-s}vu_i)b(q^{1/2+k-s}v^{-1}u_i)}
{b(uu_i^{-1}) b(q^{-1} u^{-1} u_i^{-1})\prod_{j=1,j\neq i}^{2s}b(u_iu_j^{-1})b(q^{-1}u_i^{-1}u_j^{-1})} B^{+}(\{u,\bar u_i\},m_0,2s)\ .\nonumber
\eeqa
In this expression, the gauge parameter $\alpha$ is given by (\ref{ab}) with (\ref{par}),  and we use (\ref{deltad}) for $\epsilon=+$. 
Using the previous expressions, after straightforward simplifications of $(*)$ and factorizing out the common overall factor, one gets the new combination:
 \beqa
(**)= (U-\tilde{\tA}^\diamond)\prod_{i=1}^{2s}(U_i-\tilde{\tA}^\diamond)  - \prod_{k=0}^{2s}(U-X_k) \frac{\prod_{i=1}^{2s}(U_i-\tilde{\tA}^\diamond)}{\prod_{i=1}^{2s}(U-U_i) } 
+ \sum_{i=1}^{2s} \frac{\prod_{k=0}^{2s}(U_i-X_k) }{(U-U_i)}
 \frac{\prod_{j=1,j\neq  i}^{2s}(U_j-\tilde{\tA}^\diamond)(U-\tilde{\tA}^\diamond) }{\prod_{j=1,j\neq i}^{2s}(U_i-U_j) }\ .\nonumber
\eeqa
The combination $(**)$ is a polynomial in  $\tilde{\tA}^\diamond$ of degree $2s+1$ with coefficients  that are meromorphic functions of $U,U_i$. It has poles located at $U=U_i$ and $U_i=U_j$, and the residues at these points vanishes. Setting $U=\tilde{\tA}^\diamond $, one finds:
\beqa
(**)=  (-1)^{2s+1} \prod_{k=0}^{2s}(\tilde{\tA}^\diamond - X_k)\ .
\eeqa
We now consider $(**)$ on the  finite dimensional representation $\bar{\cal V}$.
Recall (\ref{tCop}) with (\ref{Cop}). Using (\ref{sC}) with (\ref{par}),  observe that $X_k$ is the spectrum of $\tilde{\tA}^\diamond$. Thus, $(**)$ is proportional to the characteristic polynomial of $\tC$ which is vanishing on
$\bar{\cal V}$. This concludes the proof of Lemma  \ref{prop:diagonaloffshell} for $\epsilon=+$.\vspace{2mm}

\section{Proof of Proposition \ref{propBAU}}\label{prP}
To prepare the proof of Proposition \ref{propBAU}, we first derive some intermediate results.
Recall the notation $\bar U_i = \{ U_1,...,U_{i-1},U_{i+1},...,U_M\}$, (\ref{sBr}), (\ref{b}) and (\ref{esymdef}).
\begin{lem}\label{lemC1}
\beqa
&& \prod_{j=1,j\neq i}^{M}  b(qu_j/u_i)b(u_iu_j) =\nonumber\\
&& (q+q^{-1})^{M-1} \sum_{k=0}^{M-1} (-2)^{k+1-M}  \textsf{e}_{k}(\bar U_i) \sum_{l=0}^{M-1-k} \bin {M-1-k} {l} (q^2+q^{-2})^{M-1-k-l}(q-q^{-1})^l (-b(qu_i^2))^l U_i^{M-1-k-l}\ .\nonumber
\eeqa
\end{lem}
\begin{proof} Observe that:
\beqa
\prod_{j=1,j\neq i}^{M}  b(qu_j/u_i)b(u_iu_j) &=& (-1)^{M-1}(q+q^{-1})^{M-1} \prod_{j=1,j\neq i}^{M} \left( \frac{q^{-1}u_i^2 + qu_i^{-2}}{q+q^{-1}}- U_j\right)\ .\nonumber
\eeqa
Expand this expression in terms of the elementary symmetric polynomials in the variables $U_j,j\neq i$. Insert: 
\beqa
 \frac{q^{-1}u_i^2 + qu_i^{-2}}{q+q^{-1}} &=& \frac{1}{2}\left( (q^2+q^{-2}) U_i - (q-q^{-1})b(qu_i^2)\right).\nonumber
\eeqa
Then, expand in $U_i$ and $b(qu_i^2)$ to get the final result.
\end{proof}
\begin{rem}\label{Rem1}
\beqa
\prod_{j=1,j\neq i}^{M}  b(qu_i/u_j)b(q^2u_iu_j) &=& (-1)^{M-1} \prod_{j=1,j\neq i}^{M}  b(qu_j/u_i)b(u_iu_j)|_{u_i \rightarrow q^{-1}u_i^{-1}}\ .\nonumber
\eeqa
\end{rem}

For $N$ any even integer, note that one has:
\beqa
(b(qu_i^2))^N= \left( (q+q^{-1})^2U_i^2 - 4\right)^{N/2} := g^{(N)}_0(U_i)\ .\label{bexp}
\eeqa

Let us introduce   $\Delta_{d}(u)$ and $\Delta_{g}(u)$ respectively given by (\ref{Deltad}) and (\ref{Deltag}), and define $\Delta_{sp}(u)=1$. 
Define
\beqa
&&\quad F_\epsilon(u_i)= e^{-(\mu-\mu')/2+\epsilon(\mu-\mu')/2}q^{-(2s+1)}(q^{2s+1}u_iv^{-1} - u_i^{-1}v)(q^{2s+1} u_iv - u_i^{-1}v^{-1})
\non\\&&\qq \qq\qq \times(e^{(\mu-\mu')(1-\epsilon)/2}u_i+e^{(\mu+\mu')(1+\epsilon)/2}u_i^{-1})
(e^{(\mu+\mu')(1-\epsilon)/2}u_i+e^{(\mu'-\mu)(1+\epsilon)/2}u_i^{-1}).\nonumber
\eeqa

\begin{lem}\label{lemC3} There exist   polynomials $\{g_{a,\epsilon}^{(p)}(U_i)|p\in\{even,odd\},a\in\{sp,d,g\},\epsilon=\pm\}$  of the form
\beqa
g_{a,\epsilon}^{(p)}(U_i) = \sum_{k=0}^3  g_{a,\epsilon,[k]}^{(p)}  U_i^{k}\ ,\nonumber
\eeqa
where  $g_{a,\epsilon,[k]}^{(p)}$ are scalars such that
\beqa
-\frac{\Delta_a(u_i)}{b(qu_i^2)}F_\epsilon(u_i) + \frac{\Delta_a(q^{-1}u_i^{-1})}{b(qu_i^2)}F_\epsilon(q^{-1}u_i^{-1}) &=&   g_{a,\epsilon}^{(even)}(U_i) \ ,\label{geven}\\
\Delta_a(u_i)F_\epsilon(u_i) + \Delta_a(q^{-1}u_i^{-1})F_\epsilon(q^{-1}u_i^{-1}) &=&  g_{a,\epsilon}^{(odd)}(U_i) \ .\label{godd}
\eeqa
\end{lem}
\begin{proof} Routine.
\end{proof}

As a consequence of Lemma \ref{lemC1}, eq. (\ref{bexp}), Lemma \ref{lemC3} and Remark \ref{Rem1}, it follows: 
\begin{cor}\label{cor1}
\beqa
&& \frac{(-1)^M\Delta_a(u_i)}{b(qu_i^2)}
\prod_{j=1,j\neq i}^{M}   b(qu_j/u_i)b(u_iu_j)F_\epsilon(u_i)
+  \frac{\Delta_a(q^{-1}u_i^{-1})}{b(qu_i^2)}
\prod_{j=1,j\neq i}^{M}   b(qu_i/u_j) b(q^2u_iu_j)F_\epsilon(q^{-1}u_i^{-1}) \nonumber \\
&=& \sum_{k=0}^{M-1} \frac{(-1)^k(q+q^{-1})^{M-1}}{2^{M-1-k}} \textsf{e}_{k}(\bar U_i) \left(\sum_{l=0}^{M-1-k}\bin {M-1-k} {l}
\frac{(q-q^{-1})^{l}U_i^{M-1-k-l}}{(q^2+q^{-2})^{1+k+l-M}} g_0^{2[\frac{l}{2}]}(U_i) g_{a,\epsilon}^{( p[l])}(U_i) \right) \ ,\nonumber 
\eeqa 
where 
\beqa
p[l]=even\  (resp. \ odd) \quad \mbox{for}\quad  l \ \ \mbox{even (resp. \ odd)}.\label{defpl}
\eeqa
\end{cor}

Also, we will need:
\begin{lem}\label{lemH} For any integer or half-integer $s$, define: 
\beqa
H(U_i) = (q+q^{-1})^{2s+1} \sum_{k=0}^{2s+1} (-1)^k   \textsf{e}_{k}(X_0,X_1,...,X_{2s}) U_i^{2s+1-k} \quad  \mbox{with} \quad 
X_k = \frac{q^{2k-2s}v^2 +q^{-2k+2s}v^{-2} }{q+q^{-1}}\ .\nonumber
\eeqa
One has:
\beqa
 \prod_{k=0}^{2s}b(q^{1/2+k-s}vu)b(q^{1/2+k-s}v^{-1}u) = H(U)\ . \label{bbp}
\eeqa
\end{lem}
\begin{proof} Observe that the product (\ref{bbp}) is left invariant under the transformation (\ref{t1}). Thus, it is a polynomial in $U_i$. In terms of the elementary symmetric polynomials (\ref{esymdef}), one routinely gets (\ref{bbp}).
\end{proof}

We are now ready to show Proposition \ref{propBAU}. First,  we extract from $E_a(u_i,\bar u_i)$ the part that is left invariant under the transformations (\ref{t1}), (\ref{t2}). Introduce the rational functions in $u_i,i=1,...,M$:
\beqa
  p_a^M(u_i,\bar u_i) =  \frac{(-1)^M\Delta_a(u_i)}{b(qu_i^2)}
\prod_{j=1,j\neq i}^{M}   b(qu_j/u_i)b(u_iu_j)F_\epsilon(u_i)
+  \frac{\Delta_a(q^{-1}u_i^{-1})}{b(qu_i^2)}
\prod_{j=1,j\neq i}^{M}   b(qu_i/u_j) b(q^2u_iu_j)F_\epsilon(q^{-1}u_i^{-1})
\nonumber \\ 
 \quad \qquad \qquad \qquad
+\bar\Delta_a \prod_{k=0}^{2s}b(q^{1/2+k-s}vu_i)b(q^{1/2+k-s}v^{-1}u_i)\ .\nonumber
\eeqa
It is easy to show that:
\beqa
{E}_a(u_i,\bar u_i)=\frac{u_i^{-\epsilon}b(u_i^2)q^{(\nu+\nu')/2}}{2\prod_{j\neq i}^Mb(u_i/u_j)b(qu_iu_j)}p_a^M(u_i,\bar u_i)\ .\label{BAUcheck}
\eeqa
Then, by Lemma \ref{lemC1}, Corollary \ref{cor1} and Lemma \ref{lemH}, one finds $p_a^M(u_i,\bar u_i) = P_a^M(U_i,\bar U_i)$ with (\ref{Polya}).
\vspace{2mm}

\section{second-order q-difference operators and the Askey-Wilson algebra}\label{appqdiff}
Below, we display two examples of second-order q-difference operator realizations of the Askey-Wilson algebra.
The realization $\pi_1$ is taken from \cite{Ter03}, whereas $\pi_2$ is taken from \cite{BVZ16}. Using the invertible non-linear transformation (\ref{dirA}) that relates different Askey-Wilson algebras, in particular  it is shown that the realization $\pi$ in (\ref{AWop1a}), (\ref{AWop2a}) follows from $\pi_1$. In each case, the structure constants are given.\vspace{1mm}
 
Adapting the notations of  \cite{Ter03} for our purpose, let us consider the following linear transformation denoted $\pi_1$: AW $\mapsto$ $\mathbb{C}[z,z^{-1}]$ such that:
\beqa
\pi_1(\bar\tA) &=&  \frac{1}{2}q^{(\nu+\nu')/2}( z+z^{-1}) \ ,\label{AWo1}\\
\pi_1(\bar\tA^*)  &=&      \phi(z)(T_+ -1)  + \phi(z^{-1}) (T_- -1) + \frac{1}{2}q^{(\nu+\nu')/2} (e^{\mu'}q^{-2s} + e^{-\mu'}q^{2s})  \label{AWo2}
\eeqa
where  $T_\pm$ is such that $T_\pm(f(z))=f(q^{\pm 2}z)$ and (\ref{phidef}).
Inserting (\ref{AWo1}), (\ref{AWo2}) in (\ref{aw1}), (\ref{aw2}), the structure constants of the Askey-Wilson algebra are given by:
\beqa
\rho&=&-\frac{1}{4}q^{\nu+\nu'}(q^2-q^{-2})^2\  ,\label{newsc1}\\
 \bar\omega&=&   \frac{1}{2}(q-q^{-1})^2 q^{{\nu+\nu'}}\left(  \cosh(\mu)(q^{2s+1}+q^{-2s-1})  - \cosh(\mu')(v^2+v^{-2})   \right)  , \nonumber\\ 
\bar \eta&=&  \frac{1}{4}q^{\frac{3}{2}({\nu+\nu'})}\frac{(q^2-q^{-2})^2}{(q+q^{-1})}\left(   \cosh(\mu')(q^{2s+1}+q^{-2s-1})    - \cosh(\mu)(v^2+v^{-2})  \right)  , \nonumber \\ 
\bar \eta^*&=&  \frac{1}{8}q^{\frac{3}{2}({\nu+\nu'})}\frac{(q^2-q^{-2})^2}{(q+q^{-1})}\left(  (q^{2s+1}+q^{-2s-1})(v^2+v^{-2}) - 4\cosh(\mu)\cosh(\mu')  \right)   \ . \nonumber
\eeqa

Let $\pi_2$ : AW $\mapsto$ $\mathbb{C}[z,z^{-1}]$ such that

\beqa
\pi_2(\tA) &=&
-\frac{1}{2}(q^{\nu+s+\frac{1}{2}}v^{-1}z+q^{\nu'-s-\frac{1}{2}}vz^{-1})
\label{AWop1b}\\
&&+
\frac{q^{-2s}}{2}
\big(
2q^{(\nu+\nu')/2}\cosh(\mu)+q^{\nu-s+\frac{1}{2}}v^{-1}z+q^{\nu'+s-\frac{1}{2}}vz^{-1}
\big)T_+
 \ ,\nonumber\\
\pi_2(\tA^*)  &=& 
\frac{1}{2}(q^{\nu-s-\frac{1}{2}}vz+q^{\nu'+s+\frac{1}{2}}v^{-1}z^{-1})
\label{AWop2b}\\
&&+
\frac{q^{2s}}{2}
\big(
2q^{(\nu+\nu')/2}\cosh(\mu')-q^{\nu+s-\frac{1}{2}}vz-q^{\nu'-s+\frac{1}{2}}v^{-1}z^{-1}
\big)T_-
 \ . \nonumber
\eeqa
The structure constants in this case are given by (\ref{sc1})-(\ref{sc4}). This realization of the Askey-Wilson algebra has been used  in \cite{BVZ16}, based on \cite{WZ}.\vspace{1mm}

Using the invertible non-linear transformation (\ref{dirA}) that relates two different Askey-Wilson algebras, other realizations can be derived. Let $\tA,\tA^*$ satisfy the Askey-Wilson with structure constants  (\ref{sc1})-(\ref{sc4}). Then, using $\sqrt{-q^{\nu +\nu'-4}\left(q^4-1\right)^2 }\rightarrow i \left(q^4-1\right) q^{\frac{1}{2} (\nu +\nu'-4)}$ one finds that  $\bar\tA,\bar\tA^*$ satisfy (\ref{aw1map}), (\ref{aw2map}) with the structure constants $\{\rho,\bar\omega, \bar\eta,\bar\eta^*\}$ given by (\ref{newsc1}). In particular, one finds that $\pi_1(\tA),\pi_1(\tA^*)$ and $\pi(\tA),\pi(\tA^*)$ given by  (\ref{AWop1a}), (\ref{AWop2a}) are related through the transformation (\ref{dirA}).  For completeness, in addition to  (\ref{AWop1a}), (\ref{AWop2a}), (\ref{qc1}), one finds:
\ben
&&\label{qdifAAs}\\ 
&&\pi( [\tA,\tA^*]_q )=
b(q^2)\phi(z)\phi(q^2z)q^{-2}z^{-1}T_+^2 +
b(q^2)\phi(z^{-1})\phi(q^2z^{-1})q^{-2}zT_-^2
\non\\
&&+
\Big(
\frac{b(q) q^{(\nu+\nu')/2}}{2}  (2 \cosh (\mu ) q^{2 s}-e^{-\mu'} q^{-1}
   (v^2+v^{-2}))
+\frac{b(q)q^{(\nu+\nu')/2}}{2}  (e^{-\mu'} q^{2 s-2}
   (q+q^{-1})+e^{\mu'} q^{-2 s-1}+e^{-\mu'} q^{2 s-1})z^{-1}
\non\\
&&\qq
-b(q^2) q^{-2} z^{-1} \phi (q^2z)
-b(q) q^{-1} z^{-1} \phi (z)
- b(q z)q^{-1} \phi (z^{-1})
+b(z) \phi (q^{-2} z^{-1})
\Big)
 \phi(z)T_+
\non\\
&&+
\Big(
\frac{b(q) q^{(\nu+\nu')/2}}{2}  (2 \cosh (\mu ) q^{2 s}-e^{-\mu'} q^{-1}
   (v^2+v^{-2}))
+\frac{b(q) q^{(\nu+\nu')/2}}{2}  (e^{-\mu'} q^{2 s-2}
   (q+q^{-1})+e^{\mu'} q^{-2 s-1}+e^{-\mu'} q^{2 s-1})z
\non\\
&&\qq
- b(q^2) q^{-2} z\phi (q^2 z^{-1})
-b(q) q^{-1}z \phi(z^{-1})
-b(q z^{-1})q^{-1} \phi (z) 
+b(z^{-1}) \phi (q^{-2}z )
\Big)
 \phi(z^{-1})T_-
\non\\
&&+
\Big(
\frac{b(q)q^{\nu +\nu'}}{4}
(e^{\mu'} q^{-2 s}+e^{-\mu'} q^{2 s}) 
(
2 \cosh (\mu ) q^{2
   s}+e^{-\mu'} q^{2 s-1} (z+z^{-1})-e^{-\mu'} q^{-1} (v^2+v^{-2})
)
\non\\&&
\qq
-
\frac{b(q)q^{(\nu+\nu')/2} }{2} 
((2 \cosh (\mu ) q^{2 s}-e^{-\mu'} q^{-1} (v^2+v^{-2}))+
    (e^{\mu'} q^{-2 s-1}+2 e^{-\mu'} q^{2 s-1})z^{-1}+e^{-\mu'}   q^{2 s-1}z)\phi (z)
\non\\&&
\qq
-
\frac{b(q)q^{(\nu+\nu')/2}}{2} 
( (2 \cosh (\mu ) q^{2 s}-e^{-\mu'} q^{-1}
   (v^2+v^{-2}))+(e^{\mu'} q^{-2 s-1}+2 e^{-\mu'}
   q^{2 s-1})z+e^{-\mu'} q^{2 s-1} z^{-1} )\phi (z^{-1})
\non\\&&
\qq
+b(q) q^{-1} z^{-1} \phi (z)^2
+ b(q) q^{-1} z\phi (z^{-1})^2
+b(q) q^{-1} (z+z^{-1})\phi (z) \phi (z^{-1})
\non\\&&
\qq
+b(z) (\phi (z^{-1}) \phi (q^{-2}z )-\phi (z) \phi (q^{-2}
   z^{-1}))\Big)\ .
\non
\een

\vspace{1mm}

\section{Dynamical operators and the q-difference realization}\label{apF}
Fixing the gauge parameters $\beta=0$ and $\alpha$ by (\ref{ab}), using (\ref{AWop1a}), (\ref{AWop2a}), (\ref{qc1}) and (\ref{qdifAAs}), we obtain the following expressions for the dynamical operators (\ref{Am}), (\ref{cm}) and (\ref{dm}) :
\begingroup
\allowdisplaybreaks

\ben
&&\pi(\mathscr{A}^{+}(u,m_0))
=
-b(u^2) q^{-1} u^{-3} z^{-1} \phi (z) (z-q u^2)T_+
-b(u^2) q^{-1} u^{-3} z \phi (z^{-1}) (z^{-1}-q u^2)T_-
\non\\
&&\qq
+\frac{q^{(\nu+\nu')/2}(q+q^{-1})}{b(qz)b(qz^{-1})}(U-\frac{z+z^{-1}}{q+q^{-1}})
\Big(
-\frac{1}{2} e^{-\mu'} b(u^2) q^{2 s} u^{-1} (z^2+z^{-2})
\non\\
&&\qq\qq
+
 ((q^{-2
   s-1}+q^{2 s+1}) (\cosh (\mu ) u^{-3}+\cosh (\mu') u^{-1})-u^{-3} (v^{-2}+v^2) (u^2 \cosh
   (\mu )+\cosh (\mu')))(z+z^{-1})
\non\\
&&\qq\qq
-(q+q^{-1}) (v^{-2}+v^2) (\cosh (\mu ) u^{-3}+\cosh (\mu')
   u^{-1})+\cosh (\mu ) (q^{-2 s-1}+q^{2 s+1}) (q+q^{-1}) u^{-1}
\non\\
&&\qq\qq
+u^{-3} (\frac{1}{2} e^{-\mu'}
   ((q^{-2 s}-q^{2 s}) q^{-2}+q^{-2 s}+q^{2 s})+\frac{1}{2} e^{\mu'} (q^{-2 s-2}+q^{-2 s}+q^{2
   s}+q^{2 s+2}))
\non\\
&&\qq\qq
+\frac{1}{2} e^{-\mu'}  q^{2 s} (q^2+q^{-2})
u\Big)\ ,
\non
\een
\ben
&&\pi(\mathscr{D}^{+}(u,m_0))
=
-\frac{b(q^2u^2)b(u^2)}{b(qu^2)}\phi(z)(u^{-1}z^{-1}-qu)T_+
-\frac{b(q^2u^2)b(u^2)}{b(qu^2)}\phi(z^{-1})(u^{-1}z-qu)T_-
\non\\
&&\qq
+\frac{q^{(\nu+\nu')/2}(q+q^{-1})b(u^2)}{b(qz)b(qz^{-1})b(qu^2)}(U-\frac{z+z^{-1}}{q+q^{-1}})
\Big(
\frac{1}{2}
   e^{-\mu'} b(q^2 u^2)q^{2 s} u^{-1} (z^2+z^{-2}) 
\non\\
&&\qq\qq
+ ((q^{-2 s-1}+q^{2 s+1}) (
    \cosh (\mu )q^2\textup{}+\cosh (\mu') u^{-1})-u^{-1} (v^{-2}+v^2) (\cosh (\mu )+q^2 u^2 \cosh
   (\mu')))(z+z^{-1})
\non\\
&&\qq\qq
-(q+q^{-1}) (v^{-2}+v^2)
   ( \cosh (\mu )q^2 u+\cosh (\mu') u^{-1})
+ \cosh (\mu') (q^{-2
   s-1}+q^{2 s+1}) (q+q^{-1})q^2 u
\non\\
&&\qq\qq
+u^{-1} (\frac{1}{2} e^{-\mu } ((q^{-2 s}-q^{2 s}) q^{-2}+q^{-2 s}+q^{2
   s})
+\frac{1}{2} e^{\mu } (q^{-2 s-2}+q^{-2 s}+q^{2 s}+q^{2 s+2}))
\non\\
&&\qq\qq
-\frac{1}{2} q^{2 s} (q^2+q^{-2})  (e^{-\mu'} b(q^2 u^2)-e^{-\mu })u^{-1}
\Big)\ ,
\non
\een

\newpage 
\ben
&&\pi(\mathscr{C}^{+}(u,m_0))
=\non\\
&&
2e^{\mu'}q^{-2s-2}q^{-(\nu+\nu')/2}u^{-1}b(u^2)z^{-1}\phi(z)\phi(q^2z)T_+^2
+
2e^{\mu'}q^{-2s-2}q^{-(\nu+\nu')/2}u^{-1}b(u^2)z\phi(z^{-1})\phi(q^2z^{-1})T_-^2
\non\\
&&+
\Bigg(
\frac{2 \cosh (\mu ) (e^{-\mu'} q^{4 s}+e^{\mu'})}{q+q^{-1}}
-\frac{(v^2+v^{-2}) q^{-2
   s-1} (e^{-2 \mu'} q^{4 s}+1)}{q+q^{-1}}+q u^2+q^{-1}u^{-2}
\non\\&&\qq
+
\Big(-\frac{2 e^{-\mu'} \phi (z)
   (e^{2 \mu'}+q^{4 s}) q^{-(\nu+\nu')/2-2s-1}}{q+q^{-1}}
\non\\&&\qq
-\frac{2 e^{\mu'} \phi
   (z^{-1}) q^{(\nu+\nu')/2} (z^2 (q^{-2 s}-e^{-2 \mu'} q^{2 s-2})+e^{-2
   \mu'} q^{2 s}-q^{-2 s-2})}{q^2-q^{-2}}
\non\\&&\qq
+\frac{2 e^{\mu'} q^{(\nu+\nu')/2} \phi
   (q^{-2}z^{-1}) (z^2 (q^{-2 s}-e^{-2 \mu'} q^{2 s+2})+e^{-2 \mu'} q^{2 s-2}-q^{-2
   s})}{q^2-q^{-2}}
\non\\&&\qq
-2 e^{\mu'} \phi (q^2 z) q^{\frac{1}{2} (-\nu -\nu')-2 s-2}+\frac{e^{2
   \mu'} q^{-4 s-1}+e^{-2 \mu'} q^{4 s-1}-q+q^{-1}}{q+q^{-1}}
\Big)z^{-1}
+e^{-2 \mu'}  q^{4 s}z
\Bigg)b(u^2)u^{-1}\phi(z)T_+
\non\\
&&+
\Bigg(
\frac{2 \cosh (\mu ) (e^{-\mu'} q^{4 s}+e^{\mu'})}{q+q^{-1}}
-\frac{(v^2+v^{-2}) q^{-2
   s-1} (e^{-2 \mu'} q^{4 s}+1)}{q+q^{-1}}+q u^2+q^{-1}u^{-2}
\non\\&&\qq
+
\Big(-\frac{2 e^{-\mu'} \phi (z^{-1})
   (e^{2 \mu'}+q^{4 s}) q^{-(\nu+\nu')/2-2s-1}}{q+q^{-1}}
\non\\&&\qq
-\frac{2 e^{\mu'} \phi
   (z) q^{(\nu+\nu')/2} (z^{-2} (q^{-2 s}-e^{-2 \mu'} q^{2 s-2})+e^{-2
   \mu'} q^{2 s}-q^{-2 s-2})}{q^2-q^{-2}}
\non\\&&\qq
+\frac{2 e^{\mu'} q^{(\nu+\nu')/2} \phi
   (q^{-2}z) (z^{-2} (q^{-2 s}-e^{-2 \mu'} q^{2 s+2})+e^{-2 \mu'} q^{2 s-2}-q^{-2
   s})}{q^2-q^{-2}}
\non\\&&\qq
-2 e^{\mu'} \phi (q^2 z^{-1}) q^{\frac{1}{2} (-\nu -\nu')-2 s-2}+\frac{e^{2
   \mu'} q^{-4 s-1}+e^{-2 \mu'} q^{4 s-1}-q+q^{-1}}{q+q^{-1}}
\Big)z
+e^{-2 \mu'}  q^{4 s}z^{-1}
\Bigg)b(u^2)u^{-1}\phi(z^{-1})T_-
\non\\
&&+
\Bigg(
-\frac{2 e^{\mu'} (-q^{-2 s}+e^{-2 \mu'} q^{2 s-2}+(q^{-2 s}-e^{-2 \mu'} q^{2 s+2}) z^2)
   \phi (q^{-2}z^{-1}) \phi (z) q^{(\nu+\nu')/2}}{(q^2-q^{-2}) u z}
\non\\
&&\qq
+\frac{2
   e^{\mu'} ((e^{-2 \mu'} q^{2 s+2}-q^{-2 s})z^{-1}+(q^{-2 s}-e^{-2 \mu'} q^{2 s-2}) z)
   \phi (z^{-1}) \phi (q^{-2}z) q^{(\nu+\nu')/2}}{(q^2-q^{-2})
   u}
\non\\
&&\qq
-\frac{e^{-3 \mu'} (e^{2 \mu'} q^2-1) (e^{\mu'}-q^{2 s+1}) (q^{2
   s+1}+e^{\mu'}) (v^2+v^{-2}) q^{\frac{\nu +\nu'}{2}-3}}{2 (q+q^{-1})
   u}
\non\\
&&\qq
-\frac{e^{-2 \mu'} (q^{2 s}-q^{-2 s}) (e^{2 \mu'}-q^{4 s+2}) \cosh (\mu ) q^{\frac{\nu
   +\nu'}{2}-2}}{(q+q^{-1}) u}
+\frac{2 e^{-\mu'} (q^{4 s}+e^{2 \mu'}) z \phi
   (z^{-1})^2 q^{-(\nu+\nu')/2-2s-1}}{(q+q^{-1}) u}
\non\\
&&\qq
+\frac{2 e^{-\mu'}
   (q^{4 s}+e^{2 \mu'}) \phi (z)^2 q^{-(\nu+\nu')/2-2s-1}}{(q+q^{-1}) u
   z}
+\frac{2 e^{-\mu'} (q^{4 s}+e^{2 \mu'}) (z+z^{-1}) \phi (z^{-1}) \phi
   (z) q^{-(\nu+\nu')/2-2s-1}}{(q+q^{-1}) u}
\non\\
&&\qq
+\frac{e^{-3 \mu'} (-e^{2 \mu'}
   (q^4+1) q^{4 s}+q^{8 s+2}+e^{4 \mu'} q^2) (z+z^{-1}) q^{\frac{\nu +\nu'}{2}-2
   s-3}}{2 (q+q^{-1}) u}
\non\\
&&\qq
-\bigg(-\frac{e^{-2 \mu'} (q^{4 s}+e^{2 \mu'})
   (v^2+v^{-2}) q^{-2 s-1}}{(q+q^{-1}) u}+u q+\frac{e^{-\mu -\mu'} (1+e^{2 \mu })
   (q^{4 s}+e^{2 \mu'})}{(q+q^{-1}) u}
\non\\
&&\qq\qq
+\frac{(e^{2 \mu'} q^{-4 s}+2 e^{-2 \mu'}
   q^{4 s}-q^2+1-q^{-2}) z}{(q+q^{-1}) u q}+\frac{e^{-2 \mu'} q^{4
   s}+1}{(q+q^{-1}) u z q}+q^{-1}u^{-3}\bigg) \phi (z^{-1})
\non\\
&&\qq
-
\bigg(-\frac{e^{-2 \mu'}
   (q^{4 s}+e^{2 \mu'}) (v^2+v^{-2}) q^{-2 s-1}}{(q+q^{-1}) u}+u q+\frac{e^{-\mu
   -\mu'} (1+e^{2 \mu }) (q^{4 s}+e^{2 \mu'})}{(q+q^{-1}) u}
\non\\
&&\qq\qq
+\frac{(e^{-2
   \mu'} q^{4 s}+1) z}{(q+q^{-1}) u q}+\frac{e^{2 \mu'} q^{-4 s}+2 e^{-2 \mu'} q^{4
   s}-q^2+1-q^{-2}}{(q+q^{-1}) u z q}+q^{-1}u^{-3}\bigg) \phi (z)
\Bigg)
b(u^2)
\non\ .
\een

\endgroup

\end{appendix}

\end{document}